\documentclass[a4paper,onecolumn,11pt,accepted=2022-11-23]{quantumarticle}
\pdfoutput=1
\usepackage[utf8]{inputenc}
\usepackage[english]{babel}
\usepackage[T1]{fontenc}
\usepackage{amsmath}
\usepackage[numbers]{natbib}
\usepackage{hyperref}

\usepackage[numbers]{natbib}
\usepackage{tikz}
\usepackage{lipsum}

\usepackage{bbm,mathtools}
\usepackage{amsfonts,amssymb,dsfont}
\usepackage{type1cm,graphicx}
\usepackage{color}
\usepackage{graphicx}
\usepackage{algorithm}
\usepackage{algpseudocode}
\usepackage{dcolumn}% Align table columns on decimal point
\usepackage{bm}% bold math
\usepackage{extarrows}

\newcommand{\bk}[1]{\left |#1 \right \rangle}

\newcommand{\ketbra}[2]{|#1\rangle\! \langle #2|}

\newcommand{\nrm}[1]{\| #1 \|_{2}}

\newcommand{\negl}[1]{\text{negl}(#1)}
\newcommand{\henc}[1]{\textnormal{MHE.Enc}_{#1}}
\newcommand{\aenc}[1]{\textnormal{AltMHE.Enc}_{#1}}
\newcommand{\tn}[1]{\textnormal{#1}}
\newcommand{\e}{\vec{e}}
\newcommand{\s}{\vec{s}}
\newcommand{\x}{\vec{x}}

\renewcommand\vec{\mathbf}

\newcommand{\C}{\mathbb{C}}
\newcommand{\N}{\mathbb{N}}
\newcommand{\R}{\mathbb{R}}
\newcommand{\Z}{\mathbb{Z}}
\newcommand{\I}{\mathbb{I}}
\newcommand{\mS}{\mathbb{S}}
\newcommand{\vt}{\vec{t}}
\newcommand{\mi}{\mathrm{i}}

\newcommand{\tV}{\widetilde{V}}
\newcommand{\vh}{\vec{h}}
\newcommand{\vg}{\vec{g}}

 \newtheorem{thm}{Theorem}
  \newtheorem{prop}[thm]{Proposition}
 \newtheorem{cor}[thm]{Corollary}
 \newtheorem{lem}[thm]{Lemma}
 \newtheorem{defn}[thm]{Definition}

\numberwithin{thm}{section}
\numberwithin{equation}{section}

\makeatletter
\newenvironment{breakablealgorithm}
  {% \begin{breakablealgorithm}
   \begin{center}
     \refstepcounter{algorithm}% New algorithm
     \hrule height.8pt depth0pt \kern2pt% \@fs@pre for \@fs@ruled
     \renewcommand{\caption}[2][\relax]{% Make a new \caption
       {\raggedright\textbf{\ALG@name~\thealgorithm} ##2\par}%
       \ifx\relax##1\relax % #1 is \relax
         \addcontentsline{loa}{algorithm}{\protect\numberline{\thealgorithm}##2}%
       \else % #1 is not \relax
         \addcontentsline{loa}{algorithm}{\protect\numberline{\thealgorithm}##1}%
       \fi
       \kern2pt\hrule\kern2pt
     }
  }{% \end{breakablealgorithm}
     \kern2pt\hrule\relax% \@fs@post for \@fs@ruled
   \end{center}
  }

\makeatletter

\newtheorem{remark}{Remark}[section]

\newenvironment{proof}{\noindent\textit{Proof: }}{}

\begin{document}

\title{Quantum Fully Homomorphic Encryption by Integrating Pauli One-time Pad with Quaternions}

\author{Guangsheng Ma}
\email{gsma@amss.ac.cn}
\affiliation{Academy of Mathematics and Systems Science, Chinese Academy of Sciences, Beijing, China}
\affiliation{School of Mathematics and Physics, North China Electric Power University, Beijing, China}
\thanks{supported by China National Key Research and Development Projects 2020YFA0712300, 2018YFA0704705, Chinese Postdoctoral Science Foundation 2020M680716.}
\author{Hongbo Li}
\email{hli@mmrc.iss.ac.cn}
\affiliation{Academy of Mathematics and Systems Science, Chinese Academy of Sciences, Beijing, China}
\affiliation{University of Chinese Academy of Sciences, Beijing, China}
\maketitle

\begin{abstract}
Quantum fully homomorphic encryption (QFHE) allows to evaluate quantum circuits on encrypted data. We present a novel QFHE scheme, which extends Pauli one-time pad encryption by relying on the quaternion representation of SU(2). With the scheme, evaluating 1-qubit gates is more efficient, and evaluating general quantum circuits is polynomially improved in asymptotic complexity.

Technically, a new encrypted multi-bit control technique is proposed, which allows to perform any 1-qubit gate whose parameters are given in the encrypted form. With this technique, we establish a conversion between the new encryption and previous Pauli one-time pad encryption, bridging our QFHE scheme with previous ones. Also, this technique is useful for private quantum circuit evaluation.

The security of the scheme relies on the hardness of the underlying quantum capable FHE scheme, and the latter sets its security on the learning with errors problem and the circular security assumption.
\end{abstract}

\section{Introduction}
Fully homomorphic encryption (FHE) scheme is an encryption scheme that allows any efficiently computable circuit to perform on plaintexts by a third party holding the corresponding ciphertexts only. As the quantum counterpart, quantum FHE (QFHE) allows a client to delegate quantum computation on encrypted plaintexts to a quantum server, in particular when the client outsources the computation to a quantum server and meanwhile hides the plaintext data from the server.

%The earliest work related to the QFHE can be traced back to the studies on the quantum blind computing, which has a similar goal with QFHE that is secure delegation of quantum computation. However, from the point of view of quantum generalizing homomorphic computing, these solutions have drawbacks such as more than one round interaction, requiring the client a certain of quantum computing capability. After Gentry's ground-breaking work on classical FHE, more and more work consider this topic under the context of quantum analogy of FHE, A main obstacle to prevent these method to become a QFHE is the need . The question of how to design a QFHE scheme has not been addressed well until the work of \cite{mahadev2018classical} occurring.

There are two main differences between quantum FHE and classical FHE. First, in QFHE, the plaintexts are quantum states (or qubits), rather than classical bits. Second, in QFHE, the homomorphic operations are quantum gates, rather than arithmetic ones. Since it is possible to simulate arbitrary classical computation in the quantum setting, a QFHE scheme allows to perform any computation task running on a classical FHE, but not vice versa. From this point, QFHE is a more general framework. It has drawn a lot of attention in the last decade, e.g., \cite{brakerski2018quantum,broadbent2015quantum,childs2001secure,dulek2016quantum,mahadev2018classical,ouyang2018quantum,yu2014limitations}.\\
\\
\textbf{Previous Works.} In 2015, Broadbent and Jeffery \cite{broadbent2015quantum} proposed a complete QHE scheme based on quantum Pauli one-time pad encryption. Specifically, They encrypted every single-qubit of a quantum state (plaintext) with a random Pauli gate (called Pauli one-time pad \cite{ambainis2000private}), and then encrypted the two classical bits used to describe the Pauli pad with a classical FHE, and considered homomorphic evaluations of the universal gates \{Clifford gates, $T$-gate\} for quantum computation. They showed that the evaluation of a Clifford gate can be easily done by public operations on the quantum ciphertext and classical encrypted bits of the pad; the latter will make use of the homomorphic property of classical FHE. They also showed two different approaches to evaluating the non-Clifford gate $T$, at the cost of ciphertext size (or depth of decryption circuit) growing with the number of supported $T$-gate, yielding a QHE for circuits with a constant number of $T$-gates. Since then, how to efficiently evaluate the non-Clifford gate became a key issue.

In 2016, Dulek, Schaffner and Speelman \cite{dulek2016quantum} introduced some special quantum gadgets for achieving the evaluation of $T$-gate, where each gadget is not reusable and duplicable due to its quantum property. Their scheme requires the client to generate a number of quantum gadgets proportional to the number of the $T$-gates to be evaluated, allowing to privately and compactly outsource quantum computation at the cost of additional preparation of quantum evaluation key. In comparison with \cite{broadbent2015quantum}, the dependence on the number of non-Clifford gate is transformed from the ciphertext size (or depth of decryption circuit) to the quantum key.

In 2018, Mahadev proposed the first QFHE scheme with a fully classical key generation process, which reduced the requirement for the quantum capability on the client, so that the client can be completely classical. This scheme used the Pauli one-time pad encryption, and made the evaluations of the universal gates \{Clifford gates, Toffoli-gate\}. To evaluate a Toffoli gate, Mahadev proposed a revolutionary technique called controlled-CNOT operation, which allows to implement a controlled-CNOT gate while keeping the control bit private. With a new approach to evaluating non-Clifford gates, the scheme of \cite{mahadev2018classical} satisfies the compactness requirement of fully homomorphic encryption, and meanwhile there is no longer an explicit bound on the number of supported non-Clifford gates.

One particular requirement of Mahadev's encrypted CNOT operation is that the control bit must be encrypted by an FHE scheme of exponential modulus and equipped with a trapdoor. Later in 2018, Brakerski \cite{brakerski2018quantum} improved Mahadev's work by proposing an alternative approach to realize the encrypted CNOT operation, where the underlying FHE was significantly simplified by reducing the exponential modulus to polynomial modulus, and where the requirement of a trapdoor was also removed. Due to the polynomial noise ratio of the underlying FHE, Brakerski's QFHE scheme achieves a higher level of security, which matches the best-known security for classical FHE, up to polynomial factors. Also, Brakerski showed a close connection between the quantum homomorphic evaluation and the circuit privacy of classical FHE.

As pointed out in \cite{brakerski2018quantum}, one of the most promising applications of QFHE in anticipation is private outsourcing of quantum computation. Improving the efficiency of evaluation is a fundamental question in the studies on homomorphic encryption. In this paper, we focus on improving the efficiency of evaluating quantum algorithms (circuits).

Usually, quantum algorithms (gate-level circuits) are designed by using single-qubit gates and controlled gates (CNOT), such as the famous quantum Fourier transform (cf. Figure 1). When evaluating ``1-qubit+CNOT"-style quantum algorithms with existing ``Clifford+non-Clifford''-style QHE schemes, e.g., \cite{broadbent2015quantum,dulek2016quantum,mahadev2018classical,brakerski2018quantum}, it is required to first decompose each evaluated 1-qubit gate into Clifford/non-Clifford gates, followed by evaluating them one by one (each evaluation requires to perform at least 1 quantum gate). Practically, in average cases, tens of thousands of Clifford/non-Clifford gates are required to approximate a 1-qubit gate within a few bits of precision \cite{william2018Efficient}. So, we consider that if it is possible to design a QFHE scheme that allows to more conveniently and efficiently evaluate 1-qubit gates and thus quantum algorithms, particularly measured in terms of the quantum cost.

This inconvenience in evaluating 1-qubit gate is essentially derived from the small pad space of the encryption scheme. One idea for improvement is to enlarge the pad space from Pauli group to the group SU(2), relying on the notion of approximate computation. This notion, useful for classical FHE \cite{cheon2017homomorphic}, has recently been used in the QFHE setting \cite{mahadev2018classical}. We design a new QFHE scheme based on the above idea, where an important issue addressed is evaluating the CNOT gate in the much more complicated pad setting. Interestingly, our work also provides a tool that allows hiding the evaluated 1-qubit gate from the server. It may be useful for private circuit evaluation in the quantum setting \cite{mohassel2013hide,mohassel2014actively,chardouvelis2021rate}.\\
\\
\textbf{Our Contributions.} We design a new QFHE scheme, which is based on a generalized one-time pad encryption method, called the quaternion one-time pad encryption. We call our quantum ciphertext \emph{the quaternion one-time pad encrypted state} (\textbf{QOTP-encrypted state}), in contrast to \emph{the Pauli one-time pad encrypted state} (\textbf{Pauli-encrypted state}) used in \cite{broadbent2015quantum}. Our scheme has several properties as follows:
\begin{enumerate}
\item[$\bullet$]\textbf{Efficiency.} The cost of evaluating single-qubit is completely classical and not expensive compared to previous QFHE schemes.

\vspace{0.2cm}
With previous ``Clifford+non-Clifford'' QFHE schemes, evaluating a general 1-qubit gate within a specific precision $\epsilon$ requires to evaluate a sequence of Clifford and non-Clifford gates of length O($\log^{2} \frac{1}{\epsilon}$) (by the optimal Solovay-Kitaev algorithm\footnote{The original Solovay-Kitaev algorithm can find a sequence of O$(\log^{3.97}(1/\epsilon))$ quantum gates from a chosen finite set of generators of a density subset of $SU(2)$ to approximate any unitary $SU(2)$ in precision $\epsilon$. However, for the specific finite set $\{ \tn{Clifford gate}, \tn{T-gate} \}$, there is a better version of the SK algorithm with approximation factor O$(\log^{2}(1/\epsilon))$; see \cite{dawson2005solovay} for more details. }), which requires to perform at least O($\log^{2} \frac{1}{\epsilon}$) 1-qubit quantum gates; in comparison, using our scheme only requires to classically homomorphically compute a simple degree-2 polynomial function in $O(\log \frac{1}{\epsilon})$-bit numbers, cf. (\ref{21}).

\vspace{0.1cm}

Practically, in the average case, a sequence of Clifford+T gates of length 25575 is required to approximate a general element of SU(2) within 0.0443 trace distance \cite{william2018Efficient}; in contrast, $14$-bit gate key can represent any element of SU(2) within $\frac{1}{2^{12.5}}$ $L^{2}$-distance, cf. Lemma \ref{ts1}, which guarantees the trace distance no more than $\frac{1}{2^5}=0.0315$, cf. (\ref{42}).

\vspace{0.2cm}

\item[$\bullet$]\textbf{Privacy.} Our scheme achieves the private 1-qubit gate evaluation.

\vspace{0.21cm}
In our scheme, the server only needs to know the encryption of the 1-qubit gate to be evaluated. In contrast, previous schemes require each evaluated gate to be applied in an explicit way.

\vspace{0.2cm}

\item[$\bullet$]\textbf{Conversion.} Our scheme is able to switch back and forth with previous QFHE schemes that are based on Pauli one-time pad encryption.

\vspace{0.2cm}

We show that it is possible to transform a QOTP-encrypted state into its Pauli-encrypted form (cf. Proposition \ref{65}), and a Pauli-encrypted state is in natural QOTP-encrypted form (cf. Lemma \ref{pro1}).

\vspace{0.1cm}

Roughly speaking, the overhead of transforming a QOTP-encrypted state to its Pauli-encrypted form in precision $\epsilon$ is only a fraction O($\frac{1}{\log \frac{1}{\epsilon}}$) of that of evaluating a general 1-bit gate in the same precision $\epsilon$ using the previous QFHE scheme of \cite{mahadev2018classical}, cf. `Efficiency Comparison' in Section \ref{secc5}.

\vspace{0.1cm}

\end{enumerate}

%The decomposition way is transformed from ``Clifford+non-Clifford" decomposition to binary decomposition, the improvement in efficiency relies on a simple fact that using $k$-bit entries is sufficient to represent the element of SU(2) in precision O$(\frac{1}{2^k})$, while the number of required Clifford/non-Clifford gate to achieve the same precision is O($k^{3.97}$).Specifically, transforming a QOTP-encrypted state into its Pauli-encrypted form within precision $\epsilon$ requires O($\log \frac{1}{\epsilon}$) uses of Algorithm \ref{20}, while evaluating a general 1-qubit gate within precision $\epsilon$ using previous QFHE scheme of \cite{mahadev2018classical} requires O($\log^{2} \frac{1}{\epsilon}$) uses of encrypted-CNOT operations. Here, the runtime of encrypted-CNOT operation of \cite{mahadev2018classical} is chosen as the basis for comparison, since it is roughly equal to that of Algorithm \ref{20}.

In comparison with the previous ``Clifford+non-Clifford''-style QFHE schemes, our scheme is less costly in evaluating 1-qubit gates, but more costly in evaluating CNOT gates. For evaluating quantum circuits consisting of $p$ percentage CNOT gates and $(1-p)$ percentage 1-qubit gates within the precision $\tn{negl}(\lambda)$, the complexity advantage of our scheme over the previous ones is O$(\frac{  (1-p)\lambda^2}{ p \lambda  })=O(\lambda)$, when constant $p$ is away from both one and zero. Therefore, except for the extreme case where there are overwhelmingly many CNOTs and negligible 1-qubits gates, our scheme is polynomially better asymptotically, cf. Section \ref{secc5}.

 %An efficiency comparison shows that our QFHE scheme could evaluate some quantum algorithms more efficiently in certain cases. Specifically, for the task of evaluating the quantum Fourier transform on a system of $\lambda$-qubits, the quantum cost of completing this task with our QFHE scheme is only a fraction O($\frac{1}{\lambda^2}$) of the previous one, in the case where we assume the worst non-Clifford-count in approximation by Solovay-Kitaev algorithm, cf. Section \ref{secc5}.

Moreover, by the conversion between our new QFHE scheme and previous ``Clifford +non-Clifford" QHE schemes \cite{broadbent2015quantum,dulek2016quantum,mahadev2018classical,brakerski2018quantum}, one can evaluate the quantum circuits in a hybrid way, which may be more efficient than using a single scheme: for parts of circuits mainly consisting of Clifford gates (or easily approximated by Clifford gates), they can be evaluated in the Pauli one-time pad setting; for parts containing single-qubit gates difficult to approximate, they can be evaluated in the QOTP setting.

The second contribution of this work is a new technique called \emph{encrypted conditional rotation} (encrypted-CROT), which allows the server to perform (up to a Pauli mask) any 1-qubit unitary operator whose parameters are given in encrypted form, cf. Theorem \ref{75}. This technique can bring the following benefits:

\begin{enumerate}

\item[$\bullet$] It provides an approach to private 1-qubit gate evaluation for the QHE schemes based on the Pauli one-time pad.

\vspace{0.2cm}

To be more explicit, this technique allows the server to perform any 1-qubit gate whose parameters are given in encrypted form, and the introduced Pauli mask can be merged with the encryption pad.

\vspace{0.2cm}

\item[$\bullet$] It can be used in the QHE scheme of \cite{broadbent2015quantum} towards constructing a ``Clifford +T''-style QFHE scheme, cf. Remark \ref{cw1}.

\vspace{0.2cm}

It is providing a meaningful alternative to ``Clifford+Toffoli''-style QFHE of \cite{mahadev2018classical}, because although any non-Clifford gate, together with Clifford group, is universal for quantum computation, the efficiency of approximating a particular quantum gate with different non-Clifford gates is different. In practice, the 1-qubit-level T-gate is a more popular choice than the 3-qubit-level Toffoli gate, as the representative element of non-Clifford gates \cite{dawson2005solovay,kliuchnikov2015practical,william2018Efficient}.

\vspace{0.2cm}

\item[$\bullet$] It allows to transform a QOTP-encrypted state into its Pauli-encrypted form.
\end{enumerate}

To the best of our knowledge, this work enriches the family of QFHE schemes by providing the first one of not ``Clifford +non-Clifford''-style. Due to the absence of the Clifford gate decomposition, the scheme avoids some difficulties in its (practical) implementation, but possibly loses some potential advantages in error-correction or fault-tolerant. With the conversion between these QFHE schemes, it is possible to exploit their respective strengths, and provide diverse options for evaluating distinct quantum circuits. This work also provides a useful tool for private circuit evaluations in the quantum setting.

\subsection{Technical Overview.}
Our basic idea to improve the efficiency is to avoid decomposing 1-qubit gate into numerous Clifford+non-Clifford gates during the evaluation process. This idea is hard to realize in previous Pauli one-time pad setting. We show why it is hard. In the QHE scheme based on Pauli one-time pad, traced back to \cite{broadbent2015quantum}, a 1-qubit state (plaintext) is encrypted in form $X^{a}Z^{b}\bk{\psi}$, where $X,Z$ are Pauli matrices, and the Pauli keys $a,b\in\{0,1\}$ are also encrypted by using a classical FHE. Any Clifford gate can be easily evaluated in this setting.

Now, we use $U(\alpha,\beta,\gamma)$ to denote a 1-qubit gate $U$ in Euler angle representation, i.e., $\alpha,\beta,\gamma\in[0,1)$, known as (scaled) Euler angles,
\begin{align}\label{45}
U(\alpha,\beta,\gamma)=R_{\alpha}T_{\beta}R_{\gamma}, \quad \textnormal{where} \quad  R_{\alpha}=\begin{bmatrix}1 & \\ & e^{2\mi\pi\alpha}\end{bmatrix}, T_{\beta}=\begin{bmatrix}\cos(\pi\beta) & -\sin(\pi\beta)\\\sin(\pi\beta) & \cos(\pi\beta)\end{bmatrix}.
\end{align}
To evaluate a 1-qubit gate $U(\alpha,\beta,\gamma)$, by the conjugate relation between the 1-qubit gate $U(\alpha,\beta,\gamma)$ and Pauli pads $X^{a}Z^{b}$, i.e.,
\begin{align}\label{ddd1}
U\big{(} (-1)^{a} \alpha, (-1)^{a+b}\beta, (-1)^{a}\gamma\big{)}X^{a}Z^{b}=X^{a}Z^{b}U(\alpha,\beta,\gamma)
\end{align}
it seems sufficient to directly perform the operator $U\big{(} (-1)^{a} \alpha, (-1)^{a+b}\beta, (-1)^{a}\gamma\big{)}$ on the encrypted state. Unfortunately, things are not so simple. We ignore the fact that the parameters of this operator depend on the secret keys $a,b$, which are not allowed to be known by the server. Indeed, even realizing a simple operation with a private 1-bit parameter takes a lot of effort (cf. $\tn{CNOT}^{x}$ of \cite{mahadev2018classical}).

On the other hand, we observe that the ease of evaluating Clifford gates comes from the Pauli pad, since the encrypted pad keys make private Pauli operators possible. If we choose the pad among all single-qubit unitary gates, then it will be easy to evaluate any 1-qubit gate; below, we call this new one-time pad encryption scheme \emph{the quaternion one-time pad encryption} (\textbf{QOTP}). Specifically, to evaluate a 1-qubit gate $V$ on a QOTP-encrypted state $U\bk{\psi}$, it suffices to update the one-time pad from $U$ to $UV^{-1}$, since $U\bk{\psi}=U V^{-1}(V \bk{\psi})$. This update can be easily done by classical FHE computations on encrypted pad keys, cf. Section 4.2.

Still, things are not so simple. Indeed, in the QOTP setting, the evaluation of CNOT gate (necessary for universal quantum computation) is not easy: similar to the case of (\ref{ddd1}), the 2-qubit-level CNOT gate does not preserve the pad space SU(2)$\times$SU(2) by conjugation, and the problem seems to be more complicated than before, since it is now on a 2-qubit system. Looking closely, we find that this problem can be solved in a relatively simple way by going back to the 1-qubit system.

Our solution is to rely on a conversion between QOTP and Pauli one-time pad. Specifically, we want to be able to transform a QOTP-encrypted state, together with the encrypted pad key into a Pauli-encrypted form. This allows to easily evaluate the CNOT gate on the converted ciphertext, and the resulting Pauli-encrypted state is in natural QOTP-encrypted form.

Transforming a QOTP-encrypted state to its Pauli-encrypted form is highly nontrivial. In fact, this means evaluating the decryption circuits of QOTP in the Pauli one-time pad setting, similar to the implementation of bootstrapping in classical FHE. However, apart from the hard-to-use information-theoretical secure quantum ciphertexts, the only thing we can use here for bootstrapping is the encrypted pad keys. Current QFHE techniques of taking one encrypted 1-bit as control are insufficient in utilizing encrypted multi-bit pad key. To achieve the desired conversion, we develop a new technique.\\
%We are inspired by the fact that the ciphertext format can be switched by bootstrapping in classical homomorphic encryption. Using a quantum analogy of bootstrapping, we are able to change a complicated QOTP encryption into a simpler Pauli-encrypted form, and thus enabling easy evaluation of CNOT gates in the Pauli pad setting. Let us see how to realize this particular bootstrapping, i.e., evaluating the decryption circuits of QOTP in the Pauli one-time pad setting.\textbf{Key Technique.}
\\
\textbf{Key Technique.} The new technique is an encrypted multi-bit control technique, which allows to implement (up to a Pauli matrix) any 1-qubit gate whose parameters are given in encrypted form. To see the transformation functionality of this technique, given a QOTP encryption $U(\alpha,\beta,\gamma)\bk{\psi}$ and the encrypted pad key Enc$(\alpha,\beta,\gamma)$, performing $U(\alpha,\beta,\gamma)^{-1}$ on the QOTP encryption will output a state $\bk{\psi}$ in Pauli-encrypted form.

As for the implementation of the technique, by the Euler representation (\ref{45}), the key is to implement such an operation that allows to perform any rotation $R_{\alpha}$ of the angle $\alpha$ given in encrypted form.
We call such an operation the encrypted-CROT. Below, we outline how to achieve it.

Inspired by the fact that the conditional rotation is realized by successive 1-bit controlled rotations, i.e., $R^{-1}_{\alpha}=R^{-1}_{\alpha_12^{-1}}...R^{-1}_{\alpha_m2^{-m}}$, where $\alpha=\sum^{m}_{j=1}\alpha_j2^{-j}$, $\alpha_j\in\{0,1\}$, we first consider the implementation of the encrypted 1-bit controlled rotation. By the idea of \cite{mahadev2018classical} for achieving encrypted 1-bit controlled CNOT operation, we show that it is possible to implement the encrypted 1-bit controlled rotation of arbitrary rotation angle $\omega\in[0,1)$, up to a Pauli matrix and a rotation of double angle $2\omega$. Unlike Mahadev's encrypted-CNOT, there is an additional random rotation $R^{d}_{2\omega}$ (the random bit $d\in\{0,1\}$) that is introduced to the output state. Although such rotation is undesired, it also serves as a mask to protect the output and is necessary for security, making it difficult to remove directly. Looking closely, in the multi-bit case, we observe that these undesired rotations can be removed gradually by relying on the implementation structure of the multi-bit conditional rotation.

Specifically, to realize the Enc($\alpha$)-controlled rotation $R^{-1}_{\alpha}$, given the encrypted angle Enc($\alpha$), first use as control the encrypted least significant bit $\alpha_m$ to perform the encrypted 1-bit controlled rotation $R^{\alpha_m}_{2^{-m}}$. The resulting undesired rotation is of angle $2^{-(m-1)}$, and then can be merged with the controlled rotations in the waiting list; this merging is done by homomorphic evaluations on encrypted pad keys. Using an iterative procedure, we are able to realize the desired rotation $R^{-1}_{\alpha}$, and the undesired rotation is of an angle finally growing to $1/2$, becoming a Pauli mask.

While the above implementation of encrypted-CROT relies on the Euler angle representation of SU(2), we observe that the quaternion representation of SU(2) provides an arithmetic circuit implementation of much smaller depth for the product in SU(2), more consistent with our main purpose of speeding up the evaluation of 1-qubit gate. So, in the QOTP encryption scheme, we use the quaternion-valued pad key, and the corresponding Euler angles of the pad can be obtained by classical homomorphic computations.

%Through the computational overhead lens, our scheme involves more classical FHE computations, in particular, the evaluation of 1-qubit gates is done entirely by classical FHE computations without any physical quantum operation. This undoubtedly leads to an increase in the cost of classical homomorphic evaluations, while a reduction in quantum cost is its benefit. The latter is most needed in QFHE application scenarios, where quantum computing power is a precious and scarce resource.

The encrypted-CROT technique is useful. It can make QHE of \cite{broadbent2015quantum} a QFHE. Specifically, the main problem in \cite{broadbent2015quantum} is to evaluate non-Clifford gate $T$, as defined in (\ref{z4}), on the Pauli-OTP encrypted state $X^{a}$Z$^{b}\bk{\psi}$. By the conjugate relation\footnote{It follows by combining $P^{a}X^{a}=X^{a}P^{a}Z^{a}$ and $TX^{a}Z^{b}=X^{a}Z^{a+b}P^{a}T$, cf. (5) in \cite{broadbent2015quantum}.} $TX^aZ^b = P^aX^aZ^bT$, after performing $T$ on $X^{a}Z^{b}\bk{\psi}$, the server still needs to correct $P^a$. Simple implementation of the operator $P^a$ may reveal the secret Pauli key $a$. Now with the encrypted-CROT, it is able to perform controlled phase-gate $P^a$ (i.e., rotation of angle $a/4$) with control bit $a$ given in encrypted form; see Remark \ref{cw1} for more details.

Also, the encrypted-CROT enables the private 1-bit gate evaluation in the Pauli-OTP setting. To evaluate a 1-qubit gate $V$ privately on some Pauli-encrypted state $Z^{a}X^{b}\bk{\psi}$, when given the encrypted parameters of the unitary operator $VX^{-b}Z^{-a}$, using the encrypted-CROT allows to prepare the desired state $V\bk{\psi}$ in Pauli-encrypted form.

In comparison with the scheme in \cite{mahadev2018classical}, our QFHE scheme avoids the Clifford decomposition of 1-qubit gates during evaluations, while the 1-qubit pads still need to be split into many 1-bit rotations to achieve encrypted-CROT when evaluating CNOT gates. This eventually leads to an improvement in asymptotic complexity of evaluating general circuits, as analyzed in detail in Section \ref{secc5}. Intuitively, the improvement relies on the simple fact that using $O(\lambda)$-bit binary parameters allows to describe the 1-qubit gate within approximation error $\frac{1}{2^\lambda}$; while approximating 1-qubit gate to the same accuracy, on average case $O(\lambda^2)$ number of T-gate is required by the best known algorithms. At a high level, our QFHE harnesses the efficiency advantage of binary decomposition over Clifford decomposition. In practice, finding the binary decomposition is also much easier than finding the Clifford decomposition.

%This change of representation is not necessary in the real-valued quantum computation (cf. \cite{aharonov2003simple}, Lemma 4.6 of \cite{kitaev1997quantum}), where all the involved 1-qubit quantum gates are in SO(2), and in that case, the rotation representation: $\begin{bmatrix}\cos2\pi\alpha & -\sin2\pi\alpha\\ \sin2\pi\alpha & \cos2\pi\alpha \end{bmatrix}$ where $\alpha\in[0,1)$, already provides a low-depth circuit implementation for the product in SO$(2)$. The former can be improved by the rapidly evolving technique of classical FHE, and

\subsection{Paper Organization}
We begin with some preliminaries in Section \ref{pr4}. Section \ref{secc3} presents the encrypted multi-bit control technique\,{\textemdash}\,the main technique of this paper. Section \ref{secc4} provides the QOTP encryption scheme and methods for performing homomorphic evaluations on QOTP-encrypted state. In Section \ref{secc5}, we present a new QHE scheme, show that it is a leveled QFHE, and make an efficiency comparison between the new QFHE and the previous QFHE in \cite{mahadev2018classical}.

\section{Preliminaries}\label{pr4}

\subsection{Notation}
%A super-polynomial function, $SP(\cdot)$, is a function such that for any polynomial $P(\cdot)$, $\displaystyle\lim_{\lambda\rightarrow\infty} \frac{P(\lambda)}{SP(\lambda)}=0$.

A negligible function $f=f(\lambda)$ is a function in a class \textnormal{negl}($\cdot$) of functions, such that for any polynomial function $P(\lambda)$, it holds that $\displaystyle\lim_{\lambda\rightarrow\infty} f(\lambda)P(\lambda)=0$. A probability $p(\lambda)$ is overwhelming if $1-p=\tn{negl}(\lambda)$. For all $q\in\N$, let $\Z_q$ be the ring of integers modulo $q$ with the representative elements in the range $(-q/2,q/2]\bigcap\Z$. We use $\mi$ to denote the imaginary unit, and use $\I$ to denote the identity matrix whose size is obvious from the context. We use $\mS^3=\{ \vt \big{|} ||\vt||_2=1, \vt \in\R^4 \}$ to denote the unit 3-sphere.

The $L^2$-norm of vector $\vec{a}=(a_j)$ is denoted by $\nrm{\vec{a}}:=\sqrt{\sum_{j} |a_j|^2}$. The $L^2$-spectral norm of matrix $A=(a_{ij})$ is $||A||_2=\max_{||\vec{v}||_2=1}||A\vec{v}||_2$.
The $L^\infty$-norm of $A$ is $||A||_\infty=\max_{i,j}|a_{ij}|$.

For a qubit system that has probability $p_i$ in state $\bk{\psi_i}$ for every $i$ in some index set, the density matrix is defined by $\rho=\sum_{i}p_i\ketbra{\psi_i}{\psi_i}$.

%and for any exponential function $E(x)$,
%\begin{equation}
%lim_{\lambda\rightarrow\infty} \frac{f(\lambda)}{E(\lambda)}=0.
%\end{equation}
%For instance, $\lambda^{log \lambda}$ is a such kind of function.

\textbf{$H$-distance and trace distance.}
Let $X$ be a finite set. For two quantum states $|\psi_1\rangle=\sum_{x\in X}f_1(x)|x\rangle$ and $|\psi_2\rangle=\sum_{x\in X}f_2(x)|x\rangle$, the \emph{H-distance}\footnote{For states of positive real amplitude, this distance is often referred to as \emph{the Hellinger distance}, and can be bounded by \emph{the total variation distance}, cf. Lemma 12.2 in \cite{Harsha2011Hellinger}.} between them is
\begin{align}\label{cc1}
\| |\psi_1\rangle- |\psi_2\rangle \|^{2}_{H}= \frac{1}{2}\sum_{x\in X} |f_1(x)-f_2(x)|^2.
\end{align}
The \emph{trace distance} between two normalized states $|\psi_1\rangle$ and $|\psi_2\rangle$ is
\begin{align}\label{s3}
\| |\psi_1\rangle - |\psi_2\rangle \|_{tr}=  \frac{1}{2} tr\left( \sqrt{(|\psi_1\rangle\langle\psi_1|- |\psi_2\rangle\langle\psi_2|)^\dagger(|\psi_1\rangle\langle\psi_1|- |\psi_2\rangle\langle\psi_2|)} \right).
\end{align}
If $|\psi_1\rangle$ and $|\psi_2\rangle$ are pure states, their $H$-distance and trace distance are related as following (cf. Thm 9.3.1 in \cite{wilde2013quantum}):
\begin{align}\label{42}
\| |\psi_1\rangle - &|\psi_2\rangle \|_{tr} \leq \sqrt{1-|\langle \psi_1|\psi_2 \rangle|^2}= \sqrt{1- |1- \langle \psi_1|\psi_1-\psi_2\rangle|^2}\nonumber\\
& \leq \sqrt{2| \langle \psi_1|\psi_1-\psi_2\rangle|}+| \langle \psi_1|\psi_1-\psi_2\rangle|\leq 2 \sqrt{ \| |\psi_1\rangle- |\psi_2\rangle \|_{H}}+ \sqrt{2} \| |\psi_1\rangle- |\psi_2\rangle \|_{H},
\end{align}
where $|\psi_1-\psi_2\rangle$ denotes $|\psi_1\rangle-|\psi_2\rangle$, and the last inequality is by Cauchy-Schwarz inequality.

\textbf{Discrete Gaussian distribution.} The discrete Gaussian distribution over $\Z_q$ with parameter $B\in\N$ ($B\leq \frac{q}{2}$) is supported on $\{x\in\Z_q:\,|x|\leq B\}$ and has density function
\begin{equation}\label{s1.1}
 D_{\Z_q,B}(x) \,=\, \frac{e^{\frac{-\pi |x|^2}{B^2}}}{\sum\limits_{x\in\Z_q,\, |x|\leq B}e^{\frac{-\pi | x|^2}{B^2}}} \;.
\end{equation}
For $m\in\N$, the discrete Gaussian distribution over $\Z_q^m$ with parameter $B$ is supported on $\{x\in\Z_q^m:\,||x||_{\infty}\leq B\}$ and has density
\begin{equation}
D_{\Z_q^m,B}(x) \,=\, D_{\Z_q,B}(x_1)\cdots D_{\Z_q,B}(x_m), \quad \forall x = (x_1,\ldots,x_m) \in \Z_q^m.
\end{equation}

\textbf{Pauli matrices.} The Pauli matrices $X,Y,Z$ are the following $2\times2$ unitary matrices:
 \begin{equation}\label{1.12}
X=\left[               %左括号
  \begin{array}{cc}
 & 1\\
1 &
  \end{array}
\right], \quad
Z=\left[                %左括号
  \begin{array}{cc}
 1& \\
 & -1
  \end{array}
\right],  \quad
Y=\left[               %左括号
  \begin{array}{cc}
 & -\mi\\
\mi &
  \end{array}
\right].
\end{equation}
The Pauli group (on 1-qubit) is generated by Pauli matrices. Any element in the group can be written (up to a global phase) as $X^{a}Z^{b}$ where $a,b\in\{0,1\}$.

\subsection{Representation of Single-qubit Gate }\label{sec2.3}

Any single-qubit gate can be represented by a $2\times2$ unitary matrix. We restrict our attention to the special unitary group SU(2), i.e., the group consisting of all $2\times2$ unitary matrices with determinant $1$, since any $2\times2$ unitary matrix can be written as the product of an element of SU(2) with a global phase factor, the latter being unimportant and unobservable by physical measurement (cf. Section 2.27 in \cite{nielsen2000quantum}). We first present the quaternion representation of SU(2). Recall from (\ref{1.12}) the Pauli matrices $X,\ Z,\ Y$, and denote $\sigma_1=\mi X$, $\sigma_2=\mi Z$, $\sigma_3=\mi Y$, where $\sigma_1,\sigma_2,\sigma_3\in SU(2)$. Remember that $\I_2$ denotes the $2\times2$ identity matrix. It is easy to verify that
\begin{align}
&\sigma_1\sigma_2=\sigma_3,\quad \sigma_2\sigma_3=\sigma_1,\quad \sigma_3\sigma_1=\sigma_2, \\
&\sigma_k\sigma_j=-\sigma_j\sigma_k, \qquad   k,j \in \{ 1,2,3\}, \quad k\neq j,\\
&\sigma_j^2=-\mathbb{I}_2, \qquad  \qquad  \qquad  j \in \{ 1,2,3\}.
\end{align}
So $\sigma_1, \sigma_2, \sigma_3$ can be viewed as a basis of the $\R$-space of pure quaternions.

\textbf{Elements of SU(2).} Any $2\times2$ unitary matrix must be of the form $\begin{bmatrix}x & y \\ w & z\end{bmatrix}$, where $x,y,w,z\in\C$, such that:
\begin{equation}
\left[               %左括号
  \begin{array}{cc}
 x&\ y\\
w&\ z
  \end{array}
\right]\left[                %左括号
  \begin{array}{cc}
 \bar{x}&\ \bar{w} \\
 \bar{y}&\ \bar{z}
  \end{array}
\right]=\left[                %左括号
  \begin{array}{cc}
 \bar{x}x+\bar{y}y,&\ x\bar{w}+y\bar{z} \\
 w\bar{x}+z\bar{y},&\ w\bar{w}+z\bar{z}
  \end{array}
\right]=\left[\begin{array}{cc}
 1& \\
 &1
  \end{array}
\right].
\end{equation}
Therefore, it holds that $x:(-y)=\bar{z}:\bar{w}$. Let $w=c\bar{y}$, $z=-c\bar{x}$ for some $c\in\C$. By $|x|^2+|y|^2=|w|^2+|z|^2=1$, one gets $|c|=1$. This implies that any $2\times2$ unitary matrix is of the form $\begin{bmatrix*}[r]x &\ y \\ c\bar{y} &\ -c\bar{x}\end{bmatrix*}$, where $c$ is unimodular. In particular, any element of SU(2), as a $2\times2$ unitary matrix with determinant 1, can be written as $\begin{bmatrix*}[r]x &\ y \\ -\bar{y} &\ \bar{x}\end{bmatrix*}$, where $x,y \in \C$ such that $|x|^2+|y|^2=1$.

\begin{defn}\label{2.1}
For any vector $\vec{t}=(t_1,t_2,t_3,t_4)\in\R^{4}$, the linear operator $U_{\vec{t}}$ indexed by $\vt$ is
\begin{align}\label{c2}
U_{\vec{t}}=\begin{bmatrix*}[r]
x && y  \\  -\bar{y} && \bar{x}\end{bmatrix*}, \hspace{1.5cm} \text{where} \quad x=t_1+t_3\mi, \quad y=t_4+t_2\mi.
\end{align}
\end{defn}

\begin{defn}
The quaternion representation of $U_{\vec{t}}\in\tn{SU(2)}$, where $\vt \in \mathbb{S}^3$ is
\begin{align}\label{341}
U_{(t_1,t_2,t_3,t_4)}=t_1\mathbb{I}_2+t_2\sigma_1+t_3\sigma_2+t_4\sigma_3,
\end{align}
where $\sigma_1,\sigma_2,\sigma_3$ are the basis pure quaternions.
\end{defn}
%By abuse of notation slightly we also use $U_{\vec{t}}$ to denote the operator defined as in \emph{(\ref{341})} for $|\vec{t}|\neq1$

Any element of SU$(2)$ has a unique unit 4-vector index. The \emph{inversion} and \emph{multiplication} in SU(2) are realized in the unit vector index form by:
\begin{align}\label{2023}
U_{(t_1,t_2,t_3,t_4)}^{-1}=U_{(t_1,-t_2,-t_3,-t_4)},
\end{align}
\begin{align}\label{21}
U_{(t_1,t_2,t_3,t_4)}U_{(k_1,k_2,k_3,k_4)}=U_{(t_1k_1-t_2k_2-t_3k_3-t_4k_4,\ t_1k_2+t_2k_1+t_3k_4-t_4k_3,}  \nonumber\\
_{t_1k_3+t_3k_1+t_4k_2-t_2k_4,\ t_1k_4+t_4k_1+t_2k_3-t_3k_2)}.&
\end{align}

When $|\ \nrm{\vec{t}}-1|\ll 1$ and $\nrm{\vec{t}}\neq1$, there are several methods to approximate the non-unitary operator $U_{\vec{t}}$ by a unitary operator. We give a specific method as follows:

\begin{lem}\label{2.2}
For any $\vt\in\R^{4}$ such that $\nrm{\vec{t}}\neq1$ and $\big{|} \nrm{\vec{t}}-1 \big{|} =m\leq 1$, there is an algorithm to find a vector $\vec{t}'$ such that $\nrm{\vec{t}'}=1$, $\nrm{\vec{t}-\vec{t'}}\leq \sqrt{3m}$, and $||U_{\vec{t}}-U_{\vec{t}'}||_2\leq\sqrt{3m}$.
\end{lem}

\begin{proof}
We proceed by constructing an approximate vector $\vec{t}'$. Starting from a 4-dimensional vector $\vec{t}'=\vec{0}$, assign values $t'_i=t_i$ for $i$ from $1$ to $4$, one by one, as much as possible until $\sum^{4}_{i=1}|t'_i|^2=1$. More specifically, there are two cases in total:
\begin{enumerate}
\item $\nrm{\vec{t}}\geq1$. In this case, there must exist some $l\in \{1,2,3,4\}$ such that $\sum^{l}_{i=1}t_i^2\geq1$ and $\sum^{l-1}_{i=1}t_i^2<1$. Let sgn$(t_{l})$ be the sign of $t_{l}$. We set
 \begin{equation}\label{q1}
t'_{i}=\left\{               %左括号
  \begin{array}{cc}
t_{i}, \quad & 1\leq i \leq l-1 \\
sgn(t_{l})\sqrt{1-\sum^{l-1}_{s=1}t_s^2}, \quad &i=l. \\
0, \quad & i>l
  \end{array}
\right.
\end{equation}

As an example, if $\vec{t}=(\frac{1}{2},\frac{3}{4},\frac{1}{2},\frac{1}{2})$, then $\vec{t}'=(\frac{1}{2},\frac{3}{4},\frac{\sqrt{3}}{4},0)$.

\item $\nrm{\vec{t}}<1$. In this case, we set
 \begin{equation}\label{q2}
t'_{i}=\left\{               %左括号
  \begin{array}{cc}
t_{i}, \quad & 1\leq i \leq 3 \\
sgn(t_{4})\sqrt{1-\sum^{3}_{s=1}t_s^2}, &i=4. \\
  \end{array}
\right.
\end{equation}
As an example, if $\vec{t}=(\frac{1}{2},\frac{1}{2},\frac{1}{2},0)$, then $\vec{t}'=(\frac{1}{2},\frac{1}{2},\frac{1}{2},\frac{1}{2})$.

\end{enumerate}

In cases 1, we have $(t_l-t'_l)^2\leq t^2_l-t'^2_l$, and then
\begin{align}
\nrm{\vec{t}-\vec{t'}}=\sqrt{\displaystyle\sum^{4}_{i=1}(t_i-t'_i)^2}=\sqrt{(t_l-t'_l)^2+ \displaystyle\sum^{4}_{j=l+1}t^2_j} \leq \sqrt{ \sum^{4}_{j=1}t^2_j-\sum^{4}_{j=1}t'^2_j }.
\end{align}

In cases 2, we have $(t'_4-t_4)^2\leq t'^2_4-t^2_4$, and then
\begin{align}
\nrm{\vec{t}-\vec{t'}}=\sqrt{\displaystyle\sum^{4}_{i=1}(t'_i-t_i)^2}=\sqrt{ (t'_4-t_4)^2} \leq \sqrt{ \sum^{4}_{j=1}t'^2_j-\sum^{4}_{j=1}t^2_j }.
\end{align}

In both cases, we have
\begin{align}
\nrm{\vec{t}-\vec{t'}}\leq \sqrt{|\ \nrm{\vec{t}}^2-\nrm{\vec{t'}}^2  \ |}
=\sqrt{\left|(\nrm{\vec{t}}-1)(\nrm{\vec{t}}+1)\right|}\leq \sqrt{3m}.
\end{align}
By (\ref{341}), it holds that $U_{\vec{t}}-U_{\vec{t}'}=\sum^{3}_{i=0}(t_i-t'_i)\sigma_{i}$ where $\sigma_0=\I_2$, and thus $(U_{\vec{t}}-U_{\vec{t}'})^{\dagger}(U_{\vec{t}}-U_{\vec{t}'})=\nrm{\vec{t}-\vec{t'}}^2\ \I_2$. So, we have
\begin{align}
\nrm{U_{\vec{t}}-U_{\vec{t}'}}=\max\limits_{\nrm{\vec{v}}=1}\sqrt{ \vec{v}^{\dagger}(U_{\vec{t}}-U_{\vec{t}'})^{\dagger}(U_{\vec{t}}-U_{\vec{t}'})\vec{v} } = \nrm{\vec{t}-\vec{t'}}\leq\sqrt{3m}.
\end{align}
\end{proof}$\hfill\blacksquare$

The following is a direct corollary of Lemma \ref{2.2}.
\begin{cor}\label{c2.3}
For any $4$-dimensional vector-valued function $\vt=\vt(\lambda)$ that satisfies $\big{|} \nrm{\vec{t}}-1\big{|} = \textnormal{negl}(\lambda)$, one can find a vector-valued function $\vec{t'}=\vt'(\lambda)$ that satisfies $\nrm{\vec{t'}-\vec{t}}= \textnormal{negl}(\lambda)$ and $\|\vec{t}'\|_2=1$. Moreover, it holds that $\nrm{U_{\vt}-U_{\vt'}}= \textnormal{negl}(\lambda)$.
\end{cor}

%For a 4-dimensional real-valued vector $|\vec{t}|\neq1$ that satisfies $|\ |\vec{t}|-1|\leq \textnormal{negl}(\lambda)$, we also use the notation $U_{(\vec{t})}$ to refer to the unitary $U_{(\vec{t'})}$, where $\vec{t'}$ is set to as same as $\vec{t}$ as possible\footnote{Namely, for the initial vector $\vec{t}'=\vec{0}$, assigning values $t'_i=t_i$ for $i$ from $1$ to $4$ until $\sum t_i^2=1$. For $|\vec{t}|\leq 1$, modifying $t'_4$ to meet $|\vec{t}'|=1$. It can be verified that $|\vec{t}'-\vec{t}|^2=|1-|\vec{t}|^2|$.} along with the direction from $t'_1$ to $t'_4$. We can verify that $|\vec{t}'-\vec{t}|\leq \textnormal{negl}(\lambda)$.

%As for the efficiency of this representation, a $\log^2 (\lambda)$-bit finite precision quaternion representation of SU(2) is enough to stimulate $poly(\lambda)$ number of one-qubit gates up to a negligible error (proved in Lemma \ref{234}).

%In our scheme, we introduce quaternion representation is for efficient evaluating single-qubit gate. The

%In general, quantum algorithms are designed with treating the qubits’ amplitudes as complex number. Unfortunately, there is no algorithm that can transform any given qubit to its real representation. To see this, considering two indistinguishable quantum states $|0\rangle$ and $i|0\rangle$. Translating them to their real representation meaning that we can distinguish them, since their real representation ($|00\rangle$ and $|01\rangle$) are orthogonal then distinguishable.

\textbf{Euler angle representation.} The unitary operator $U(\alpha,\beta,\gamma)$ in Euler angle representation is of the form:
\begin{equation}\label{3.11}
U(\alpha,\beta,\gamma):=\left[               %左括号
  \begin{array}{cc}   % 该矩阵一共3列，每一列都居中放置
1& \\
 & e^{2\mi\pi\alpha}
  \end{array}
\right]\left[             %左括号
  \begin{array}{cc}
\cos(\pi\beta) & -\sin(\pi\beta)\\
\sin(\pi\beta) & \cos(\pi\beta)
  \end{array}
\right]\left[               %左括号
  \begin{array}{cc}   % 该矩阵一共3列，每一列都居中放置
1& \\
 & e^{2\mi\pi\gamma}
  \end{array}
\right],   \quad  \alpha,\ \beta,\ \gamma\in[0,1).       %右括号
\end{equation}

We use $\xlongequal{\tn{i.g.p.f}}$ to denote that the equality holds after ignoring a global phase factor. By (4.11) in \cite{nielsen2000quantum}, for any unitary $U_{\vec{t}}$, there are parameters $\alpha, \beta, \gamma\in[0,1)$ such that $U(\alpha,\beta,\gamma)\xlongequal{\tn{i.g.p.f}}U_{\vt}$; conversely, for any $ \alpha,\ \gamma\in[0,1),\ \beta\in[\frac{1}{2},1]$,
\begin{align}
U(\alpha,\beta,\gamma) \xlongequal{\tn{i.g.p.f}} U(\alpha+\frac{1}{2} \mod 1,1-\beta,\gamma+\frac{1}{2} \mod 1).
\end{align}
Hence, for any $\vec{t}\in\mS^3$, there are parameters $\alpha,\ \gamma\in[0,1),\ \beta\in[0,\frac{1}{2}]$ such that \\ $U(\alpha,\beta,\gamma)\xlongequal{\tn{i.g.p.f}}U_{\vt}$. If $t_1^2+t_3^2\neq0$, then
\begin{align}\label{qo1}
U(\alpha,\beta,\gamma)&= \left[                %左括号
\begin{array}{cc}   % 该矩阵一共3列，每一列都居中放置
\cos(\pi\beta), & -\sin(\pi\beta)e^{2\pi\mi \gamma}\nonumber \\
\sin(\pi\beta)e^{2\pi\mi\alpha},  & \cos(\pi\beta)e^{2\pi\mi(\alpha+\gamma)}
  \end{array}
\right] \xlongequal{\tn{i.g.p.f}}\left[               %左括号
  \begin{array}{cc}   % 该矩阵一共3列，每一列都居中放置
t_1+t_3\mi ,& t_4+t_2\mi\\
-t_4+t_2\mi ,& t_1-t_3\mi
  \end{array}
\right]\\
&\xlongequal{\tn{i.g.p.f}}\frac{t_1+t_3\mi}{\sqrt{t^2_1+t^2_3}}\left[               %左括号
  \begin{array}{cc}   % 该矩阵一共3列，每一列都居中放置
\sqrt{t_1^2+t_3^2},&\ t_4+t_2 \mi \frac{\sqrt{t^2_1+t^2_3}}{t_1+t_3\mi}\\
-t_4+t_2 \mi\frac{\sqrt{t^2_1+t^2_3}}{t_1+t_3\mi} ,&  t_1-t_3\mi\frac{\sqrt{t^2_1+t^2_3}}{t_1+t_3\mi}
  \end{array}
\right],  %右括号
\end{align}
in particular,
\begin{align}\label{w2}
\begin{array}{lcl}
 \cos(\pi\beta)&=&\sqrt{t_1^2+t_3^2}, \vspace{1ex} \\
\sin(\pi\beta)&=&\sqrt{t_2^2+t_4^2},  \vspace{1ex} \\
e^{2\pi \mi\alpha}\sqrt{t_2^2+t_4^2}&=&\frac{-t_4+t_2 \mi}{t_1+t_3\mi}\sqrt{t^2_1+t^2_3},  \vspace{1ex} \\
-e^{2\pi \mi\gamma}\sqrt{t_2^2+t_4^2}&=&\frac{t_4+t_2 \mi}{t_1+t_3\mi}\sqrt{t^2_1+t^2_3}.
\end{array}
\end{align}

\subsection{Clifford gates and Pauli one-time pad encryption}
The following are formal definitions \cite{ambainis2000private,broadbent2015quantum} of some terminology mentioned in Section 1.
\begin{itemize}
\item[]
The \emph{Pauli group} on $n$-qubit system is $P_{n}=\{V_1 \otimes...\otimes V_n | V_j\in\{X,\ Z,\ Y,\ \I_2\},1\leq $\\$j\leq n \}$. The \emph{Clifford group} is the group of unitaries preserving the Pauli group:
$$C_{n}=\{ V\in U_{2^{n}} |V P_n V{^\dagger}=P_n \}.$$
A \emph{Clifford gate} refers to any element in the Clifford group. A generating set of the Clifford group consists of the following gates:
 \begin{equation}\label{z4}
X, \quad
Z,  \quad
P=\left[               %左括号
  \begin{array}{cc}
1 & \\
 & \mi
  \end{array}
\right], \quad
H=\frac{1}{\sqrt{2}}\left[               %左括号
  \begin{array}{cc}
1 &1 \\
1 & -1
  \end{array}
\right], \quad
\text{CNOT}=\left[               %左括号
  \begin{array}{cccc}
1 & & &\\
 & 1& &\\
 & & &1\\
 & &1 &
  \end{array}
\right].
 \end{equation}
Adding any \emph{non-Cliffod gate}, such as $T=\begin{bmatrix}1 & \\ & e^{\mi\frac{\pi}{4}}\end{bmatrix}$, to (\ref{z4}), leads to a universal set of quantum gates.
\end{itemize}

The \emph{Pauli one-time pad encryption}, traced back to \cite{ambainis2000private}, encrypts a multi-qubit state qubitwise. The scheme for encrypting 1-qubit message $\bk{\psi}$ is as follows:
\begin{itemize}
\setlength{\itemsep}{5pt}
\item \textbf{Pauli one-time pad encryption}
\item[$\bullet$] Keygen(). Sample two classical bits $a,b \leftarrow \{0,1\}$, and output $(a,b)$.
\item[$\bullet$] Enc($(a,b)$, $\bk{\psi}$). Apply the Pauli operator $X^{a}Z^{b}$ to a 1-qubit state $\bk{\psi}$, and output the resulting state $\bk{\widetilde{\psi}}$.
\item[$\bullet$] Dec($(a,b)$,$\bk{\widetilde{\psi}}$). Apply $X^{a}Z^{b}$ to $\bk{\widetilde{\psi}}$.
\end{itemize}
Since $XZ=-ZX$, the decrypted ciphertext is the input plaintext up to a global phase factor $(-1)^{ab}$. In addition, the Pauli one-time pad encryption scheme guarantees the information-theoretic security, since for any 1-qubit state $\bk{\psi}$, it holds that
\begin{align}
\frac{1}{4}\sum_{a,b\in \{0,1\} }X^{a}Z^{b}|\psi\rangle\langle\psi|Z^{b}X^{a}=\frac{\I_2}{2}.
\end{align}

\subsection{Pure and Leveled Fully Homomorphic Encryption}
The following definitions come from \cite{brakerski2018quantum} and \cite{mahadev2018classical}.
A homomorphic (public-key) encryption scheme HE = (HE.Keygen, HE.Enc, HE.Dec, HE.Eval) for single-bit plaintexts is a quadruple of PPT algorithms as following:

\begin{itemize}
\item \textbf{Key Generation.} The algorithm $(pk,evk,sk)\leftarrow$ HE.Keygen$(1^{\lambda})$ on input the security parameter $\lambda$ outputs a public encryption key $pk$, a public evaluation key $evk$ and a secret decryption key $sk$.
\item \textbf{Encryption.} The algorithm $c\leftarrow$ HE.Enc$_{pk}(\mu)$ takes as input the public key $pk$ and a single bit message $\mu\in\{0,1\}$, and outputs a ciphertext $c$. The notation HE.Enc$_{pk}(\mu;r)$ will be used to represent the encryption of message $\mu$ using random vector $r$.
\item \textbf{Decryption.} The algorithm $\mu^*\leftarrow$ HE.Dec$_{sk}(c)$ takes as input the secret key $sk$ and a ciphertext $c$, and outputs a message $\mu^*\in\{0,1\}$.
\item \textbf{Homomorphic Evaluation.} The algorithm $c_{f}\leftarrow$ HE.Eval$_{evk}(f,c_1,\ldots,c_l)$ on input the evaluation key $evk$, a function $f:\{0,1\}^l\rightarrow\{0,1\}$ and $l$ ciphertexts $c_1,\ldots,c_l$, outputs a ciphertext $c_f$ satisfying:
\begin{equation}
\mathrm{HE.Dec}_{sk}(c_f) = f(\mathrm{HE.Dec}_{sk}(c_1),\ldots,\mathrm{HE.Dec}_{sk}(c_l))
\end{equation}
with overwhelming probability.
\end{itemize}

%\remark{Some homomorphic encryption schemes do not require evaluation keys, cf. BGV \cite{brakerski2014leveled}, GSW \cite{gentry2013homomorphic}.}

\begin{defn}(Classical pure FHE and leveled FHE)\\
A homomorphic encryption scheme is compact if its decryption circuit is independent of the evaluated function. A compact scheme is (pure) fully homomorphic if it can evaluate any efficiently computable boolean function. A compact scheme is leveled fully homomorphic if it takes $1^L$ as additional input in key generation, where parameter $L$ is polynomial in the security parameter $\lambda$,  and can evaluate all Boolean circuits of depth $\leq L$.
\end{defn}

\textbf{Trapdoor for LWE problem.} Learning with errors (LWE) problem \cite{regev2009lattices} is the security basis of most FHE schemes. Let $m,n,q$ be integers, and let $\chi$ be a distribution on $\Z_{q}$. The search version of LWE problem is to find $\s\in\Z_q^n$ given the LWE samples ($A$, $A\s+\e$ mod $q$), where $A\in\Z_q^{m \times n}$ is sampled uniformly at random, and $\e$ is sampled randomly from the distribution $\chi^m$. Under a reasonable assumption on $\chi$ (namely, $\chi(0) > 1/q + 1/\tn{poly}(n))$, an algorithm for the search problem with running time $2^{O(n)}$ is known \cite{blum2003noise}, while no polynomial time algorithm is known.

%Let $\lambda$ be the security parameter.  is to distinguish distribution between ($A$, $A\s+\e$) and $(A, u)$, where $A\in\Z_q^{n \times m}$, $s\in\Z_q^n$, $u\in\Z_q^m$ are sampled uniformly at random, $\e$ drawn at random from some distribution $\chi(\lambda)^m$. A large number of current FHE scheme bases on the hardness of LWE$_{n,m,q,\chi}$, where

Although it is hard to solve LWE problem in the general case \cite{brakerski2013classical,peikert2009public,peikert2017pseudorandomness,regev2009lattices}, it is possible to generate a matrix $A\in Z_{q}^{m\times n}$ and a relevant matrix, called the \emph{trapdoor} of $A$ (cf. Definition 5.2 in \cite{micciancio2012trapdoors}), that allow to efficiently recover $\s$ from the LWE samples $(A, A\s + \e)$:

%LWE samples $(A, As + e)$, there is a special matrix called trapdoor (cf. Definition 5.2 in \cite{micciancio2012trapdoors}) that allows to recover $s,e$ from the samples.

\begin{lem}\label{4200} \rm{\textbf{(Theorem 5.1 in \cite{micciancio2012trapdoors}; Theorem 3.5 in \cite{mahadev2018classical})}} Let $n, m \geq1$, $q \geq 2$ such that $m = \Omega(n\log q)$. There is an
efficient randomized algorithm GenTrap$(1^n,1^m,q)$ that returns a matrix $A\in \mathbb{Z}^{m\times n}_q$ and a trapdoor $t_A\in\Z^{(m-n\log q)\times n\log q}$, such that the distribution of $A$ is negligibly (in $n$) close to the uniform distribution on $\Z^{m\times n}_q$. Moreover, there is an efficient algorithm ``Invert'' that, on input $A$, $t_A$ and $A\vec{s}+\vec{e}$, where $\vec{s}\in\Z^{n}_q$ is arbitrary, $\|\vec{e}\|_{2}\leq q/(C_T \sqrt{n \log q})$, and $C_T$ is a universal constant, returns $\vec{s}$ and $\vec{e}$ with overwhelming probability.
\end{lem}
GenTrap($\cdot$) can be used to generate a public key (matrix $A$) with a trapdoor $(t_{A})$ for LWE-based FHE schemes.

In the quantum setting, a quantum homomorphic encryption (QHE) is a scheme with syntax similar to the above classical setting, and is a sequence of algorithms (QHE.Keygen, QHE.Enc,QHE.Dec, QHE.Eval). A hybrid framework of QHE with classical key generation is given below.
\begin{itemize}
\item \textbf{QHE.Keygen} The algorithm $(pk,evk,sk)\leftarrow$ HE.Keygen$(1^{\lambda})$ takes as input the security parameter $\lambda$, and outputs a public encryption key $pk$, a public evaluation key $evk$, and a secret key $sk$.
\item \textbf{QHE.Enc}  The algorithm $\bk{c}\leftarrow$ QHE.Enc$_{pk}(\bk{m})$ takes the public key $pk$ and a single-qubit state $\bk{m}$, and outputs a quantum ciphertext $\bk{c}$.
\item \textbf{QHE.Dec}  The algorithm $\bk{m^*}\leftarrow$ QHE.Dec$_{sk}(\bk{c})$ takes the secret key $sk$ and a quantum ciphertext $\bk{c}$, and outputs a single-qubit state $\bk{m^*}$ as the plaintext.
\item \textbf{QHE.Eval}  The algorithm $\bk{c'_1},\ldots,\bk{c'_{l'}}\leftarrow$ QHE.Eval$(evk,C,\bk{c_1},\ldots,\bk{c_l})$ takes the evaluation key $evk$, a classical description of a quantum circuit $C$ with $l$ input qubits and $l'$ output qubits, and a sequence of quantum ciphertexts $\bk{c_1},\ldots,\bk{c_{l}}$. Its output is a sequence of $l'$ quantum ciphertexts $\bk{c'_1},\ldots,\bk{c'_{l'}}$.
\end{itemize}

\begin{defn}(Quantum pure FHE and leveled FHE)\\
Given a scheme $\tn{QHE=(QHE.Key,QHE.Enc,QHE.Eval,QHE.Dec)}$ and the security parameter $\lambda$, with the keys $(pk,evk,sk)=$ \tn{HE.Keygen}$(1^{\lambda})$, the scheme is called quantum fully homomorphic, if for any BQP circuit $C$ and any $l$ single-qubit states $\bk{m_1},...,\bk{m_l}$ where $l$ is the number of input qubits of $C$, the state $C(\bk{m_1},...,\bk{m_l})$ is within negligible trace distance from the state \tn{QHE.Dec}$_{sk}\big{(} \tn{QHE.Eval}(evk,C,\bk{c_1},...,$\\$\bk{c_l}) \big{)}$, where $\bk{c_i}=$ \tn{\tn{QHE.Enc}}$_{pk}(\bk{m_i})$. The scheme is leveled quantum fully homomorphic if it takes $1^L$ as additional input in key generation, and can evaluate all depth-$L$\footnote{The depth of a quantum circuit refers to the number of layers of the circuit, where each layer consists of quantum gates that can be executed in parallel.} quantum circuits.
\end{defn}

The difference between leveled QFHE and pure QFHE is whether there is an a-priori bound on the depth of the evaluated circuit. A circular security assumption can help convert a leveled classical FHE into a pure classical FHE \cite{gentry2009fully1,gentry2009fully}. In the existing QFHE schemes \cite{brakerski2018quantum,mahadev2018classical}, the security assumption is required to make the classical FHE that encrypts Pauli keys a pure FHE. In \cite{brakerski2018quantum,mahadev2018classical}, there is no quantum analogy of `bootstraping' that is able to reduce the noise of a quantum ciphertext.

\subsection{Quantum Capable Classical Homomorphic Encryption}\label{sec2.4}

%\subsubsection{Polynomially bounded function.} A function $n: \N \to \R_+$ is \emph{polynomially bounded} if there exists a polynomial $p$ such that $n(\lambda)\leq p(\lambda)$ for all $\lambda \in \N$.

In \cite{mahadev2018classical}, Mahadev proposed a FHE scheme called quantum capable classical homomorphic encryption (Definition 4.2 in \cite{mahadev2018classical}), which is an LWE-based FHE scheme of GSW-style with a trapdoor to the public key. The scheme consists of two sub-schemes: Mahadev's HE (MHE) and Alternative HE (AltMHE).

\vspace{0.5cm}
%allowing recovery of $(\mu_0, r_0)$ from a ciphertext Enc$(\mu_0; r_0)$ (Enc$(\mu_0; r_0)$ denotes the encryption of a bit $\mu_0$ with randomness variables $r_0$ ).

 % for achieving the encrypted CNOT operation .

%It consists of two parts, the scheme HE and the scheme AltHE.

%Shortly later, Brakerski propose a more efficient classical FHE scheme for realizing encrypted CNOT operation. The scheme goes beyond the definition of QCHE, and thus without using the term of `QCHE'.

%goes beyond the definition of QCHE, and thus without using the term of `QCHE'.

%For convenience, we use QCHE to refer to both two particular classical FHE schemes presented in \cite{brakerski2018quantum} and \cite{mahadev2018classical}.

%Brakerski's scheme with polynomial modulus are more efficient than Mahadev's scheme with exponential modulus.

%In this subsection, we introduce Mahadev's QCHE scheme. It consists of two parts, the scheme HE and the scheme AltHE.

\textbf{Notation 1.}\emph{The parameters in MHE scheme are the following:}
\begin{itemize}

\item[1.]The security parameter: $\lambda$. All other parameters are functions in $\lambda$.
\item[2.]The modulus: $q$, which is a power of $2$. Also, $q$ satisfies item 7 below.
\item[3.]The size parameters: $n$=\emph{$\tn{poly}(\lambda)$}, $m=\Omega(n\log q)$, and $N=(m+1)\log q$. The gadget matrix is \\
$G=\mathbb{I}_{m+1}\bigotimes(1,2,2^2,..,q/2)\in\mathbb{Z}^{(m+1)\times N}$.
\item[4.]$L^{\infty}$-norm of the initial encryption noise: it is bounded by the parameter $\beta_{init}\geq2\sqrt{n}$.

\vspace{0.3cm}
There are two more parameters indicating the evaluation capability of the HE scheme:
\item[5.](Classical capability) the maximal evaluation depth $\eta_{c}$ before bootstrapping. $\eta_{c}=\Theta(\log(\lambda))$ is required to be larger than the depth of the decryption circuit, so that before bootstrapping, the accumulated noise of the ciphertext can be upper bounded by $\beta_{init}(N+1)^{\eta_{c}}$ (cf. Theorem 5.1 in \cite{mahadev2018classical}).
\item[6.](Quantum capability) the so-called CNOT-precision $\eta$ that satisfies $\eta=\Theta(\log(\lambda))$, so that the encrypted CNOT operation of the HE scheme can yield a resulting state within \~{O}($\frac{1}{(N+1)^{\eta}}$) trace distance from the correct one with all but \~{O}($\frac{1}{(N+1)^{\eta}}$) probability.\footnote{More details can be found in Lemma 3.3 of \cite{mahadev2018classical}.}
\item[7.]Let $\beta_f=\beta_{init}(N+1)^{\eta+\eta_{c}}$. It is required that $q>4(m+1)\beta_{f}$, which implies that $q$ is superpolynomial in $\lambda$.
\end{itemize}

%\footnote{This requirement ensures that Mahadev's scheme can evaluate any polynomial-depth quantum circuit to obtain the ideal resulting state up to a trace negligible distance.}

%\footnote{We must to remainder the readers the fact that each quantum CNOT operation has 2 fan-out. This implies $\lambda(\lambda+1)/2$  CNOT gates can compose a $\lambda$-depth circuits, where the degree of errors can grow exponentially as $2^{k}$ [see Figure]. However, it seems too strict to measure a random error using its upper-bound degree, we adopt the author's argument that the approximation error caused by a $poly(\lambda)$ number of $\frac{1}{SP(\lambda)}$-approximation CNOT is negligible.}

Mahadev's (public-key) homomorphic encryption scheme MHE=(MHE.Keygen; \\ MHE.Enc; MHE.Dec; MHE.Convert; MHE.Eval) is a PPT algorithm as follows \cite{mahadev2018classical}:

\begin{itemize}
\setlength{\itemsep}{5pt}
\item \textbf{MHE Scheme (Scheme 5.2 in \cite{mahadev2018classical})}
\item[$\bullet$] MHE.KeyGen: Choose $\e_{sk}\in\{0,1\}^{m}$ uniformly at random. Use \tn{GenTrap}$(1^n,1^m,q)$ in Lemma \ref{4200} to generate a matrix $\*A$ $\in \mathbb{Z}^{m\times n}$, together with the trapdoor $t_{\*A}$. The secret key is $sk=(-\e_{sk},1) \in \mathbb{Z}^{m+1}_q$. The public key $\*A' \in \Z^{(m+1)\times n}_q$ is the matrix composed of $\*A$ (the first $m$ rows) and $\e^{T}_{sk}\*A \bmod q$ (the last row).

\item[$\bullet$] MHE.Enc$_{pk}$($\mu$): To encrypt a bit $\mu\in\{0,1\}$, choose $\*S\in \mathbb{Z}^{n\times N}_q$ uniformly at random and create $E \in \mathbb{Z}^{(m+1)\times N}_q$ by sampling each entry of it from $D_{\mathbb{Z}_q,\beta_{init}}$. Output $\*A'\*S+\*E+\mu \*G\in \mathbb{Z}^{(m+1)\times N}_q$.

\item[$\bullet$] MHE.Eval($C_0,C_1$): To apply the NAND gate, on input $C_0, C_1$, output $G-C_0\cdot G^{-1}(C_1)$.

\item[$\bullet$] MHE.Dec$_{sk}$($C$): Let $c$ be column $N$ of $C\in\Z_{q}^{(m+1)\times N}$, compute $b'=sk^{T}c\in\mathbb{Z}_{q}$. Output 0 if $b'$ is closer to 0 than to $\frac{q}{2} \bmod q$, otherwise output 1.

\item[$\bullet$] MHE.Convert($C$): Extract column $N$ of $C$.
\end{itemize}

\begin{itemize}
\item[]
\textbf{XOR.} The XOR operation on two bits $a,b\in\{0,1\}$ is defined by $a\oplus b:=a+b \mod 2$.

%\textbf{Homomorphic XOR operation.} The \emph{homomorphic XOR}, denoted by $\oplus$, is to evaluate the XOR operation homomorphically on the ciphertexts. Similar terms include \emph{homomorphic multiplication}, \emph{homomorphic addition} and \emph{homomorphic computation}.

\end{itemize}

Although the MHE scheme preserves the ciphertext form during homomorphic evaluation, when evaluating XOR operation, the noise in the output ciphertext is not simply the addition of the two input noise terms.
To overcome this drawback, Mahadev defined an extra operation MHE.Convert, which is capable of converting a ciphertext of the MHE scheme to a ciphertext of the following AltMHE scheme:
\begin{itemize}
\setlength{\itemsep}{7pt}
\item \textbf{AltMHE Scheme (Scheme 5.1 in \cite{mahadev2018classical})}
\item[$\bullet$] AltMHE.KeyGen: This procedure is the same as that of MHE.KeyGen.
\item[$\bullet$] AltMHE.Enc$_{pk}$($\mu$): To encrypt a bit $\mu\in\{0,1\}$, choose $\s\in \mathbb{Z}^{n}_q$ uniformly at random and create $\e\in \mathbb{Z}^{m+1}_q$ by sampling each entry from $D_{\mathbb{Z}_q,\beta_{init}}$. Output $A'\s+\e+(0,...,0,\mu\frac{q}{2})\in \mathbb{Z}^{m+1}_q$.
\item[$\bullet$] AltMHE.Dec$_{sk}$($c$): To decrypt $c$, compute $b'=sk^{T}c\in\mathbb{Z}_{q}$. Output $0$ if $b'$ is closer to $0$ than to $\frac{q}{2}$ mod $q$, otherwise output 1.
\end{itemize}

There are several useful facts:
\begin{itemize}
\setlength{\itemsep}{2pt}
\item[1.] AltMHE scheme has a natural evaluation of the XOR operation: adding two ciphertexts encrypting $\mu_0$ and $\mu_1$ respectively results in a ciphertext encrypting $\mu_0\oplus \mu_1$.
\item[2.] For a ciphertext AltMHE.Enc$(u;(\s,\e))$ with error $\e$ such that $\|\e \|_{2} < \frac{q}{4\sqrt{m+1}}$, the decryption procedure outputs $u\in\{0,1\}$ correctly (since ${sk}^T A' = \*0$).
\item[3.] The trapdoor $t_{A}$ can be used to recover the random vectors $\s,\e$ from ciphertext AltMHE.Enc$(u;(\s,\e))\in\Z^{m+1}_q$, and thus the plaintext $u$ can also be recovered. To see this, note that the first $m$ entries of the ciphertext can be written as $A\s + \e'$, where $\e'\in \Z_q^{m}$. Therefore, the inversion algorithm ``\tn{Invert}'' in Lemma \ref{4200} outputs $\s,\e$ on input $A\s + \e'$ and $t_{A}$, as long as $\| \e' \|_{2} < \frac{q}{C_T\sqrt{n\log q}}$ for the universal constant $C_T$ in Lemma \ref{4200}.
\item[4.] The trapdoor $t_{A}$ allows recovery of the plaintext $u$ from ciphertext MHE.Enc$(u)$ by using MHE.Convert.
\end{itemize}

Recall that a fresh \text{AltMHE} ciphertext is of the form \text{AltMHE.Enc}$(u;r)$ where the plaintext $u \in \{0,1\}$ and random vector $r=(\s,\e)\xleftarrow{} (\mathbb{U}_{\Z^{n}_{q}},D_{\Z^{m+1}_{q},\beta_{init}})$.
By item 5 of Notation 1, throughout the homomorphic computation, any MHE ciphertext is always of the form MHE.Enc$(u;(\s,\e))$, where $||\e||_{\infty}\leq \beta_{init}(N+1)^{\eta_c}$ (cf. Section 5.2.1 of \cite{mahadev2018classical}).
\begin{lem}[Theorem 5.2 in \cite{mahadev2018classical}]\label{41}
With the parameters of Notation 1, throughout the homomorphic computation, any \tn{MHE} ciphertext can be converted to a ciphertext of the form $c'=\emph{AltMHE.Enc}(u';(\s',\e'))$ by using the function $\emph{MHE.Convert}$. Then $||\e'||_{\infty}\leq \beta_{init}(N+1)^{\eta_c}$, $||\e'||_{2}\leq \beta_{init}(N+1)^{\eta_c}\sqrt{m+1}$, and the Hellinger distance between the following two distributions:
\begin{align}
&\{\textnormal{AltMHE.Enc}(\mu;r)|(\mu,r)\xleftarrow{}  (\mathbb{U}_{\Z^{n}_{q}},D_{\Z^{m+1}_{q},\beta_{f}})\}   \tn{    and    }\nonumber\\
& \{\textnormal{AltMHE.Enc}(\mu;r) \oplus c' |(\mu,r) \xleftarrow{} (\mathbb{U}_{\Z^{n}_{q}},D_{\Z^{m+1}_{q},\beta_{f}})\}
\end{align}
is equal to the Hellinger distance between the following two distributions:
\begin{equation}\label{eq:gaussiandist}
\{\e|\e\xleftarrow{}  D_{\Z_q^{m + 1},\beta_{f}}\} \quad  \tn{and} \quad \{\e + \e'|\e \xleftarrow{} D_{\Z_q^{m + 1},\beta_{f}}\},
\end{equation}
which is negligible in $\lambda$.
\end{lem}

\begin{lem}\label{61}
Let parameter $\beta_{f}=\beta_{init}(N+1)^{\eta_c+\eta}$ be as in Notation 1, and let $\e'\in\Z^{m+1}_{q}$ satisfy $||\e'||_{\infty} \leq \beta_{init}(N+1)^{\eta_c}$. Let $\rho_{0}$ be the density function of the truncated discrete Gaussian distribution $D_{\mathbb{Z}^{m+1}_q,\beta_{f}}$, and let $\rho_{1}$ be the density function of the shifted distribution $\e'+D_{\mathbb{Z}^{m+1}_q,\beta_{f}}$. Let $\widetilde{D}_{\mathbb{Z}^{m+1}_q}$ be the distribution of the random vector sampled from the distribution $D_{\mathbb{Z}^{m+1}_q,\beta_{f}}$ and the distribution $\e'+D_{\mathbb{Z}^{m+1}_q,\beta_{f}}$ with probability $p$ and $1-p$, respectively. For any $0\leq p\leq1$, any one-qubit state $|c\rangle=c_0|0\rangle+c_1|1\rangle$, any vector $\vec{\omega}\in \Z^{m+1}_q \leftarrow \widetilde{D}_{\mathbb{Z}^{m+1}_q}$, the trace distance between state $|c\rangle$ and the state
\begin{align}\label{nw1}
|c'\rangle=\frac{1}{\sqrt{\rho_{0}(\vec{\omega})|c_0|^2+\rho_{1}(\vec{\omega})|c_1|^2}} (\sqrt{\rho_{0}(\vec{\omega})}c_0|0\rangle+\sqrt{\rho_{1}(\vec{\omega})}c_1|1\rangle)
\end{align}
is $\lambda$-negligible with overwhelming probability.
\end{lem}

\emph{(Sketch Proof.)} The main idea is to prove that when $\vec{\omega}$ is sampled from $\widetilde{D}_{\mathbb{Z}^{m+1}_q}$, $\rho_{0}(\vec{\omega})$ and $\rho_{1}(\vec{\omega})$ are with overwhelming probability so close to each other that the ratio $\frac{\rho_{0}(\vec{\omega} )}{\rho_{1}(\vec{\omega})}$ is $\lambda$-negligibly close to 1, so that the normalized form of $(\sqrt{\rho_{0}(\vec{\omega})}c_0|0\rangle+\sqrt{\rho_{1}(\vec{\omega})}c_1|1\rangle)$ is within $\lambda$-negligible trace distance to the state $c_0|0\rangle+c_1|1\rangle$. The detailed proof can be found in Appendix \ref{vcb1c}.

\section{Encrypted Multi-bit Control Technique}\label{secc3}
%\subsubsection{Pauli mask} Let $a,b\in\{0,1\}$, let $\bk{\psi}$ be a single-qubit state, and let $X$ and $Z$ be the Pauli operators defined in (\ref{1.12}). The Pauli operators in state $X^{a}Z^{b}\bk{\psi}$ are called
%\emph{Pauli mask}, since it serves as a mask hiding the state $\bk{\psi}$.
%We begin with some terminology and notations arising in this section.
The main technique in this section, called encrypted conditional rotation (encrypted-CROT), is to use the encrypted $m$-bit angle MHE.Enc$(\alpha)$ to perform $R^{-1}_{\alpha}$, up to a Pauli matrix, where $\alpha=\sum^{m}_{j=1}\alpha_j2^{-j}$, $\alpha_j\in\{0,1\}$.\footnote{The case of using MHE.Enc$(\alpha)$ to to implement $R_{\alpha}$ is similar.} Below, we first explain the basic idea for achieving the encrypted-CROT.

Observe that the classical conditional rotation $R^{-1}_{\alpha}$ is realized by $m$ successive 1-bit controlled rotations: for each $1\leq j\leq m$, the corresponding rotation is $ R^{-1}_{\alpha_j2^{-j}}=R^{-\alpha_j}_{2^{-j}}$, where $\alpha_j\in\{0,1\}$ is the control bit. We first consider the implementation of 1-bit control rotation $R^{-\alpha_{j}}_{w}$ when both $w\in[0,1)$ and the encrypted 1-bit MHE.enc($\alpha_{j}$) are given. By Mahadev's idea for achieving encrypted CNOT operation (also 1-bit controlled operation), we will show in Lemma \ref{11} that given the encrypted 1-bit Enc$(\alpha_{j})$ and a general 1-qubit state $\bk{\psi}$, one can perform
\begin{align}\label{wq11}
\bk{\psi}\rightarrow Z^{d_1}R^{d_2}_{2w}R^{-\alpha_j}_{w}\bk{\psi},
\end{align}
where random parameters $d_1,d_2\in\{0,1\}$ are induced by quantum measurement in the algorithm. Unlike Mahadev's encrypted CNOT operation, when the output state of (\ref{wq11}) is taken as an encryption of $R^{-\alpha_j}_{w}\bk{\psi}$, in addition to Pauli mask $Z^{d_1}$, there is also an undesired rotation $R^{d_2}_{2w}$. Fortunately, when using MHE.Enc$(\alpha)$ to implement $R^{-1}_{\alpha}$, if we use MHE.Enc$(\alpha_m)$, the encryption of the least significant bit of $\alpha$, to perform conditional rotation $R_{-2^{-m}\alpha_m}$ by Lemma \ref{11} (i.e., setting $w=2^{-m}$ in (\ref{wq11})), then the undesired operator is $R^{d_2}_{2^{-(m-1)}}$. Since the conditional rotations that remained to be performed are $R^{-1}_{\alpha_j2^{-j}}$ ($1\leq j\leq m-1$), the undesired operator $R^{-1}_{-d_22^{-(m-1)}}$ can be merged with $R^{-1}_{\alpha_{m-1}2^{-(m-1)}}$ in the waiting list. By iteration, as to be shown in Theorem \ref{12}, we realize the encrypted-CROT as follows:
\begin{align}\label{qz2}
\bk{\psi}\rightarrow   Z^{d}R^{-1}_{\alpha}\bk{\psi},
\end{align}
where $d\in\{0,1\}$ is a random parameter.

Similarly, with MHE.Enc$(\alpha)$ at hand, one can implement another kind of rotation $T_{\alpha}$ as defined in (\ref{45}) as follows:
\begin{equation}\label{z231}
\bk{\psi}\rightarrow Z^{d}X^{d}T^{-1}_{\alpha}|\psi\rangle.
\end{equation}
Combining (\ref{qz2}) and (\ref{z231}) gives a general encrypted conditional unitary operator acting on a single qubit (Theorem \ref{75}). That is, for any unitary $U=U(\alpha,\beta,\gamma)=R_{\alpha}T_{\beta}R_{\gamma}$ (cf. (\ref{45})), where $\alpha,\beta,\gamma\in[0,1)$ are multi-bit binary angles, given the encrypted angles Enc$(\alpha,\beta,\gamma)$ and a general 1-qubit state $\bk{\psi}$, one can efficiently perform:
\begin{equation}
\bk{\psi}\rightarrow Z^{d_1}X^{d_2}U^{-1}\bk{\psi},
\end{equation}
where $d_1,d_2\in\{0,1\}$ are random parameters.\\

The following are formal definitions of some terms to be used in this section:
\begin{enumerate}
\item[]
\textbf{Up to Pauli operator.} We say that a unitary transform $U$ is applied to a $1$-qubit state $|\psi\rangle$ \emph{up to a Pauli operator}, if the following is implemented:
\begin{align}\label{vvv1}
\bk{\psi}\rightarrow VU\bk{\psi}, \quad \text{where}\quad V \in \text{\{Pauli matrices $X,Y,Z$, identity matrix $\I_2$\}}
\end{align}

%\textbf{Negl$(\lambda)$-compute} We say that we \emph{negl$(\lambda)$-compute} a quantum state if we compute a state that is within negl$(\lambda)$-trace distance of that state.

%\textbf{Rotations.} For $\alpha \in [0,1),$ define the rotations of angle $\alpha$:\begin{equation}\label{41234}
%T_{\alpha}=\left[                 %左括号
%  \begin{array}{cc}
%\cos(\pi\alpha) & -\sin(\pi\alpha)\\
%\sin(\pi\alpha) & \cos(\pi\alpha)
%  \end{array}
%\right],   \quad   R_{\alpha}=\left[                 %左括号
%  \begin{array}{cc}   % 该矩阵一共3列，每一列都居中放置
%1& \\
% & e^{2i\pi\alpha}
 % \end{array}
%\right].        %右括号
%\end{equation}

\vspace{0.3cm}

\textbf{Uniform distribution.} $\mathbb{U}_{\Z_{q}}$ denotes the uniform distribution over $\Z_{q}$ for some $q\in\Z$.

\vspace{0.3cm}

\textbf{Bit string.} A \emph{bit string} is a sequence of bits, each taking value 0 or 1.

\vspace{0.3cm}
%\textbf{AltMHE ciphertext.}\label{ce1} Recall the parameters $q,m,n,\beta_f,\beta_{init}$ from Notation 1 in Section \ref{sec2.4}, and let $\mathbb{U}_{Z^{n}_{q}}$ be the uniform distribution over $Z^{n}_{q}$. We use notation $\text{AltMHE.Enc}(u;r)$ to refer to an encryption of $u$ with the \emph{random vector} $r=(\s,\e)$. Remember that for a fresh ciphertext $\text{AltMHE.Enc}(u;r)$, $s\leftarrow \mathbb{U}_{Z^{n}_{q}}$ and $e\leftarrow D_{Z^{m+1}_{q},\beta_{init}}$, and thus $r\leftarrow (\mathbb{U}_{Z^{n}_{q}},\hskip 2pt D_{Z^{m+1}_{q},\beta_{init}})$. We use $\text{AltMHE.Enc}(u)$ to denote the encryption of $u$, without specifying $r$. We often use $\delta$ to denote the density function of the joint distribution $\mathbb{U}_{Z^{n}_{q}}\times D_{Z^{m+1}_{q},\beta_{f}}$.

\textbf{$k$-bit binary fraction.} For $m\in \N$, the \emph{$m$-bit binary fraction} represents a real number in the range $[0,1]$ of the form $x=\sum^{m}_{j=1}2^{-j}x_j$, where $x_j\in\{0,1\}$ for $1\leq j \leq m$. The binary representation of $x\in[-1,1]$ is $x=(-1)^{x_0}\sum^{\infty}_{j=1}2^{-j}x_j$, where $x_0\in\{0,1\}$ is the sign bit, and $x_1$ is the most significant bit. The sign bit of $0$ is $0$, so $(x_0,x_1,x_2,...)=(1,0,0,...)$ will never be used in this representation.

The notation MHE.Enc$(x)$ is used to refer to a bit-wise encryption MHE.Enc$(x_1,$\\$x_2, ...,x_m)$.
%Recall that $e$ is chosen from discrete Gaussian distribution $D_{Z^{m+1}_{q},\beta_{f}}$, and $s$ is chosen from uniform distribution on $Z^{n}_{q}$, denoted by $X_{Z^{n}_{q}}$.

\end{enumerate}

\subsection{Encrypted 1-bit Controlled Rotation}
Given encrypted 1-bit MHE.Enc($\alpha_{j}$), the goal is to implement the controlled rotation $R^{-\alpha_{j}}_{w}$ for fixed angle $w\in[0,1)$ on a general 1-qubit state $\bk{k}=k_0\bk{0}+k_1\bk{1}$. The basic idea follows \cite{mahadev2018classical}. Below, we present it in a brief but not very precise way. First, apply conditional operations with qubit $\bk{k}$ as control to create a superposition of the form:
\begin{align}\label{wq1}
\sum_{ l,u \in \{0,1 \} } e^{-2\mi\pi uw}{k_l}| l\rangle |u\rangle \bk{\text{AltMHE.Enc}(u \oplus l\alpha_j)}.
\end{align}
After measuring the last register of (\ref{wq1}) to obtain an encryption of some $u^{\star}\in\{0,1\}$, the resulting state is
\begin{align}\label{bmz}
\sum_{ l \in \{0,1 \} } e^{-2\mi\pi (u^{\star}\oplus l\alpha_j) w} {k_l}\bk{ l}\bk{ u^{\star}\oplus l\alpha_j } \bk{\text{AltMHE.Enc}(u^{\star})}.
\end{align}
So far, the rotation factor $e^{-2\mi\pi\alpha_jw}$ is introduced to the relative phase for $l=0$, $1$. After using Hadamard transform to eliminate the entanglement between the first two qubits of (\ref{bmz}), the resulting first qubit will be what we need. Details are as follows:

\begin{lem}\label{11}
Let \tn{MHE} and \tn{AltMHE} be Mahadev's scheme, and let $\lambda$ be the security parameter of \tn{MHE}. Suppose $\tn{MHE.Enc}(\zeta)$ is a ciphertext encrypting a 1-bit message $\zeta\in\{0,1\}$. For any angle $w\in[0,1)$, consider the conditional rotation $R^{\zeta}_{-w}$ whose control bit is $\zeta$. With parameters $m,n,q$ defined by Notation 1, there exists a quantum polynomial time algorithm that on input $w$, $\tn{MHE.Enc}(\zeta)$ and a general single-qubit state $|k\rangle$, outputs: (1) an \textnormal{AltMHE} encryption $y=\textnormal{AltMHE.Enc}(u^{\star}_0;r^{\star}_0)$, where $u^{\star}_0\in\{0,1\}$ and $r^{\star}_0 \in \Z^{m+n+1}_q$, (2) a bit string $d\in\{0,1\}^{1+(m+n+1)\log_2 q}$, and (3) a state within $\lambda$-negligible trace distance to
\begin{equation}\label{w12}
Z^{\langle d, (u_0^{\star},r_0^{\star})\bigoplus(u_1^{\star},r_1^{\star})\rangle}R^{u_0^{\star}\zeta}_{2w}R^{-\zeta}_{w}|k\rangle,
\end{equation}
where $(u_1^{\star},r_1^{\star})\in \{0,1\} \times \Z^{m+n+1}_q$ such that
\begin{align}
\emph{AltMHE.Enc}(u_0^{\star};r_0^{\star})=\emph{AltMHE.Enc(}u_1^{\star};r_1^{\star}) \oplus \emph{MHE.Convert(MHE.Enc(}\zeta)).
\end{align}
\end{lem}

\begin{remark} \tn{ In the realization of the encrypted controlled rotation $R^{-\zeta}_{w}$,\\ $Z^{d\cdot((u_0^{\star},r_0^{\star})\bigoplus(u_1^{\star},r_1^{\star}))}$ and $R^{u_0^{\star}\zeta}_{2w}$ are both serving to protect the privacy of $\zeta$ in (\ref{w12}), where the former is a Pauli mask, and the latter is to be removed later.}
\end{remark}

\begin{proof} We prove the lemma by providing a BQP algorithm---Algorithm 1 below. Recall in Notation 1 the parameters $q,m,n,\beta_f$. In Algorithm 1, Step 1 requires to create a superposition on discrete Gaussian distribution $D_{Z^{m+1}_{q},\beta_{f}}$, a typical procedure that can be found in Lemma 3.12 of [Reg05], or (70) in \cite{mahadev2018classical}. Then by creating superposition on discrete uniform distribution $\mathbb{U}_{Z_{2}}\times \mathbb{U}_{Z^{n}_{q}}$, adding an extra register $\bk{0}_G$ whose label is $G$, and using AltMHE.Enc in the computational basis, one can efficiently prepare (\ref{32}) in Algorithm 1.

After applying Step 2, by equality $R_{-\omega}\bk{u}=e^{-2\pi\mi\omega u}\bk{u}$ for $u=0,1$, the resulting state is $(\ref{cx12})$. In step 4, after adding an extra qubit initially in state $\bk{k}$ to the leftmost of the qubit system in (\ref{cx12}), then applying conditional homomorphic XOR, the resulting state is (\ref{1}). By Lemma \ref{41}, the following distributions are $\lambda$-negligibly close to each other:
\begin{align}\label{qz1}
&\{\textnormal{AltMHE.Enc}(\mu;r)|(\mu,r)\xleftarrow{}  (\mathbb{U}_{\Z^{n}_{q}},D_{\Z^{m+1}_{q},\beta_{f}}) \} \quad  \text{and}\\
&\{\textnormal{AltMHE.Enc}(\mu;r) \oplus \textnormal{MHE.Convert(MHE.Enc(}\zeta)) |(\mu,r) \xleftarrow{} (\mathbb{U}_{\Z^{n}_{q}},D_{\Z^{m+1}_{q},\beta_{f}})\}.\nonumber
\end{align}

So, in Step 5, after measuring register G of (\ref{1}), the measurement outcome is with overwhelming probability of the form $\text{AltMHE.Enc}(u^{\star}_0;r^{\star}_0 \large{)}$ where $u^{\star}_0\in\{0,1\},r^{\star}_0 \large{} \in \Z^{n}_q\times\Z^{m+1}_{\beta_{f}}$. The resulting state is
\begin{align}\label{zxx2}
\left( e^{-2\pi \mi w u^{\star}_0}\sqrt{\delta(r_{0}^{\star})}k_{0}|0\rangle|u^{\star}_0\rangle_{M}|r_{0}^{\star}\rangle_{M}+e^{-2\pi \mi w u^{\star}_1}\sqrt{\delta(r_{1}^{\star})}k_{1}|1\rangle|u^{\star}_1\rangle_{M}|r_{1}^{\star}\rangle_{M} \right)& \nonumber\\
&\hspace{-2.6cm}\Big{|}\text{AltMHE.Enc}(u^{\star}_0;r^{\star}_0 \big{)} \Big{\rangle_G},
\end{align}
where $\delta$ is the density function of distribution $\mathbb{U}_{Z^{n}_{q}}\times D_{Z^{m+1}_{q},\beta_{f}}$, $u_1^{\star}$ and $r_1^{\star}$ satisfy
\begin{align}\label{zx11}
\text{AltMHE.Enc}(u_0^{\star};r_0^{\star})=\text{AltMHE.Enc(}u_1^{\star};r_1^{\star}) \oplus \text{MHE.Convert(MHE.Enc(}\zeta)).
\end{align}
 Below, we show that ($\ref{zxx2}$) can be written as
\begin{align}\label{cz1}
\sum_{j\in\{0,1\}} e^{-2\pi \mi w u^{\star}_j}k_{j}|j\rangle|u^{\star}_j,r_{j}^{\star}\rangle_M|\text{AltMHE.Enc}(u^{\star}_0;r^{\star}_0 \big{)}\rangle_G.
\end{align}

By item 5 of Notation 1, one can assume MHE.Enc$(\zeta)$=MHE.Enc$(\zeta;(\s',\e'))$ where $\s'\in\Z^{n}_q$, $||\e'||_{\infty} \leq \beta_{init}(N+1)^{\eta_c}$. By (\ref{zx11}), $r_0^{\star}=r_1^{\star}+(\s',\e') \mod q$. Let     $r^{\star}_0=(\s^{\star}_0,\e^{\star}_0)$, and let $\rho_0$ be the density function of $D_{Z^{m+1}_{q},\beta_{f}}$, then
\begin{align}\delta(r^{\star}_0)=\frac{1}{q^{n}}\rho_0(\e^{\star}_0), \quad  \quad \delta(r^{\star}_1)=\delta \Big{(} r^{\star}_0-(\s',\e')\mod q  \Big{)}=\frac{1}{q^{n}}\rho_0(\e^{\star}_0-\e'),\end{align}
where the last equality makes use of $||\e^{\star}_0||_{\infty},$ $||\e'||_{\infty}\leq \beta_f \ll q$. Notice that $\e^{\star}_0$ is obtained by measuring $G$ in (\ref{1}), and can be viewed as being sampled from $D_{Z^{m+1}_{q},\beta_{f}}$ with probability $|k_0|^2$ (when $j=0$ in $G$ of (\ref{1})), and being sampled from $\e'+D_{Z^{m+1}_{q},\beta_{f}}$ with probability $|k_1|^2$ (when $j=1$). Applying Lemma \ref{61} to the following states (by substituting into (\ref{nw1}): $c_0=e^{-2\pi \mi w u^{\star}_0}k_{0}$, $c_1= e^{-2\pi \mi w u^{\star}_1}k_{1}$, $\omega=\e^{\star}_0$, $\rho_1(\e^{\star}_0)=\rho_0(\e^{\star}_0-\e')$), where $\bk{c'}$ is unnormalized,
\begin{align}
\bk{c}=&  e^{-2\pi \mi w u^{\star}_0}k_{0}|0\rangle+e^{-2\pi \mi w u^{\star}_1}k_{1}|1\rangle,\\
\bk{c'}=& e^{-2\pi \mi w u^{\star}_0}\sqrt{\delta(r_{0}^{\star})}k_{0}|0\rangle+e^{-2\pi \mi w u^{\star}_1}\sqrt{\delta(r_{1}^{\star})}
 k_{1}|1\rangle\nonumber\\
 =&\frac{\sqrt{\rho_0(\e^{\star}_0)}}{\sqrt{q^{n}}}e^{-2\pi \mi w u^{\star}_0}k_{0}\bk{0}+\frac{\sqrt{\rho_0(\e^{\star}_0-\e')}}{\sqrt{q^{n}}} e^{-2\pi \mi w u^{\star}_1}k_{1}\bk{1} ,
\end{align}
one gets that normalized $\bk{c'}$ is with overwhelming probability within $\text{negl}(\lambda)$ trace distance to $\bk{c}$. Observe that the first qubit of (\ref{zxx2}) is in state $\bk{c'}$, so the normalized state of (\ref{zxx2}) is with overwhelming probability within $\text{negl}(\lambda)$ trace distance to:
\begin{align}\label{22}
&\left(e^{-2\pi \mi w u^{\star}_0}k_{0}|0\rangle|u^{\star}_0,\ r_{0}^{\star}\rangle_{M}+e^{-2\pi \mi w u^{\star}_1}k_{1}|1\rangle|u^{\star}_1,\ r_{1}^{\star}\rangle_{M}\right)|\text{AltMHE.Enc}(u^{\star}_0;r^{\star}_0 \big{)}\rangle_G \nonumber \\
=&\sum_{j\in\{0,1\}} e^{-2\pi \mi w u^{\star}_j}k_{j}|j\rangle|u^{\star}_j,r_{j}^{\star}\rangle_M|\text{AltMHE.Enc}(u^{\star}_0;r^{\star}_0 \big{)}\rangle_G.
\end{align}
Now, (\ref{22}) can be taken as the result after step 5.

Let $Q=\begin{bmatrix} e^{-2\mi \pi wu^{\star}_0} & \\ & e^{-2\mi \pi wu^{\star}_1}\end{bmatrix}$. After applying Step 6, since for any $x\in Z_2,\ y\in Z^{n+m+1}_q$ and $p=1+(m+n+1)\log q$, performing bitwise Hadamard transform on the $p$-qubit state $\bk{x,y}$ will yield a state $\sum\limits_{ d\in\{0,1\}^{p}}(-1)^{\langle d,(x,y) \rangle}\bk{d}$, the resulting state is
\begin{align}\label{z21}
&(-1)^{\langle d,(u^{\star}_0,r^{\star}_0)\rangle} \big{(}k_0Q\bk{0}\big{)}|d\rangle_M|\tn{AltMHE.Enc}(u^{\star}_0;r^{\star}_0 \big{)}\rangle_G+\nonumber\\
&\hspace{5cm}(-1)^{\langle d,(u^{\star}_1,r^{\star}_1)\rangle}\big{(} k_1Q\bk{1}\big{)}|d\rangle_M|\tn{AltMHE.Enc}(u^{\star}_0;r^{\star}_0 \big{)}\rangle_G\nonumber\\
&=\big{(} Z^{ \langle d,  (u_0^{\star}, r_0^{\star})\bigoplus(u_1^{\star},r_1^{\star}) \rangle } Q |k\rangle  \big{)} |d\rangle_M|\text{AltMHE.Enc}(u^{\star}_0;r^{\star}_0 \big{)}\rangle_G.
\end{align}
By (\ref{zx11}), $u^{\star}_1,u^{\star}_0,\zeta\in\{0,1\}$ satisfy $u^{\star}_1=u^{\star}_0+\zeta \mod 2$. If $u^{\star}_0$, $\zeta$ are both $1$, then $u^{\star}_1=u^{\star}_0+\zeta$, otherwise $u^{\star}_1=u^{\star}_0+\zeta-2$. So, it holds that $u^{\star}_1-u^{\star}_0=\zeta-2\zeta u^{\star}_0$, and
\begin{equation}\label{zcq}
Q=e^{-2\mi \pi wu^{\star}_0}\left[\begin{array}{cc}   % 该矩阵一共3 列，每一列都居中放置
 1& \\
 & e^{-2\mi \pi w (\zeta-2\zeta u^{\star}_0 )}
  \end{array}
\right]=e^{-2\mi \pi wu^{\star}_0}R^{u_0^{\star}\zeta}_{2w}R^{-\zeta}_{w}.
 \end{equation}
By combining (\ref{z21}) and (\ref{zcq}), the resulting state in Step 6 is as claimed in (\ref{412}), whose first qubit is as in (\ref{w12}).

 $\hfill\blacksquare$
\end{proof}
%The encrypted trapdoor information enables us to use $d$, $HE.Enc_{pk}(s)$ and $AltMHE.Enc_{pk}(u^{\star}_0;r^{\star}_0 \large{)}$ to compute $HE.Enc_{pk'}(u^{\star}_0s)$ and the encryption of Pauli key $HE.Enc_{pk'}( d\cdot((u_0^{\star},r_0^{\star})\bigoplus(u_1^{\star},r_1^{\star}))\large{)}$. These undesirable $R^{u_0^{\star}s}_{2w}$ errors can be removed by homomorphically computing. More specifically,

\vspace{0.5cm}
    \begin{breakablealgorithm}\label{20}
        \caption{Double Masked 1-bit Controlled Rotation}\label{1123}
        \begin{algorithmic}[1] %每行显示行号
            \Require An angle $w\in[0,1)$, an encryption of one-bit message $\text{MHE.Enc}(\zeta)$, a single-qubit state $|k\rangle=k_0|0\rangle+k_1\bk{1}$; public parameters $\lambda,q,m,n,\beta_f$ in Notation 1.
            \Ensure A ciphertext $y=\tn{AltMHE.Enc}(u^{\star}_0;r^{\star}_0)$, a bit string $d$, and a state $\bk{\psi}=$\Statex $Z^{d\cdot((u_0^{\star},r_0^{\star})\bigoplus(u_1^{\star},r_1^{\star}))}R^{u_0^{\star}\zeta}_{2w}R^{-\zeta}_{w}|k\rangle$.
            \State Create the following superposition over the distribution $\mathbb{U}_{Z_{2}}\times \mathbb{U}_{Z^{n}_{q}}\times D_{Z^{m+1}_{q},\beta_{f}}$
            \begin{equation}\label{32}
\sum_{u\in\{0,1\}, r \in Z^{m+n+1}_q} \sqrt{\frac{\delta(r)}{2}}|u\rangle|r\rangle|\text{AltMHE.Enc}(u;r)\rangle_{G},
\end{equation} where $G$ is the label of the last register, and $\delta$ is the density function of distribution $\mathbb{U}_{Z^{n}_{q}}\times D_{Z^{m+1}_{q},\beta_{f}}$.
\State Perform phase rotation $R_{-\omega}$ on qubit $\bk{u}$ of (\ref{32}). The result is \begin{equation}\label{cx12}
\sum_{u\in\{0,1\}, r \in Z^{m+n+1}_q} \sqrt{\frac{\delta(r)}{2}}e^{-2\pi\mi w u }|u\rangle|r\rangle|\text{AltMHE.Enc}(u;r)\rangle_{G},
\end{equation}
\State Convert the input ciphertext MHE.Enc$(\zeta)$ into MHE.Convert(MHE.Enc$(\zeta)$).
\State Apply conditional homomorphic XOR to register $G$, with the control condition being that the single-qubit $|k\rangle=\sum_{j\in\{0,1\}}k_j|j\rangle$ is in state $\bk{1}$. The resulting state is \begin{align}\label{1}
 \sum_{j,u\in\{0,1\}, r \in  Z^{m+n+1}_q}& \sqrt{\frac{\delta(r)}{2}} k_j e^{-2\pi \mi w u}|j\rangle|u\rangle|r\rangle\ \nonumber\\
  &\big{|} \text{AltMHE.Enc}(u;r)\oplus j\cdot\text{MHE.Convert(MHE.Enc}(\zeta)) \big{\rangle_{G}}.
\end{align}
\State Measure register $G$. The outcome is with overwhelming probability of the form \Statex $\text{AltMHE.Enc}(u^{\star}_0;r^{\star}_0 \big{)}$, where $u^{\star}_0\in\{0,1\},\ r^{\star}_0 \in \Z^{n}_q\times\Z^{m+1}_{\beta_{f}}$. After the measurement, state (\ref{1}) becomes the following (unnormalized) state:
\begin{align}\label{2}
\left( e^{-2\pi \mi w u^{\star}_0}\sqrt{\delta(r_{0}^{\star})}k_{0}|0\rangle|u^{\star}_0,\ r_{0}^{\star}\rangle_{M}+e^{-2\pi \mi w u^{\star}_1}\sqrt{\delta(r_{1}^{\star})}k_{1}|1\rangle|u^{\star}_1,\ r_{1}^{\star}\rangle_{M} \right)& \nonumber\\
&\hspace{-3cm}|\text{AltMHE.Enc}(u^{\star}_0;r^{\star}_0 \big{)}\rangle_G,
\end{align}
where $M$ is the label of the middle register.

\State Perform qubit-wise Hadamard transform on register $M$ of (\ref{2}), then measure register $M$. Suppose a bit string $d\in\{0,1\}^{1+(m+n+1)\log q}$ is the state of register $M$ after the measurement. The resulting state is
\begin{align}\label{412}
\left( Z^{ \langle d, (u_0^{\star},r_0^{\star})\bigoplus(u_1^{\star},r_1^{\star})\rangle} R^{u_0^{\star}\zeta}_{2w}R^{-\zeta}_{w}  |k\rangle  \right)|d\rangle_M|\text{AltMHE.Enc}(u^{\star}_0;r^{\star}_0 \big{)}\rangle_G.
\end{align}
\State Set $\bk{\psi}$ to be the first qubit of (\ref{412}).
                   \end{algorithmic}
    \end{breakablealgorithm}

\begin{remark} \tn{The main difference of Algorithm 1 from Mahadev's encrypted CNOT operation (cf. Claim 4.3 in \cite{mahadev2018classical}) comes from (\ref{1}), where an entangled state is created in a different way.}
\end{remark}

%In general, the process of Algorithm 1 is similar to that of Mahadev's encrypted CNOT algorithm (cf. Claim 4.3 of \cite{mahadev2018classical}).

%\footnote{In fact, the procedures of Algorithm 1 are similar to Mahadev's encrypted CNOT algorithm. The main difference comes form (\ref{1}), where we create the entangled state in a different way for our purpose. For the readers who want to compare the two in detail, we strongly recommend to read the paper \cite{mahadev2018classical}.}

%Combining with the public encrypted trapdoors of MHE, we can improve the result of the above lemma, as follows:

\begin{prop}\label{345}
Let $pk_1$ and $pk_2$ be two public keys generated by $\tn{MHE.Keygen}$, and let $t_1$ be the trapdoor of public key $pk_1$. Using an encrypted control bit $\henc{pk_1}(\zeta)$ and the encrypted trapdoor $\henc{pk_2}(t_1)$, for any angle $w\in[0,1)$, any single-bit state $\bk{k}$, one can efficiently prepare (1) state of the form $Z^{d_1}R^{d_2}_{2w}R^{-\zeta}_{w}\bk{k}$, where $d_1,d_2\in\{0,1\}$ are random parameters obtained by quantum measurement, and (2) a ciphertext $\henc{pk_2}(d_1,d_2)$.
\end{prop}

\begin{proof}
By Lemma \ref{11}, it suffices to show how to produce the ciphertext $\henc{pk_2}(d_1,d_2)$ using the encrypted trapdoor $\henc{pk_2}(t_1)$. Recall from Lemma \ref{4200} that trapdoor $t_1$ allows the random vector $r$ and plaintext $u$ to be recovered from a ciphertext $\aenc{pk_1}($\\$u;r)$.
\vspace{0.15cm}

We begin with the output of Lemma \ref{11}. First, encrypt the output $\aenc{pk_1}($\\$u^{\star}_0;r^{\star}_0)$ with the MHE scheme using the public key $pk_2$. This together with the encrypted trapdoor $\henc{pk_2}(t_1)$ gives the encryptions $\henc{pk_2}(u^{\star}_0)$ and $\henc{pk_2}(r^{\star}_0)$. Update $\henc{pk_1}(\zeta)$ to $\henc{pk_2}(\zeta)$ by the encrypted trapdoor (Fact 4. in the end of Section \ref{sec2.4}). By homomorphic multiplication between $\henc{pk_2}(u^{\star}_0)$ and $\henc{pk_2}(\zeta)$, we get $\henc{pk_2}(u^{\star}_0\zeta)=\henc{pk_2}(d_2)$. Similarly, we can obtain the ciphertext $\henc{pk_2}(d_1)$, where $d_1=d\cdot((u_0^{\star},r_0^{\star})\bigoplus\big(u_1^{\star},r_1^{\star})\big)$ with parameters $u^{\star}_1,\ r^{\star}_1,\ d$ described in Lemma \ref{11}.
\end{proof}

%More specifically, by homomorphic evaluations on the encrypted trapdoor $\henc{pk_2}(t_1)$, the ciphertext $\aenc{pk_1}(u^{\star}_0;r^{\star}_0)$ and $\henc{pk_1}(\zeta)$, we can compute $\aenc{pk_2}(u^{\star}_1)$ and $\aenc{pk_2}(r^{\star}_1)$, according to (\ref{zx11}). This together with the measurement result, $d$, of Lemma \ref{11} gives $\henc{pk_2}(d_1)$, where $d_1=d\cdot((u_0^{\star},r_0^{\star})\bigoplus\big(u_1^{\star},r_1^{\star})\big)$.$\hfill\blacksquare$

\subsection{Encrypted Conditional Rotation}
%Let $\alpha$ be an $m$-bit binary angle $\alpha=\sum^{m}_{j=1}2^{-j}\alpha_j$ where $\alpha_j\in\{0,1\}$. We outline the realization of the encrypted-CROT, which is to implement $R_{\alpha}$ when given MHE.Enc$(\alpha)$, instead of $\alpha$. We first use MHE.Enc$(\alpha_m)$, the encryption of the least significant bit of $\alpha$, to perform conditional rotation $R^{-1}_{2^{-m}\alpha_m}$ by Algorithm 1. Although the output state contains undesired error $R^{d}_{2^{-m+1}}$ where $d\in\{0,1\}$, as to be shown in the following theorem, this error can be cancelled via homomorphic evaluation on the output MHE.Enc$(d)$ and the remaining encrypted bits (i.e., by removing the encrypted least significant bit) of MHE.Enc$(\alpha)$. By iteration, we can realize the conditional rotation $R_{\alpha}$ with the angle given in bitwise encrypted form:

\begin{enumerate}
\item[]
\textbf{$x$ mod 1.} For any $x \in \R$, $x\mod1$ refers to a real number $x'$ in range $[0, 1)$ such that $x'=x \mod 1$.
\end{enumerate}

\begin{thm}[Encrypted conditional rotation]\label{12}
Let angle $\alpha\in[0,1)$ be represented in $m$-bit binary form as $\alpha=\displaystyle\sum^{m}_{j=1}2^{-j}\alpha_{j}$ for $\alpha_{j}\in\{0,1\}$. Let $pk_i$ be the public key with trapdoor $t_i$ generated by \textnormal{MHE.Keygen} for $1\leq i\leq m$. Suppose the encrypted trapdoor $\textnormal{MHE.Enc}_{pk_{j+1}}(t_{j})$ is public for $1\leq j\leq m-1$. Given the bitwise encrypted angle $\textnormal{MHE.Enc}_{pk_1}(\alpha)$ and a single-qubit state $|k\rangle$, one can efficiently prepare a ciphertext $\textnormal{MHE.Enc}_{pk_{m}}(d)$, where random parameter $d\in\{0,1\}$, and a state within $\lambda$-negligible trace distance to
\begin{equation}\label{23d}
Z^{d}R^{-1}_{\alpha}|k\rangle.
\end{equation}
\end{thm}
\begin{proof}
We first prove the theorem for $m=1$. i.e., $\alpha=\alpha_1/2$. Note $R_{1/2}=Z$, so $Z^{\alpha_1}=R_{\alpha_1/2}$. We rewrite single-qubit state $\bk{k}$ as: $\bk{k}=Z^{\alpha_1}R^{-1}_{\alpha_1/2}\bk{k}$. Since $\henc{pk_{1}}(\alpha_1)$ is given in the input, the theorem automatically holds for $m=1$ by setting $d=\alpha_1$.

We prove the theorem for $m\geq2$ by providing a BQP algorithm in Algorithm \ref{311} below. Notice that $\alpha_m$ is the least significant bit of $\alpha$. In step 1 of Algorithm \ref{311}, by the procedure given in Proposition \ref{345}, on input the encrypted trapdoor $\henc{pk_2}(t_1)$, an encrypted 1-bit $\henc{pk_1}(\alpha_m)$, and a single-bit state $\bk{k}$, one obtains two encrypted single bits $\henc{pk_2}(d_1,b_1)$, where $d_1,b_1\in\{0,1\}$, and a state
\begin{align}\label{3.2622}
\bk{v_1}=Z^{d_1}R^{b_1}_{2^{-m+1}}R^{-\alpha_m}_{2^{-m}}|k\rangle=Z^{d_1}R^{b_1}_{2^{-m+1}}R^{-1}_{\alpha_m2^{-m}}|k\rangle.
\end{align}

%and the encryption \text{HE.Enc}_{pk_2}(w_1)$, $\text{MHE.Enc}_{pk_2}(w_2)$ can also be obtained by using encrypted trapdoor information $\text{MHE.Enc}_{pk_2}(t_1)$.\\

To remove the undesired operator $R^{b_1}_{2^{-m+1}}$ in (\ref{3.2622}), first use encrypted trapdoor $\henc{pk_2}(t_1)$ to the public key $pk_1$ of MHE.Enc$_{pk_1}(\alpha$) to get $m-1$ encrypted bits $\text{MHE.Enc}_{pk_2}(\alpha_j)$ for $1\leq j \leq m-1$, i.e., a bitwise encryption of angle $\sum^{m-1}_{j=1}\alpha_j2^{-j}$ in Step 2. Then, in Step 3, update this encrypted ($m-1$)-bit angle by evaluating a multi-bit addition (modulo 1) on it:
\begin{equation}
\alpha^{(1)}=\sum^{m-1}_{j=1}\alpha_j2^{-j}+ b_12^{-m+1} \mod 1,
\end{equation}
The result, $\henc{pk_2}(\alpha^{(1)})$, is a bitwise encryption of ($m-1$)-bit binary angle $\alpha^{(1)}\in[0,1)$.\\

If $m=2$, now that $\alpha^{(1)}$ only includes 1-bit: $\alpha^{(1)}=\tn{LSB}(\alpha^{(1)})=\alpha_1\oplus b_1$, the state of (\ref{3.2622}) can be written as $Z^{d_1+b_1+\alpha_1}R^{-1}_{\alpha_1/2+\alpha_2/4}\bk{k}$. By
\begin{equation}
\text{MHE.Enc}_{pk_2}\big(d_1\oplus \tn{LSB}(\alpha^{(1)}) \big)=\text{MHE.Enc}_{pk_2}(d_1)\oplus \text{MHE.Enc}_{pk_2}\big(\tn{LSB}(\alpha^{(1)})\big),
\end{equation}
the theorem holds by setting $d=d_1\oplus\tn{LSB}(\alpha^{(1)})$.

If $m\geq3$, the iteration procedure (Steps 4.2-4.4) is similar to Steps 1-3. In Step 4.2, the angle of $R^{b_{l+1}}_{2^{l+1-m}}$ becomes larger and larger with the increase of $l$, eventually reaching $1/2$ for $l=m-2$. At that time, the undesired operator $R^{b_{m-1}}_{2^{-1}}=Z^{b_{m-1}}$ becomes a Pauli mask. (\ref{5134}) in Step 4.2 can be proved by induction on $l$: For $l=1$, after applying controlled rotation $R^{\tn{LSB}(\alpha^{(l)})}_{-2^{1-m}}$ on $\bk{v_1}$, by Algorithm \ref{20} and (\ref{3.26}), the resulting state is
\begin{align}\label{kaz1}
\bk{v_2}=Z^{d_2}R^{b_2}_{2^{2-m}}R^{-1}_{\tn{LSB}(\alpha^{(1)})2^{1-m}}\bk{v_1}=Z^{d_2+d_1}R^{b_2}_{2^{2-m}}R^{-1}_{(\tn{LSB}(\alpha^{(1)})-b_1)2^{1-m}+\tn{LSB}(\alpha)2^{-m}}\bk{k}.
\end{align}
By agreeing that $d_1=d'_1$, $\alpha=\alpha^{(0)}$, (\ref{kaz1}) is just (\ref{5134}). For $l\geq 2$, after applying controlled rotation $R^{\tn{LSB}(\alpha^{(l)})}_{-2^{l-m}}$ to $\bk{v_{l}}$, the resulting state is
\begin{align}\label{bz1}
\bk{v_{l+1}}&=Z^{d_{l+1}}R^{b_{l+1}}_{2^{l+1-m}}R^{\tn{LSB}(\alpha^{(l)})}_{-2^{l-m}} \bk{v_{l}} \nonumber\\
&= Z^{d_{l+1}+d_{l}+d'_{l-1}}R^{b_{l+1}}_{2^{l+1-m}}R^{-1}_{\tiny{ \sum^{m}_{j=m-l}(\tn{LSB}(\alpha^{(m-j)})-b_{m-j})2^{-j}} } |k\rangle.
\end{align}
By (\ref{cx6}), $d_l+d'_{l-1}=d'_l$ for $l \geq 2$. So, (\ref{bz1}) becomes (\ref{5134}).

Below, we show that the state $\bk{v_{m-1}}$ obtained in (\ref{5134}) is just $Z^{d'_{m-1}+\tn{LSB}\big(\alpha^{(m-1)}\big)}R^{-1}_{\alpha} |k\rangle$. By the expressions of $\alpha^{(l)}$ from (\ref{t1}), (\ref{t2}), the following equality holds:
\begin{equation}
\begin {array}{cccccccc}
&0.&\alpha_1&...&  \alpha_{m-2} & \alpha_{m-1}& \alpha_m\\
+&0.&b_{m-1}&...& b_{2} & b_{1}&  \\
\hline
= &0.&\tn{LSB}(\alpha^{(m-1)}) & ... &\tn{LSB}(\alpha^{(2)}) & \tn{LSB}(\alpha^{(1)})& \alpha_m&
\end{array} \mod 1 ,
\end{equation}
namely,
\begin{align}\label{vx1}
\alpha=\sum^{m}_{j=1}\alpha_j 2^{-j}= \sum^{m}_{j=1}(\tn{LSB}(\alpha^{(m-j)})-b_{m-j})2^{-j} \mod  \quad 1,
\end{align}
where $\alpha^{(0)}=\alpha$, $b_0=0$. Therefore, by $d'_{m-1}=d_{m-1} \oplus d'_{m-2}$,

\begin{align}
\begin{array}{lcl}
|v_{m-1}\rangle&=&Z^{d_{m-1}+d'_{m-2}}R^{b_{m-1}}_{1/2}R^{-1}_{\tiny{ \sum^{m}_{j=2}(\tn{LSB}(\alpha^{(m-j)})-b_{m-j})2^{-j}} } |k\rangle\\
&=&Z^{d'_{m-1}}R^{\tn{LSB}\big(\alpha^{(m-1)}\big)}_{1/2}R^{b_{m-1}-\tn{LSB}\big(\alpha^{(m-1)}\big)}_{1/2}R^{-1}_{\tiny{ \sum^{m}_{j=2}(\tn{LSB}(\alpha^{(m-j)})-b_{m-j})2^{-j}} } |k\rangle\\
&=&Z^{d'_{m-1}+\tn{LSB}\big(\alpha^{(m-1)}\big)}R^{-1}_{\tiny{\sum^{m}_{j=1}(\tn{LSB}(\alpha^{(m-j)})-b_{m-j})2^{-j}} } |k\rangle\\
&\overset{(\ref{vx1})}{=}&Z^{d'_{m-1}+\tn{LSB}\big(\alpha^{(m-1)}\big)}R^{-1}_{\alpha}|k\rangle.
\end{array}
\end{align}
The theorem holds by setting $d=d'_{m-1}+\tn{LSB}\big(\alpha^{(m-1)}\big)$ in $(\ref{23d})$.
%Apply Steps 1,2 iteratively for the new encrypted angle $\henc{pk_2}(\alpha^{(1)})$, as shown in Algorithm \ref{311}. The output is exactly what is required in the theorem statement.$\hfill\blacksquare$

 $\hfill\blacksquare$
\end{proof}

\begin{enumerate}  \item[]
\vspace{0.2cm}
 \begin{breakablealgorithm}
  \hspace{-0.7cm}  \caption{Encrypted Conditional Rotation}\label{311}
        \begin{algorithmic}[0] %每行显示行号
  \State \hspace{-1cm} \textbf{Input:}\  Encrypted trapdoors $\henc{pk_{j+1}}(t_j)$ for $1\leq j \leq m-1$, an encrypted $m$-bit angle \\
  \hspace{0.15cm} $\henc{pk_1}(\alpha)$, and a single-bit state $\bk{k}$.
  \State \hspace{-1cm} \textbf{Output:} \  A state $\bk{v_{m-1}}=Z^{d}R^{-1}_{\alpha}|k\rangle$, and an encrypted bit $\text{MHE.Enc}_{pk_{m}}(d)$.
   \State \hspace{-1cm}1:\quad Use Proposition \ref{345} and ciphertext $\henc{pk_1}(\alpha_m)$ to get two encrypted single bits \\$\henc{pk_2}(d_1,b_1)$, where $d_1,b_1\in\{0,1\}$, and a state
\begin{align}\label{3.26}
\bk{v_1}=Z^{d_1}R^{b_1}_{2^{-m+1}}R^{-1}_{\alpha_m2^{-m}}|k\rangle           .
\end{align}

 \State \hspace{-1cm}2:\quad  Use encrypted trapdoor $\henc{pk_2}(t_1)$ to MHE.Enc$_{pk_1}(\alpha$) to get $m-1$ encrypted bits \\$\text{MHE.Enc}_{pk_2}(\alpha_j)$ for $1\leq j \leq m-1$.

  \State \hspace{-1cm}3:\quad Use $\text{MHE.Enc}_{pk_2}(\alpha_j)$ ($1\leq j \leq m-1$) and $\text{MHE.Enc}_{pk_2}(b_1)$ to get an encryption of $(m-1)$-bit angle $\text{MHE.Enc}_{pk_2}(\alpha^{(1)})$, where
  \begin{equation}\label{t1}
  \alpha^{(1)}=\sum^{m-1}_{j=1}\alpha_j2^{-j}+b_12^{-m+1} \mod 1.
  \end{equation}

   \State \hspace{-1cm}4:\quad \textbf{if} $m=2$ \textbf{then}

   \State \hspace{-0.3cm}4.1:\quad Homomorphically evaluate the XOR gate on $\text{MHE.Enc}_{pk_2}(d_1)$ and \\ $\text{MHE.Enc}_{pk_2}(\alpha^{(1)})$ to get
   $$\text{MHE.Enc}_{pk_2}(d)=\text{MHE.Enc}_{pk_2}(d_1\oplus\tn{LSB}(\alpha^{(1)})).$$

   \State  \quad \hspace{-0.6cm}\textbf{else}

   \State \hspace{0.35cm} \textbf{for} $l$ from $1$ to $m-2$ \textbf{do}:
           \State \begin{enumerate}  \item[]
           \hspace{-0.1cm}4.2:\hspace{0.2cm} By Algorithm \ref{20}, use as control bit the encrypted least significant bit of $\text{MHE.Enc}_{pk_{l+1}}(\alpha^{(l)})$ to realize the controlled rotation $R^{\tn{LSB}(\alpha^{(l)})}_{-2^{l-m}}$ on state $\bk{v_{l}}$. The result is two encrypted bits $\henc{pk_{l+2}}(d_{l+1},b_{l+1})$, where $d_{l+1},b_{l+1}\in\{0,1\}$, and a state of the form
           \begin{align}\label{5134}
            \bk{v_{l+1}}=Z^{d_{l+1}+d'_{l}}R^{b_{l+1}}_{2^{l+1-m}}R^{-1}_{\tiny{ \sum^{m}_{j=m-l}(\tn{LSB}(\alpha^{(m-j)})-b_{m-j})2^{-j}} } |k\rangle,
            \end{align} \hspace{0.6cm} where $d'_1=d_1$, $\alpha^{(0)}=\alpha$, $b_0=0$.\end{enumerate}
        \State \hspace{0.5cm}4.3:\hspace{0.2cm} Set \begin{equation}\label{t2}
\alpha^{(l+1)}=\sum^{m-l-1}_{j=1}\alpha^{(l)}_j2^{-j} + b_{l+1}2^{1+l-m} \mod 1.
\end{equation}
\hspace{1.2cm} Homomorphically compute the encryption of ($m-l-1$)-bit angle\\ \hspace{1.2cm} $\text{MHE.Enc}_{pk_{l+2}}(\alpha^{(l+1)})$.

\State \hspace{0.5cm}4.4:\hspace{0.2cm} Set $d'_{l+1}=d_{l+1}\oplus d'_{l}$. Homomorphically compute
\begin{equation}\label{cx6}
\henc{pk_{l+2}}(d'_{l+1})=\henc{pk_{l+2}}(d_{l+1}\oplus d'_{l}).
\end{equation}
 \State \hspace{0.4cm} \textbf{end for}
 \State \hspace{-0.1cm} 4.5:\quad Set $d=d'_{m-1}+\tn{LSB}\big(\alpha^{(m-1)}\big)$. Homomorphically compute $\text{MHE.Enc}_{pk_{m}}(d)$.

    \State \hspace{-0.7cm} \quad \textbf{end if}
                   \end{algorithmic}
    \end{breakablealgorithm}

\end{enumerate}

\begin{remark} \label{cw1}\tn {The encrypted conditional $P$-gate, i.e, $R_{x/4}$ with the control bit $x\in\{0,1\}$ given in encrypted form, can be implemented (up to Pauli error) using Algorithm \ref{311} by setting $\alpha=0/2+x/4$, and because  $Z^{d}R^{-1}_{\frac{x}{4}}=Z^{d}R^{-1}_{\frac{x}{2}}R_{\frac{x}{4}}=Z^{d-x}R^{x}_{\frac{1}{4}}$. It makes the QHE of \cite{broadbent2015quantum} a QFHE scheme.
Specifically, to evaluate non-Clifford gate $T$, one can directly perform $T$ on the ciphertext $X^{a}Z^{b}\bk{\psi}$, and then perform $P^a$ by the encrypted-CROT.
By $TX^{a}Z^{b}=P^{a}X^{a}Z^{b}T$, the above sequence of operations yield
$$X^{a}Z^{b}\bk{\psi}\xlongrightarrow{T} P^{a}X^{a}Z^{b}T\bk{\psi} \xlongrightarrow{\tn{encrypted-CROT}}  Z^{d-a}P^{a}P^{a}X^{a}Z^{b}T\bk{\psi}=Z^{d+b}X^{a}T\bk{\psi},$$
and the encrypted Pauli keys can be updated by homomorphic arithmetics on \\ MHE.Enc$(a,b,d)$. Now a ``Clifford+T''-style QFHE is obtained.}
%As the corollary of Theorem $\ref{12}$, a BQP server is also able to execute the encryption-controlled operation corresponding to this kind of rotation:
\end{remark}

Theorem \ref{12} implies that with Enc($\alpha$) at hand, one can apply $U(-\alpha,0,0)$ to a quantum state up to a Pauli-$Z$ operator. The following corollary shows how to make use of Enc($\alpha$) to implement $U(0,-\alpha,0)$, i.e., $T_{-\alpha}$ as defined in (\ref{45}).

\begin{cor}\label{4} Consider an angle $\alpha\in[0,1)$ represented in $m$-bit binary form as $\alpha=\displaystyle\sum^{m}_{j=1}2^{-j}\alpha_{j}$, where $\alpha_{j}\in\{0,1\}$. Let $pk_i$ be the public key with trapdoor $t_i$ generated by \textnormal{MHE.Keygen} for $1\leq i\leq m$. Suppose the encrypted trapdoor $\emph{MHE.Enc}_{pk_{j+1}}(t_{j})$ is public for $1\leq j\leq m-1$. Given the bitwise encrypted angle $\emph{MHE.Enc}_{pk_1}(\alpha)$ and a general single-qubit state $|k\rangle$, one can efficiently prepare (within $\lambda$-negligible trace distance) the following state:
\begin{equation}\label{231}
Z^{d}X^{d}T^{-1}_{\alpha}|k\rangle,
\end{equation}
as well as a ciphertext $\textnormal{MHE.Enc}_{pk_{m}}(d)$, where random parameter $d\in\{0,1\}$ depends on quantum measurement.
\end{cor}

\begin{proof}
Let \begin{equation}
S=\frac{1}{\sqrt{2}}\left[               %左括号
  \begin{array}{cc}
1 & 1\\
\mi & -\mi
  \end{array}
\right].
\end{equation}
Then for any $\alpha\in[0,1)$,
\begin{equation}\label{xc1}
T_{\alpha}=e^{-\mi\pi\alpha}SR_{\alpha}S^{-1}.
\end{equation}
To prepare $T^{-1}_{\alpha}\bk{k}$ up to Pauli operator, first act $S^{-1}$ on $|k\rangle$. Then by Theorem $\ref{12}$, use MHE.Enc$_{pk_1}(\alpha)$ to prepare
\begin{equation}\label{3}
Z^{d}R^{-1}_{\alpha}S^{-1}|k\rangle.
\end{equation}
Finally, act $S$ on (\ref{3}) to get (\ref{231}) (after ignoring a global phase factor), because
\begin{equation}
SZ^{d}R^{-1}_{\alpha}S^{-1}\bk{k}=(-\mi)^{d}e^{-\mi\pi\alpha}Z^{d}X^{d}T^{-1}_{\alpha}\bk{k},
\end{equation}
where the equation is by combining (\ref{xc1}) and the fact that for any $d\in\{0,1\}$,
\begin{equation}
SZ^{d}=(-\mi)^{d}Z^{d}X^{d}S.
\end{equation}
$\hfill\blacksquare$
\end{proof}

\subsection{Encrypted Conditional Unitary Operator on Single Qubit}

The following is the main result of this paper:

\begin{thm}\label{75}
Let $m$-bit binary fractions $\alpha,\beta,\gamma\in[0,1)$ be the Euler angles of a $2\times 2$ unitary $U$, that is, $U=R_{\alpha}T_{\beta}R_{\gamma}$. Let $pk_i$ be the public key with trapdoor $t_i$ generated by \textnormal{MHE.Keygen} for $1\leq i\leq 3m$. Suppose the encrypted trapdoor $\emph{MHE.Enc}_{pk_{j+1}}(t_{j})$ is public for $1\leq j\leq 3m-1$. Given the ciphertexts \emph{MHE.Enc}$_{pk_1}(\alpha,\beta,\gamma)$ and a general one-qubit state $|k\rangle$, one can efficiently prepare ciphertexts $\emph{MHE.Enc}_{pk_{3m}}(d_1,d_2)$, where random parameters $d_1,d_2\in\{0,1\}$, and a state within $\lambda$-negligible trace distance to
\begin{equation}
 Z^{d_1}X^{d_2}U^{-1}|k\rangle.
\end{equation}
\end{thm}

\begin{proof}
We prove the theorem by providing a BQP algorithm in Algorithm \ref{411} below. By Theorem 3.3, in step 1 of Algorithm \ref{411}, by performing an encrypted conditional phase rotation $R^{-1}_{\alpha}$ on state $|k\rangle$, one obtains an encrypted bit $\textnormal{MHE.Enc}_{pk_m}(w_1)$, where $w_1\in\{0,1\}$, and a state
\begin{align}\label{zz99}
Z^{w_1}R^{-1}_{\alpha}\bk{k}=Z^{w_1}R^{-1}_{\alpha}(R_{\alpha}T_{\beta}R_{\gamma})(R_{\alpha}T_{\beta}R_{\gamma})^{-1}\bk{k}&=Z^{w_1}T_{\beta}R_{\gamma}U^{-1}|k\rangle\nonumber\\
&= T_{(-1)^{w_1}\beta}Z^{w_1}R_{\gamma}U^{-1}|k\rangle,
\end{align}
where the last equality comes from $T_{\beta}Z=ZT_{-\beta}$.

In step 2, to prepare the ciphertext MHE.Enc$_{pk_m}((-1)^{w_1}\beta \mod 1)$, the algorithm first homomorphic evaluates the bitwise XOR of MHE.Enc$_{pk_m}(\sum^{m}_{j=1}2^{-j}w_1)$ and MHE.Enc$_{pk_m}(\beta)$, then homomorphically adds MHE.Enc$_{pk_m}$\\$(w_12^{-m})$ to the result.

In step 3, by applying encrypted controlled rotation $T^{-1}_{(-1)^{w_1}\beta \mod 1}$ to the state (\ref{zz99}), and using the relations, up to a global phase factor, $ZX=XZ$ and $T_{-\beta \mod 1}=T_{-\beta}$ for any $\beta\in[0,1)$, one gets that
\begin{align}\label{zz6}
Z^{w_2}X^{w_2}T^{-1}_{(-1)^{w_1}\beta}T_{(-1)^{w_1}\beta}Z^{w_1}R_{\gamma}U^{-1}|k\rangle=X^{w_2}Z^{w_1+w_2}R_{\gamma}U^{-1}|k\rangle.
\end{align}

In step 4, since for any $\gamma\in[0,1)$, it holds that $R_{\gamma}X=e^{2\pi\mi\gamma}XR_{-\gamma}$, $R_{-\gamma \mod 1}=R_{-\gamma}$ and $R_\gamma Z=Z R_\gamma$, the result of performing encrypted phase rotation $R^{-1}_{(-1)^{w_2}\gamma}$ on (\ref{zz6}) is:
\begin{align}
Z^{w_3}R^{-1}_{(-1)^{w_2}\gamma}X^{w_2}Z^{w_1+w_2}R_{\gamma}U^{-1}|k\rangle=X^{w_2}Z^{w_1+w_2+w_3}U^{-1}|k\rangle.
\end{align}\\
The ciphertext MHE.Enc$_{pk_{3m}}(w_2)$ can be produced by using MHE.Enc$_{pk_{2m}}(w_2)$ and encrypted trapdoors \\  MHE.Enc$_{pk_{j+1}}(t_j)$ ($2m \leq j \leq 3m-1$). The ciphertext MHE.Enc$_{pk_{3m}}(w_1\oplus w_2\oplus w_3)$ is obtained by applying homomorphic XOR operators on MHE.Enc$_{pk_{3m}}(w_1,w_2,w_3)$. The theorem holds by setting $d_1=w_2$ and $d_2=w_1\oplus w_2\oplus w_3$. $\hfill\blacksquare$
\end{proof}

\vspace{0.5cm}
    \begin{breakablealgorithm}
        \caption{Encrypted Conditional Unitary Operator on Single Qubit}\label{411}
        \begin{algorithmic}[1] %每行显示行号
            \Require Encrypted trapdoors $\henc{pk_{j+1}}(t_j)$ for $1\leq j \leq 3m-1$, encrypted $m$-bit Euler angles
            \Statex $\henc{pk_1}(\alpha,\beta,\gamma)$, and a single-qubit state $\bk{k}$.
            \Ensure Two encrypted bits $\tn{MHE.Enc}_{pk_{3m}}(d_1,d_2)$, and a state $Z^{d_1}X^{d_2}U(\alpha,\beta,\gamma)^{-1}|k\rangle$.
            \State Perform the MHE.Enc$_{pk_1}(\alpha)$-controlled phase rotation $R^{-1}_{\alpha}$ on state $|k\rangle$. The result is a state
\begin{align}\label{99}
T_{(-1)^{w_1}\beta}Z^{w_1}R_{\gamma}U^{-1}|k\rangle,\end{align}
together with the encrypted Pauli-key $\textnormal{MHE.Enc}_{pk_m}(w_1)$, where $w_1\in\{0,1\}$.
           \State Use MHE.Enc$_{pk_m}(w_1, \beta)$ to get ciphertext MHE.Enc$_{pk_m}((-1)^{w_1}\beta \mod 1)$ by homomorphic computation.
           \State Apply Corollary \ref{4} to (\ref{99}) with the encrypted angle MHE.Enc$_{pk_m}((-1)^{w_1}\beta \mod 1)$. The output is a state
\begin{align}\label{6}
X^{w_2}Z^{w_1+w_2}R_{\gamma}U^{-1}|k\rangle,
\end{align}
together with an encryption MHE.Enc$_{pk_{2m}}(w_2)$, where $w_2\in\{0,1\}$.
\State  Apply the MHE.Enc$_{pk_{2m}}((-1)^{w_2}\gamma \mod 1)$-controlled encrypted phase rotation \Statex $R^{-1}_{(-1)^{w_2}\gamma \mod 1}$ on state (\ref{6}). The result is an encrypted bit MHE.Enc$_{pk_{3m}}(w_3)$, and a state
\begin{align}
X^{w_2}Z^{w_1+w_2+w_3}U^{-1}|k\rangle.
\end{align}
\State Set $d_1=w_2$, $d_2=w_1+w_2+w_3$. Homomorphically compute MHE.Enc$_{pk_{3m}}(w_1\oplus w_2\oplus w_3)$.
                   \end{algorithmic}
    \end{breakablealgorithm}

\section{The Components of Our QFHE Scheme}\label{secc4}
\subsection{Quaternion one-time pad Encryption (QOTP)}\label{sec3.1}
\begin{enumerate}
\item[]
\setlength{\itemsep}{7pt}

\vspace{0.3cm}
\textbf{$k$-bit representation of unitary operator.} Given a unitary $U_{\vec{t}}$ where $\vec{t}\in \mS^3$, let $\vec{t'}\in\R^4$, whose elements in binary form are the sign bit and the $k$ most significant bits in the binary representation of the corresponding elements of $\vec{t}$. We call $U_{\vec{t'}}$ \emph{the $k$-bit finite precision representation of unitary} $U_{\vec{t}}$. Note that $U_{\vt'}$ is only a linear operator, not a unitary one.

\vspace{0.3cm}
\textbf{Unitary approximation of $k$-bit precision linear operator.} Given a linear operator $U_{\vec{t}}$, where each element of $\vec{t}\in \R^4$ is a $k$-bit binary fraction, the \emph{unitary approximation of} $U_{\vt}$ is $U_{\vt'}$, where $\vt'\in\mS^3$ is defined by (\ref{q1}), (\ref{q2}) in Lemma \ref{2.2}, such that $||U_{\vec{t}}-U_{\vec{t}'}||_2\leq\sqrt{3 \big{|} \nrm{\vt}-1 \big{|}}$ when $\big{|}\nrm{\vec{t}}-1\big{|}$ is small.
\vspace{0.3cm}

\end{enumerate}

%\footnote{Starting from a $4$-dimensional vector $\vec{t'}=\vec{0}$, assign values $t'_i=t_i$ for $i$ from $1$ to $4$, one by one, as much as possible until $\sum^{4}_{i=1}|t'_i|^2=1$, cf. (\ref{q1}). If $\nrm{\vt}\leq1$, modify $t'_4$ to meet $\nrm{\vec{t'}}=1$, cf. (\ref{q2}).}

We use the following quaternion one-time pad method to encrypt a single qubit, and encrypt a multi-qubit state qubitwise.

\begin{itemize}
\setlength{\itemsep}{5pt}
\item \textbf{Quaternion one-time pad encryption of a single qubit message}
\item[$\bullet$] QOTP.Keygen($k$). Sample three classical $k$-bit binary fractions $(h_1,h_2,h_3)$ uniformly at random, where $h_i \in[0,1)$ and $\sum^{3}_{i=1}h_i^2 \leq1$. Compute a $k$-bit binary fraction approximate to $\sqrt{1-\sum^{3}_{i=1} h^2_i}$, and denote it by $h_4$. Output $(t_1,t_2,t_3,t_4)$, which is a random permutation of $(h_1,h_2,h_3,h_4)$ followed by multiplying each element with 1 or $-1$ of equal probability. Notice that $\sum^{4}_{i=1}t^2_i\neq1$ in general.
\item[$\bullet$] QOTP.Enc($(t_1,t_2,t_3,t_4)$, $\bk{\phi}$). Apply the unitary approximation of linear operator $U_{(t_1,t_2,t_3,t_4)}$ to single-qubit state $\bk{\phi}$ and output
the resulting state $|\hat{\phi}\rangle$.
\item[$\bullet$] QOTP.Dec($(t_1,t_2,t_3,t_4)$, $|\hat{\phi}\rangle$). When $\sum^{4}_{i=1}|t_i|^2=1$, apply the inverse, $U_{(t_1,-t_2,-t_3,-t_4)}$, of $U$ to $|\hat{\phi}\rangle$. If $\sum^{4}_{i=1}|t_i|^2 \neq 1$, apply the unitary approximation of $U_{(t_1,-t_2,-t_3,-t_4)}$ to $|\hat{\phi}\rangle$.
\end{itemize}

%\footnote{In short, QOTP.Keygen is to sample four $k$-bit binary numbers at random, such that the sum of squares of them is within $2/2^k$ to 1. }

%The information-theoretic security of our encryption scheme is guaranteed by
%\begin{lem}\emph{(Information-theoretic security)}
%Let $\vec{t}=(t_i)$ be a $4$-dimensional vector whose each element is represented by $k$-bits as $t_i=\sum^{k}_{j=1}2^{-j}t_{ij}$, where $1\leq i\leq4$ and $t_{ij}\in\{ 0,1\}$. Let $U_{\vec{t}}$ be defined as in Definition %\ref{2.1}. For any single-qubit system with a density matrix $\rho$,
%\begin{equation}
%\frac{1}{2^{4k}}\sum_{t\in \{0,1\}^{4k}} U_{\vec{t}}\rho U^{-1}_{\vec{t}}=\frac{\I_2}{2}
%\end{equation}
%\end{lem}

The motivation for introducing QOTP is that our evaluation algorithms require the use of multi-bit pad keys, while the previous Pauli-OTP encryption only provides pad keys of 1-bit. We extend Pauli-OTP encryption to a version of multi-bit pad keys. Noting that $U_{(1,0,0,0)}=Z^0 X^0$, $U_{(0,1,0,0)}= X$, $U_{(0,0,1,0)}= Z$, $U_{(0,0,0,1)}= XZ$, Pauli-OTP encryption is an extreme case of QOTP encryption with a bit size $k=1$. Our QFHE also allows to use Pauli-OTP for encryption first, and then expand the key length during evaluations.

The following lemma guarantees the information-theoretic security of the QOTP encryption scheme.
\begin{lem}\label{2.5}\emph{(Information-theoretic security)}
Let $M$ be the set of all possible output vectors of the QOTP.Keygen, and let the probability of outputting vector $\vec{t}\in\R^4$ be $p(\vec{t})$, where the elements of $\vec{t}$ are $k$-bit binary fraction. For any single-qubit system with density matrix $\rho$,
\begin{equation}\label{4.1}
\sum_{\vec{t} \in M} p(\vec{t}) U_{\vec{t}}\rho U^{-1}_{\vec{t}} = \frac{\I_2}{2}.
\end{equation}

\end{lem}

%\footnote{To see the correctness of (\ref{332}), notice that if $g^{*}(\vt_1)=\vt_2$, where $\vt_1,\vt_2 \in M$, $g^{*}\in S_4$, supposing $\vt_1$ is generated by applying permutation to some $\vec{h}\in M$ in QOTP.Keygen, then it holds that $\tn{card} (\{g|g(\vec{h})=\vec{t_1},g\in S_4\})=\tn{card}( \{g|g(\vec{h})=\vec{t_2},g\in S_4\} )$, resulting in p$(\vt_1)$=p$(\vt_2)$.}

\begin{proof}
Let $S_4$ be the 4-th order symmetric group. From the symmetry in generating $\vec{t}=(t_1,t_2,t_3,t_4)$, we get that the probability function $p$ satisfies:
\begin{align}
&p(\vec{t})=p(\vec{t'}), \quad \forall\ \vt,\ \vt'\in M\ s.t. \ (|t_1|,|t_2|,|t_3|,|t_4|)=(|t'_1|,|t'_2|,|t'_3|,|t'_4|),\label{331} \\
&p(\vec{t})=p(g(\vec{t}) ), \quad \forall\ g \in S_4,\ \vec{t}\in M \label{332}.
\end{align}
It is not difficult to verify that for any matrix $A=\begin{bmatrix*}[r]a_{11} &\ a_{12} \\ a_{21} &\ a_{22} \end{bmatrix*}\in\C^{2\times2}$,
\begin{align}\label{313}
\frac{1}{4}\sum_{ a,b\in\{0,1\}}X^{a}Z^{b}AZ^{-b}X^{-a}=\frac{1}{2}\sum_{ a\in\{0,1\}}X^{a}\begin{bmatrix*}[r]a_{11} & \\ &\ a_{22} \end{bmatrix*} X^{-a}=\frac{tr(A)}{2}\I_2.
\end{align}
where $X=\mi\sigma_1$, $Z=\mi\sigma_2$ are Pauli matrices.
By (\ref{331}), (\ref{332}) and (\ref{341}), for any $a,b\in\{0,1\}$, any matrix $A\in\C^{2\times2}$,
\begin{align}\label{323}
\sum_{\vt\in M}p(\vt) X^{a}Z^{b} U_{\vt}AU^{-1}_{\vt} Z^{-b}X^{-a}&=\sum_{\vt\in M}p(\vt) \sigma_1^{a}\sigma_2^{b} U_{\vt}A (\sigma_1^{a}\sigma_2^{b} U_{\vt})^{-1} \\
&= \sum_{\vt\in M}p(\vt) U_{g_{\sigma^a_1,\sigma^b_2}(\vt)}AU^{-1}_{g_{\sigma^a_1,\sigma^b_2}(\vt)}\\
&= \sum_{\vt\in M}p(\vt) U_{\vt}AU^{-1}_{\vt}, \label{zxx1}
\end{align}
where $g_{\sigma^a_1,\sigma^b_2}(\vec{t})=(t'_1,t'_2,t'_3,t'_4)$ such that:
\begin{align}
\sigma^a_1\sigma^b_2(t_1+t_2\sigma_1+t_3\sigma_2+t_4\sigma_3)=t'_1+t'_2\sigma_1+t'_3\sigma_2+t'_4\sigma_3.
\end{align}
Combining (\ref{313}) and (\ref{zxx1}) gives
\begin{align}
\sum_{t\in M} p(\vt) U_{\vec{t}}\rho U^{-1}_{\vec{t}}&=\frac{1}{4}\sum_{\substack{ \vec{t} \in M\\ a,b\in\{0,1\}}} p(\vt) X^{a}Z^{b}U_{\vec{t}} \rho U^{-1}_{\vec{t}}Z^{-b}X^{-a}\nonumber\\
&=\frac{1}{4}\sum_{ a,b\in\{0,1\}}X^{a}Z^{b}\left(\sum_{\vec{t} \in M}p(\vt)U_{\vec{t}} \rho U^{-1}_{\vec{t}}\right)Z^{-b}X^{-a}\nonumber\\
&=tr(\displaystyle\sum_{\vec{t} \in M}p(\vt) U_{\vec{t}} \rho U^{-1}_{\vec{t}}) \frac{\I_2}{2}=\frac{\I_2}{2},
\end{align}
where the last equality follows from
\begin{align}
tr( U_{\vt} \rho U^{-1}_{\vec{t}})=tr(\rho)=1, \qquad  \qquad  \forall \vt \in M.
\end{align} $\hfill\blacksquare$
\end{proof}

A common confusion is why the length of pad key is irrelevant to the security. Particularly, it seems impossible to securely encrypt 1-qubit whose amplitudes store a continuous number of information by just a finite-length pad key (even 1-bit pad key of Pauli-OTP). This intuition ignores an important quantum nature: extracting information from the amplitude of a 1-qubit is hard. Specifically, each 1-qubit collapses after a single measurement with an outcome $\bk{0}$ or $\bk{1}$. To extra the information in amplitudes, numerous copies of the same 1-qubit are required to be measured.

The concrete meaning behind (\ref{4.1}) is that the encrypted numerous copies of arbitrary 1-qubit (with each copy encrypted by independently random pads) form the same system, which can not be distinguished by measurements on any basis. So, in the current quantum setting, QOTP is as secure as Pauli-OTP, regardless of how long the pad key is.

From the pure mathematical model lens, QOTP has a potential advantage that it can provide security against super adversaries who can directly read the amplitude of qubits with a single measurement \cite{walborn2006experimental}, whereas Pauli-OTP must fail to be secure against such an adversary.

The following lemma guarantees that the decryption of a ciphertext by QOTP.Dec is correct up to negl$(k)$ $L^2$-distance, where $k$ is the number of bits for representation.

\begin{lem}\label{2.22}\rm{(Correctness)}
Given a unitary operator $U_{\vec{t}}$ where $\vt\in \R^4$ and $||\vec{t}||_2=1$, let $U_{\vec{t'}}$ be the $k$-bit finite precision representation of $U_{\vec{t}}$, and let $U_{\vec{t''}}$ be the unitary approximation of linear operator $U_{\vec{t'}}$. Then
$$ \nrm{  U_{\vec{t''}}- U_{\vec{t}}  } \leq \nrm{  U_{\vec{t'}}- U_{\vec{t}}}\leq{\frac{4}{\sqrt{2^{k}}}}.$$
\end{lem}
\begin{proof}
Let $t_i$ be the $i$-th coordinate of vector $\vec{t}$ for $1\leq i\leq4$. Now that $\vt'$ is the $k$-bit approximation of $\vt$,
\begin{align}
\Big|\ \nrm{\vec{t'}}-\nrm{\vt}\Big|\leq\nrm{\vt'-\vt}\leq \sqrt{\frac{4}{2^{2k}}}=\frac{2}{2^{k}}.
\end{align}
By Lemma \ref{2.2} where $m=\frac{2}{2^{k}}\geq \Big|\nrm{\vt'}-1\Big|$, we have $\nrm{\vt'-\vt''}\leq\frac{\sqrt{6}}{\sqrt{2^{k}}}$, and
\begin{align}
 \nrm{  U_{\vec{t}}- U_{\vec{t''}} } &\leq  \nrm{  U_{\vec{t'}}- U_{\vec{t}} }+ \nrm{  U_{\vec{t''}}- U_{\vec{t'}} } \leq \frac{2}{2^{k}}+\frac{\sqrt{6}}{\sqrt{2^{k}}} \leq \frac{4}{\sqrt{2^{k}}}.
\end{align} $\hfill\blacksquare$
\end{proof}

\subsection{Homomorphic Evaluation of Single-qubit Gates}\label{sec4.2}
Single qubit gates and the CNOT gate are a set of universal quantum gates. We show below how the server evaluates a single-qubit quantum gate homomorphically.

In our QFHE scheme, the server receives a ciphertext that is composed of a quantum message encrypted by QOTP, together with the (classical) QOTP key (\textbf{called the gate key}) encrypted by MHE.
Let the encrypted gate key held by the server be Enc$(\vec{t})$, where $\vec{t}=(t_1,t_2,t_3,t_4)$ is a vector whose elements are $k$-bit binary fractions.

To evaluate a unitary gate whose $k$-bit precision representation is $U_{\vec{k}}$, the server needs to use Enc$(\vec{k})$ and Enc$(\vec{t})$ to compute a new ciphertext Enc$(\vec{t'})$ that satisfies $U_{\vec{t'}}=U_{\vec{t}}U^{-1}_{\vec{k}}$, where $\vec{t}'$ is a 4-dimensional vector whose elements are $k$-bit binary fractions. This can be done by homomorphic computation, according to (\ref{2023}) and (\ref{21}). The ciphertext Enc$(\vec{t'})$ serves as the new encrypted gate key for the next round of evaluation.

%\remark{From (\ref{21}), we can see that the degree of representation-error increases linearly in the number of evaluated gates. Thus, we can correctly decrypt ciphertexts after evaluating polynomial single-bit gates.}

\subsection{Homomorphic Evaluation of the CNOT Gate}\label{zx1}

%\subsubsection{Terminology}
%We use the term `\textbf{Pauli-encrypted state}' to refer to the quantum state encrypted using Pauli one-time pad, and call the secret pad \textbf{Pauli-key}. We use the term `\textbf{QOTP-encrypted state}' to refer to the quantum state encrypted using QOTP, and call the secret encryption key \textbf{gate key}.
\begin{itemize}
\item[]
\textbf{CNOT$_{1,2}$ operation.} For a two-qubit state $\bk{\psi}$, the notation $U\otimes V \bk{\psi}$ refers to performing $U$ on the first qubit of $\bk{\psi}$, and performing $V$ on the second qubit. The notation $\rm{CNOT}_{1,2}$ denotes a CNOT operation with the first qubit as the control and the second qubit as the target.
\vspace{0.3cm}
\end{itemize}

To evaluate the CNOT gate, we first change a QOTP-encrypted state into a Pauli-encrypted state, and output the encryptions of Pauli-keys. Then, we evaluate the CNOT gate on the Pauli-encrypted state by the following relation:
\begin{equation}\label{13}
\tn{CNOT}_{1,2} (X^{a_1}Z^{b_1}\otimes X^{a_2}Z^{b_2})|\psi\rangle= (X^{a_1}Z^{b_1+b_2}\otimes X^{a_2+a_1}Z^{b_2})\tn{CNOT}_{1,2}|\psi\rangle,
\end{equation}
where $|\psi\rangle$ is a two-qubit state, and $a_i, b_i\in\{0,1\}$ are the Pauli keys of the $i$-th qubit for $i=1,2$.

%(see Appendix C in \cite{broadbent2015quantum})
Now, we show how to transform a QOTP-encrypted state into its Pauli-encrypted version. First, with the encrypted gate key MHE.Enc$(\vec{t})$ at hand, one can homomorphically compute Euler angles for unitary operator $U_{\vec{t}}$ according to relation (\ref{w2}). The detailed procedure involves several lemmas, all of which are moved to the Appendix \ref{7.3}. Then, by the encrypted conditional rotation technique, one can transform a QOTP-encrypted state into its Pauli-encrypted version, by the following proposition:
\begin{prop}\label{65}
Let $\vt\in\mS^3$. Given a 1-qubit state $U_{\vec{t}}\bk{\psi}$ and an encrypted gate key MHE.Enc$(\vec{t}')$, where $\vec{t}'\in\R^4$ such that $||\vt'-\vt||_{\infty}=\tn{negl}($k$)$, one can prepare a state within $k$-negligible trace distance to Pauli-encrypted state $Z^{d_1}X^{d_2}\bk{\psi}$, together with the encrypted Pauli keys MHE.Enc$(d_1,d_2)$, where random parameters $d_1,d_2\in \{0,1 \}$.
\end{prop}
\begin{proof}
Let $\alpha,\beta,\gamma\in[0,1)$ be defined as in (\ref{w2}), such that $U(\alpha,\beta,\gamma)\xlongequal{\tn{i.g.p.f}}U_{\vt}$. Now that $||\vt'-\vt||_{\infty}=$negl($k$). By Lemma \ref{6zq} in the Appendix, from MHE.Enc$(\vec{t}')$ one can produce a ciphertext MHE.Enc\\$(\alpha',\beta',\gamma')$, such that
\begin{equation}\label{bv1}
\nrm{U(\alpha',\beta',\gamma')-e^{\mi \delta}U_{\vec{t}}} =\tn{negl}(k),  \quad \tn{where $e^{\mi \delta}$ is a global phase factor.}
\end{equation}
Theorem \ref{75} allows to use the encrypted angles MHE.Enc$(\alpha',\beta',\gamma')$ to perform
\begin{equation}\label{zxc1}
U_{\vec{t}}\bk{\psi} \rightarrow Z^{d_1}X^{d_2}U(\alpha',\beta',\gamma')^{-1}U_{\vec{t}}\bk{\psi},
\end{equation}
and meanwhile get two encrypted bits MHE.Enc$(d_1,d_2)$. Below, we prove that\\ $U(\alpha',\beta',\gamma')^{-1}U_{\vec{t}}\bk{\psi}$ is within $k$-negligible trace distance from $\bk{\psi}$. First,
\begin{align}\label{zcq1}
||U(\alpha',\beta',\gamma')^{-1}U_{\vec{t}}\bk{\psi}-e^{-\mi \delta} \bk{\psi}||_{H}&= \frac{\sqrt{2}}{2}\nrm{ e^{\mi \delta}U(\alpha',\beta',\gamma')^{-1}U_{\vec{t}}\bk{\psi}-\bk{\psi}} \nonumber\\
& \leq \frac {\sqrt{2}}{2}\nrm{ e^{\mi \delta}U_{\vec{t}}-U(\alpha',\beta',\gamma')} =\tn{negl}(k).
\end{align}
By (\ref{s3}), the trace distance of two states is invariant under a global phase scaling to one of the states, so
\begin{align}
||U(\alpha',\beta',\gamma')^{-1}U_{\vec{t}}\bk{\psi}-\bk{\psi}||_{tr}=||U(\alpha',\beta',\gamma')^{-1}U_{\vec{t}}\bk{\psi}-e^{-\mi \delta}\bk{\psi}||_{tr}=\tn{negl}(k).
\end{align}
 $\hfill\blacksquare$
\end{proof}

%one can produce the ciphertext MHE.Enc$(\alpha,\beta,\gamma)$ such that $\nrm{U(\alpha,\beta,\gamma)-U_{\vec{t'}}}= \tn{negl}(k)$ with O$(k)$ calls of homomorphic multiplications and additions on encryption Enc$(\vec{t'})$. Then, The encrypted control technique (Theorem \ref{75}) allows to use the encrypted angles MHE.Enc$(\alpha,\beta,\gamma)$ and a QOTP-encrypted state $U_{\vec{t}}\bk{\psi}$ ($\bk{\psi}$ is some 1-qubit plaintext) to get a state within $k$-negligible trace distance from some Pauli-encrypted state $X^{d_1}Z^{d_2}\bk{\psi}$, where $d_1,d_2\in \{0,1 \}$, and the encrypted Pauli keys MHE.Enc$(d_1,d_2)$.

%\footnote{For a bit-by-bit encryption, the evaluation of function `max' and `if' are accessible.}.

%The proof is similar to using Chebyshev polynomial to approximate the convolution of truncated ·arccos’ and some smooth cut-off function.
%See lemma shows there exist a poly($\lambda$)-degree circuits to approximate the `arccos' to precision $\textnormal{negl}(\lambda)$.

\begin{remark}\label{cw2}\tn{ The reason why we do not directly adopt the Euler representation of SU$(2)$ at the beginning, but rather use the quaternion representation, is that the latter provides an arithmetic circuit implementation of much smaller depth for the product in SU(2), as shown in (\ref{21}). This change of representation is not necessary in the real representation of quantum circuits and states (cf. \cite{aharonov2003simple}, Lemma 4.6 of \cite{kitaev1997quantum}), where all the involved 1-qubit quantum gates are in SO(2), and in that case, the rotation representation: $\begin{bmatrix}\cos2\pi\alpha & -\sin2\pi\alpha\\ \sin2\pi\alpha & \cos2\pi\alpha \end{bmatrix}$ where $\alpha\in[0,1)$, already provides a low-depth circuit implementation for the product in SO$(2)$.}
\end{remark}

\section{Our Quantum FHE Scheme}\label{secc5}

The design of our QFHE scheme follows the following guideline/idea:
\begin{itemize}
\setlength{\parsep}{6pt}

\item[1.] The client uses the QOTP scheme to encrypt a quantum state (the message), and then encrypts the gate keys with MHE scheme.
\vspace{0.5ex}
\item[2.] The client sends both the encrypted quantum state and the gate keys to the server, and also sends the server the following tools for homomorphic evaluation: the public keys, encrypted secret keys and encrypted trapdoors.
\vspace{0.5ex}
\item[3.] To evaluate a single-qubit gate, the server only needs to update the encrypted gate keys.
\vspace{0.5ex}
\item[4.] To evaluate a CNOT gate on an encrypted two-qubit state:
\vspace{0.5ex}

(4.1) The server first computes the encryptions of the Euler angles of the $2\times 2$ unitary gates that are used to encrypt the two qubits, by homomorphic evaluations on the gate keys.

\vspace{0.3ex}

(4.2) Then, the server applies the encrypted conditional rotations to obtain a Pauli-encrypted state, as well as the encrypted Pauli keys.
\vspace{0.3ex}

(4.3) The server evaluates the CNOT gate on the Pauli-encrypted state, and updates the encrypted Pauli keys according to (\ref{13}).
\vspace{0.3ex}

(4.4) By Lemma \ref{pro1} below, the resulting state in Pauli-encrypted form is (up to a global factor) in natural QOTP-encrypted form. It can be directly used in the next round of evaluation.
\vspace{0.5ex}

\item[5.] During decryption, the client first decrypts the classical ciphertext of the gate keys, then uses the gate keys to decrypt the quantum ciphertext received from the server.
\end{itemize}
\begin{lem}\label{pro1}
For any $x_1,x_2\in\{0,1\}$, any 1-qubit state $|\psi\rangle$, the Pauli-encrypted state $Z^{x_1}X^{x_2}|\psi\rangle$ can be written (up to a global factor) in QOTP-encrypted form as follows:
\begin{align}\label{213}
U_{((1-x_1)(1-x_2),x_2(1-x_1),x_1(1-x_2),-x_1x_2)}|\psi\rangle.
\end{align}
\end{lem}
\begin{proof}
Note that $(\mi Z)^{x_1}=U_{(1-x_1,0,x_1,0)}$ and $(\mi X)^{x_2}=U_{(1-x_2,x_2,0,0)}$. Then by (\ref{341}),
\begin{align}
Z^{x_1}X^{x_2}&=(-\mi)^{x_1+x_2}\Big((1-x_1)\I_2+x_1\sigma_2\Big)\Big((1-x_2)\I_2+x_2\sigma_1\Big)\\
&=(-\mi)^{x_1+x_2}U_{((1-x_1)(1-x_2),x_2(1-x_1),x_1(1-x_2),-x_1x_2)}
\end{align}
\end{proof}

\begin{itemize}
\item[] \textbf{Parameters to be used in the scheme:}
\begin{enumerate}
\item Assume the quantum circuit to be evaluated can be divided into $L$ levels, such that each level consists of serval single-qubit gates, followed by a layer of non-intersecting CNOT gates.
\item Let $L_c=L_m+L_s$, where $L_m=$ maximum depth of the quantum circuit composed of all the single-qubit gates in a level, $L_s=$ depth of the classical circuits on the encrypted gate key for homomorphically evaluating a CNOT gate. The MHE scheme is assumed to have the capability of evaluating any $L_c$-depth circuit.
\item Let $k$ be the number of bits used to represent the gate key, i.e, the parameter of QOTP.Keygen($\cdot$) in Section\ref{sec3.1}. A typical choice is $k=\log^2\lambda$, where $\lambda$ is the security parameter.
\end{enumerate}
\end{itemize}

\begin{itemize}
\setlength{\itemsep}{7pt}
\item \rm{\textbf{Our new QHE scheme}}

\item[$\bullet$] \textbf{QHE.KeyGen}($1^\lambda$, $1^L$, $1^k$):
\begin{enumerate}
\item For $1\leq i \leq 3kL+1$, let $(pk_i, sk_i, t_{sk_i}, evk_{sk_i})=\tn{MHE.Keygen}(1^\lambda,1^{L_c})$, where $t_{sk_i}$ is the trapdoor required for randomness recovery from the ciphertext.

\item The public key is $pk_1$, and the secret key is $sk_{3kL+1}$. The other public information includes $evk_{sk_i}$ for $1 \leq i \leq 3kL+1$, and $(pk_{i+1}$, MHE.Enc$_{pk_{i+1}}(sk_i)$, MHE.Enc$_{pk_{i+1}}(t_{sk_i}))$ for $1 \leq i \leq 3kL$.
\end{enumerate}

\item[$\bullet$]\textbf{QHE.Enc}$_{pk_1}(|\psi\rangle)$: Use QOTP to encrypt each qubit of $\bk{\psi}$; for any single-qubit state $\bk{v}$, its ciphertext consists of $U_{\vec{t}}|v\rangle$ and $\textnormal{MHE.Enc}_{pk_1}(\vec{t})$, where the 4$k$-bit gate key $\vec{t}=(t^{(1)}_h,t^{(2)}_h,t^{(3)}_h,t^{(4)}_h)\in\{0,1\}^{4k}$, $h=1,...,k$.\footnote{The gate keys for different qubits of $\bk{\psi}$ are generated independently.}

\item[$\bullet$] \textbf{QHE.Eval}:
 \begin{enumerate}
\item To evaluate a single-qubit unitary $U_{\vec{k}}$ on an encrypted qubit $U_{\vec{t}}|\psi_1\rangle$, only needs to update the encrypted gate key from $\textnormal{MHE.Enc}_{pk_j}(\vec{t})$ to $\textnormal{MHE.Enc}_{pk_j}(\vec{t'})$ where $\vt'=\vt \vec{k}^{-1}$ according to Section \ref{sec4.2}, for some $1\leq j \leq 3kL+1$.
\item To evaluate the CNOT gate on two encrypted qubits $U_{\vec{t_1}}\otimes U_{\vec{t_2}}|\psi_2\rangle$ with encrypted gate key $\textnormal{MHE.Enc}_{pk_j}(\vec{t_1},\vec{t_2})$:
 \begin{enumerate}
\item Compute the Euler angles of unitary operators $U_{\vec{t_1}},U_{\vec{t_2}}$ homomorphically, with the angles represented in $k$-bit binary form to approximate the unitary operators to precision negl$(k)$ (not necessarily $\frac{1}{2^k}$). Denote the encrypted Euler angle 3-tuples of $\vt_1$, $\vt_2$ in binary form by MHE.Enc$_{pk_j}(\vec{\alpha_1},\vec{\alpha_2})$, where $\vec{\alpha_1},\vec{\alpha_2}\in\{0,1\}^{3k}$.
\item Use the encrypted angles $\textnormal{MHE.Enc}_{pk_j}(\vec{\alpha_1},\vec{\alpha_2})$ to apply the corresponding encrypted conditional rotation to the input quantum ciphertext $U_{\vec{t_1}}\otimes U_{\vec{t_2}}|\psi_2\rangle$. The result is a Pauli-encrypted state. The MHE encryption of the Pauli key can also be obtained.
\item Evaluate the CNOT gate on the Pauli-encrypted state according to (\ref{13}). The resulting state is in QOTP-encrypted form, whose encrypted gate key can be computed by homomorphic evaluation according to (\ref{213}).
\end{enumerate}
\end{enumerate}
\item[$\bullet$]\textbf{QHE.Dec}$_{sk_{3kL+1}}(U_{\vec{t}}|\psi\rangle, \tn{MHE.Enc}_{pk_{3kL+1}}(\vec{t}))$: Decrypt the classical ciphertext\\ $\textnormal{MHE.Enc}_{pk_{3kL+1}}(\vec{t})$ to obtain the gate key $\vec{t}$, then apply $U^{-1}_{\vec{t}}=U_{(t_1,-t_2,-t_3,-t_4)}$ to the quantum ciphertext $U_{\vec{t}}|\psi\rangle$ to obtain the plaintext state.
\end{itemize}
\textbf{Leveld FHE property of our QHE scheme.}

\vspace{0.3cm}

We show that any choice of parameter $k$ that satisfies $\frac{1}{2^k}=\tn{negl}(\lambda)$ is sufficient to make the new QHE scheme leveled fully homomorphic.

\begin{thm}
The new \tn{QHE} scheme is a quantum leveled fully homomorphic encryption scheme, if parameter $k=O(\tn{poly}(\lambda))$ satisfies $\frac{1}{2^{k}}$=\tn{negl}$(\lambda)$.
\end{thm}
\begin{proof}

In the encryption step, we encrypt each qubit by using QOTP, with the $4k$-bit gate key encrypted by MHE in poly$(\lambda)$ time.

%(note that $k$ is at most poly$(\lambda)$ in a PPT algorithm),

Due to the $k$-bit finite representation of quantum gates, each evaluation of a single-qubit gate introduces a \emph{quantum error}, which is measured by the trace distance between the decrypted ciphertext and the correct plaintext. The following Proposition \ref{2376} guarantees that after evaluating poly$(\lambda)$ number of single-qubit gates, the quantum error is still $\lambda$-negligible.

%the resulting encrypted state can be decrypted correctly within a negl$(\lambda)$ trace distance.

To evaluate a CNOT gate, we first compute the encrypted approximate Euler angles for the $2\times 2$ unitary operator represented by the gate keys. By Lemma \ref{6zq}, this can be done in time \rm{poly}$(k)$ to get negl$(k)$-approximation. By encrypted conditional rotation (see Theorem \ref{75}), we can transform the quantum ciphertexts into Pauli-encrypted form, and then evaluate the CNOT gate on the Pauli-encrypted states.

Although the quantum error increases with the use of encrypted conditional rotations, the evaluation of the CNOT gate on Pauli-encrypted states (see (\ref{13})) is so simple that it does not cause any increase in the quantum error. Each encrypted conditional rotation requires O$(k)$ uses of Algorithm \ref{1123}, and the output state of Algorithm \ref{1123} is correct within negl$(\lambda)$ trace distance by Lemma \ref{11}.
So the quantum error for evaluating a CNOT gate is $\lambda$-negligible. After evaluating poly$(\lambda)$ number of CNOT gates, the quantum error is still negl$(\lambda)$. $\hfill\blacksquare$

\end{proof}
%O$(3k)$ times that of performing one encrypted CNOT operation of \cite{mahadev2018classical}, an operation similar to algorithm \ref{1123}. Recalling the parameters $\eta,N$ from Definition \ref{421}, then the error for evaluating a CNOT gate can be bounded by $O(\frac{3k}{(N+1)^{\eta}})=\text{negl}({\lambda})$. It follows that

\begin{lem}[Precision in $L^{2}$-norm]\label{ts1}
For any $2\times2$ unitary operator $U_{\vt}$, let $U_{\vt'}$ be the $k$-bit precision quaternion representation of $U_{\vt}$, then $||U_{\vt}-U_{\vt'}||_{2}\leq \frac{1}{2^{k-1.5}}$.
\end{lem}
\begin{proof}
By (\ref{c2}), from $||\vt-\vt'||_{\infty}\leq\frac{1}{2^k}$, one gets $||U_{\vt}-U_{\vt'}||_{\infty}\leq \big{|} \|\vt-\vt'\|_{\infty}+\mi \|\vt-\vt'\|_{\infty} \big{|} \leq\frac{\sqrt{2}}{2^k}$,
so $||U_{\vt}-U_{\vt'}||_{2}\leq 2||U_{\vt}-U_{\vt'}||_{\infty}\leq\frac{1}{2^{k-1.5}}$.
\end{proof}

\begin{lem}\label{234}
Let linear operators $U'_1,U'_2,...U'_m$ be the $k$-bit finite precision quaternion representations of $2\times 2$ unitary operators $U_1,U_2,..,U_m$ respectively, so that $||U_i-U'_i||_2 \leq \frac{1}{2^{k-1.5}}$ for $1\leq i\leq m$. Given a multi-qubit system $|\psi\rangle$, let $V_j$ ($V'_j$) be the multi-qubit gate on $|\psi\rangle$ that describes the acting of $U_j$ ($U'_j$) on some (the same) 1-qubit of $\bk{\psi}$ for $1\leq j \leq m$. If $m$ is a polynomial function in $\lambda$, and $k$ is a function in $\lambda$ such that $\frac{m}{2^{k}}=\textnormal{negl}(\lambda)$, then
\begin{align}\label{a12}
&\| V_m...V_1|\psi\rangle - V'_m...V'_1|\psi\rangle \|_{H}= \textnormal{negl}(\lambda),
\end{align}
\end{lem}

\begin{proof}
Observe that the matrix form of $V_j-V'_j$ is the tensor product of $2\times 2$ matrix $U_j-U'_j$ with some identity matrix. Thus, for any state $|\psi\rangle$, any $1\leq j\leq m$,
%Observe that the unitary matrix forms of $V_j$ and $V'_j$ are 2-sparse, that is, each of their columns and rows has at most 2 non-zero elements. Thus, for any state $|\psi\rangle$, any $1\leq j\leq m$,
\begin{align}\label{71}
\|V_j|\psi\rangle-V'_j|\psi\rangle\|_H= \frac{\sqrt{2}}{2}\|V_j|\psi\rangle-V'_j|\psi\rangle \|_2 \leq \frac{\sqrt{2}}{2}||V_j-V'_j||_2 = \frac{\sqrt{2}}{2} ||U_j-U'_j||_{2}\leq\frac{1}{2^{k-1}}.
\end{align}

Set $P_j=V_j....V_1$ for $ j=1,...,m$, and set $P'_j=V'_j...V'_1$. Since unitary operators $P_j$ are $L^2$-norm preserving, it holds that, for $j=2,...,m-1$,
\begin{align}
||P_j|\psi\rangle-P'_j|\psi\rangle||_H &\leq ||V_jP_{j-1}|\psi\rangle-V'_jP_{j-1}|\psi\rangle ||_H+||V'_jP_{j-1}|\psi\rangle-V'_jP'_{j-1}|\psi\rangle ||_H\\
&\leq \frac{1}{2^{k-1}}+  ||V'_j||_2  ||P_{j-1}|\psi\rangle-P'_{j-1}|\psi\rangle ||_H, \label{zz1}
\end{align}
where the last equality is by (\ref{71}) and the distance relation: $||\cdot||_{H}=\frac{\sqrt{2}}{2}||\cdot||_2$. Set $a_j=||P_j|\psi\rangle-P'_j|\psi\rangle||_H$, and set $M=\displaystyle\max_{1\leq j\leq m}||U_j-U'_j||_{2}$, then by (\ref{71}), $a_1\leq \frac{1}{2^{k-1}}$, $M\leq\frac{1}{2^{k-1.5}}\leq\frac{1}{2^{k-2}}$, and
\begin{align}\label{xx1}
||V'_j||_2=||U'_j||_2\leq||U_j||_2+||U_j-U'_j||_2 \leq 1+M \leq 1+\frac{1}{2^{k-2}},  \quad  \forall 1\leq j\leq m.
\end{align}

Combining (\ref{xx1}) and (\ref{zz1}) gives the following recursive relation on $a_j$:
\begin{align}
a_j\leq (1+\frac{1}{2^{k-2}})a_{j-1}+\frac{1}{2^{k-1}},  \quad  \forall \ 2\leq j\leq m,
\end{align}
so
\begin{align}\label{31}
a_m+\frac{1}{2}\leq (1+\frac{1}{2^{k-2}})(a_{m-1}+\frac{1}{2})\leq... &\leq(1+\frac{1}{2^{k-2}})^{m-1}(a_{1}+\frac{1}{2}) \nonumber\\
&\leq (\frac{1}{2^{k-1}}+\frac{1}{2})(1+\frac{1}{2^{k-2}})^{m-1}.
\end{align}

If $\frac{m}{2^k}=\tn{negl}(\lambda)$, then the $H$-distance between $P_m|\psi\rangle$ and $P'_m|\psi\rangle$ can be bounded by
\begin{align}
\hspace{-0.3cm}||P_m|\psi\rangle-&P'_m|\psi\rangle ||_{H}=a_m\leq (\frac{1}{2^{k-1}}+\frac{1}{2})(1+\frac{1}{2^{k-2}})^{2^{k-2}\frac{m-1}{2^{k-2}}}-\frac{1}{2}  \quad \quad  \quad     (k\rightarrow \infty)\\
&\leq (\frac{1}{2^{k-1}}+\frac{1}{2}) 3^{\frac{m-1}{2^{k-2}}}-\frac{1}{2}=(\frac{1}{2^{k-1}}+\frac{1}{2}) \big{(}1+\tn{ln3}\frac{m-1}{2^{k-2}}+o(\frac{m-1}{2^{k-2}}) \big{)}-\frac{1}{2} =\tn{negl}(\lambda),\label{cx51}
\end{align}
where the last inequality is by $(1+\frac{1}{n})^{n}\xrightarrow{n\rightarrow\infty} e < 3$ and the Taylor expansion: $3^{x}=e^{x\tn{ln}3}=1+x\tn{ln}3+o(x)$.$\hfill\blacksquare$
\end{proof}

\begin{prop}\label{2376}
Let $m(\lambda)$ be a polynomial function in $\lambda$, and let $k(\lambda)$ be a function in $\lambda$ such that $\frac{m(\lambda)}{2^{k(\lambda)}}=\textnormal{negl}(\lambda)$. Then, the new \tn{QHE} scheme with precision parameter $k=k(\lambda)$ allows to evaluate $m=m(\lambda)$ number of single-qubit gates while guaranteeing the correctness of the decryption of the evaluation result to be within negl$(\lambda)$ trace distance.
\end{prop}

%and the unitary approximation of $V'_j$ as defined in Seciton \ref{sec3.1}, denoted by $V''_j$ for $1\leq j\leq m$, satisfy
%\begin{equation}
%\| V_m...V_1|\psi\rangle - V''_m...V''_1|\psi\rangle \|_{tr}\leq \textnormal{negl}(\lambda),
%\end{equation}

\begin{proof}
In the notations in the proof of Lemma \ref{234}, let $V_j$ be the unitary operator for $1 \leq j\leq m$ that realizes the $2\times 2$ unitary $U_j$ on some multi-qubit system $\bk{\psi}$, let $V'_j$ be the linear operator that realizes $U'_j$, the $k$-bit finite precision quaternion representations of $U_j$ (see Lemma \ref{234}) for $1\leq j\leq m$, and let $V''_j$ be the unitary approximation (see Section \ref{sec3.1}) of $U'_j$ for $1\leq j\leq m$. Define $P=V_m....V_1$, $P'=V'_m...V'_1$ and $P''=V''_m...V''_1$. We need to prove
\begin{align}
||P''|\psi\rangle-P|\psi\rangle||_{tr}\leq  \textnormal{negl}(\lambda),
\end{align}
for a general multi-qubit state $\bk{\psi}$.

Firstly, group together those $V_j$ that act on the same qubit for $1\leq j \leq m$. Assume that the result consists of $m_1$ ($m_1\leq m$) unitary operators: $ \widetilde{V}_j$, $1\leq j\leq m_1$, each of which acts on a different qubit. Since linear operators that act on different qubits are commutative, we can rewrite $P=\prod^{{m_1}}_{j=1}\widetilde{V}_j$. By grouping the operators acting on the same qubit, we can define $\tV'_j$ and $\tV''_j$ similarly, such that $P'=\prod^{m_1}_{j=1}\tV'_j$ and $P''=\prod^{{m_1}}_{j=1}\tV''_j$.

Now that each $\widetilde{V_j}$ is composed of serval unitary operators, with $\widetilde{V'_j}$ being the approximation of $V_j$. By (\ref{a12}), it holds that $\nrm{\tV'_j-\tV_j} = \negl{\lambda}$, and then
\begin{align}\label{zq1}
\big{|} \nrm{\tV'_j}-1\big{|} =\big{|}  \nrm{\tV'_j}-\nrm{\tV_j}\big{|} \leq \nrm{\tV'_j-\tV_j} = \negl{\lambda},   \qquad \forall \  1\leq j\leq m.
\end{align}
Similarly, $\nrm{\tV'_j-\tV''_j}= \negl{\lambda}$ for $1\leq j\leq m_1$. By making induction similar to (\ref{31}), we can deduce that $\nrm{P'-P''}=\negl{\lambda}$ (notice that $m_1\leq m$). It follows from (\ref{a12}) that $\nrm{P-P'}=\nrm{V_m....V_1- V'_m....V'_1}=\negl{\lambda}$. Therefore,
\begin{align}
||P''-P||_2\leq ||P''-P'||_2+||P'-P||_2 =   \text{negl}(\lambda).
\end{align}
According to (\ref{42}), for any quantum state $|\psi\rangle$, $||P''|\psi\rangle-P|\psi\rangle||_{tr} = \textnormal{negl}(\lambda).$ $\hfill\blacksquare$
\end{proof}

%By Theorem 5.2 (Lemma \ref{41}) in \cite{mahadev2018classical}), the HE allows to perform any $\rm{poly}(\lambda)$ number of encrypted CNOT operation to obtain a ciphertext that can be decrypted correctly within a negl$(\lambda)$ trace distance. It follows that

%For a $\lambda$-qubit state $|\psi\rangle$, applying QOTP$(k)$ scheme to obtain the encryption $U_{\vec{t}}|\psi\rangle$, and the gate key $\vec{t}=(t_1,t_2,t_3,t_4)$ is encrypted in the form $\textnormal{HE.Enc}_{sk_3L}(\vec{t})$.

The security of our scheme is by combining the security of QOTP (see Lemma \ref{2.5}) and the security of Mahadev's HE scheme (see Theorem 6.1 of \cite{mahadev2018classical}).

\vspace{0.3cm}

\hspace {-0.7cm} \textbf{Efficiency Comparison.}

\vspace{0.2cm}

We make an efficiency comparison between our QFHE scheme and the QFHE scheme of \cite{mahadev2018classical} for the task of evaluating the quantum circuits. Overall, for evaluating a general circuit composed of $p$ percentage of CNOTs and (1-$p$) percentage of 1-qubits gates within the precision negl$(\lambda)$, the quantum complexity advantage of our scheme over the scheme of \cite{mahadev2018classical} is
\begin{align}
O(\frac{  (1-p)\lambda^2}{ p \lambda  })=O(\lambda),
\end{align}
when constant $p$ is away from both one and zero. The detailed analysis is as follows.

First, let the quantum complexity of encrypted-CNOT operation of \cite{mahadev2018classical} be $T_Q$, which is roughly equal\footnote{This can be seen by comparing the Algorithm \ref{20} and Mahadev's encrypted CNOT operation, cf. Claim 4.3 of \cite{mahadev2018classical}.} to that of Algorithm \ref{20}, so that it is the basis for comparison.

In the scheme of \cite{mahadev2018classical}, to evaluate a 1-qubit gate and obtain a state within O$(\frac{1}{2^\lambda})$ trace distance from the correct result, by SK algorithm, the number of Hadamard/Toffoli gates required to be evaluated is O($\lambda^2$). Since evaluating a Toffoli gate requires a constant number of encrypted-CNOT operations \cite{mahadev2018classical}, in the worst case, the number of required Toffoli gates is O$(\lambda^{2})$, and the quantum complexity of evaluating 1-qubit gate within precision O$(\frac{1}{2^\lambda})$ (in trace distance) by the scheme of \cite{mahadev2018classical} is O$(\lambda^{2})T_Q$; the quantum complexity of evaluating a CNOT gate is O$(1)$.

In our scheme, to evaluate 1-qubit gate within the precision negl$(\lambda)$, we set the number of bits used to present the gate key to be $k=\lambda$, and then the complexity is totally classical and is O$(\lambda) T_C$, where $T_C$ is the complexity required for homomorphic evaluations on each bit; evaluating a CNOT gate requires O($\lambda$) uses of Algorithm \ref{20}, so the quantum complexity is O($\lambda$)$T_Q$.

For circuits composed of $p$ percentage of CNOTs and (1-$p$) percentage of 1-qubits, the quantum complexity of the QFHE scheme of \cite{mahadev2018classical} is
\begin{align}\label{cl1}
O(\frac{(1-p) \lambda^2 T_Q + p }{  p \lambda T_Q})=O(\lambda)
\end{align}
times of that by our QFHE scheme, when constant $p\neq0, 1$.

Also, if $T_Q$, $T_C$ and $p$ $(\neq 0, 1)$ are considered as constants, then our scheme has the overall complexity advantage
\begin{align}\label{cl1}
O(\frac{(1-p) \lambda^2 T_Q + p }{(1-p)\lambda T_C+ p \lambda T_Q})=O(\lambda),
\end{align}

According to the above arguments, in the worst case where there are overwhelmingly many CNOTs and negligible 1-qubits gates, our method is less efficient than previous QFHE schemes such as \cite{mahadev2018classical}. In the general case, however, our scheme is polynomially better asymptotically.

For some typical quantum circuits, like quantum Fourier transform (QFT) (cf. Figure 1 and Figure 2), the numbers of CNOTs and 1-qubits
are roughly equal, and thus the percentage $p=1/2$. This is the case of a tie, with no bias towards any one, showing that our scheme has advantage over the previous schemes in general.

There are two worthwhile points about the above comparison:

(1) Compared to the QFHE scheme of \cite{mahadev2018classical} combined with the specific SK algorithm of approximation parameter $c=2$, the advantage of our scheme is $O(\lambda)$, significant. Moreover, the lowest bound for approximation parameter is $c = 1$ (see (23) in [DN05]). So, our method reaches the best complexity that can be achieved by \cite{mahadev2018classical} together with any approximation algorithm. To our best knowledge, no method in the literature has ever achieved this complexity.

(2) For particular quantum circuits (such as approximate-QFT \cite{nam2020approximate}), there may exist some direct ``Clifford+non-Clifford'' implementation that makes the evaluation by QFHE scheme of \cite{mahadev2018classical} more efficient. However, for general quantum circuits (usually designed by 1-qubit/CNOT gates), finding their efficient ``Clifford+non-Clifford'' implementation is as (if not more) hard as redesigning the algorithm. So, using previous QFHE schemes to evaluate a circuit is generally done in two steps: first decompose each 1-qubit gate into Clifford gates and T gates, then evaluate them one by one.

\begin{figure}[H]
\includegraphics[width=\textwidth]{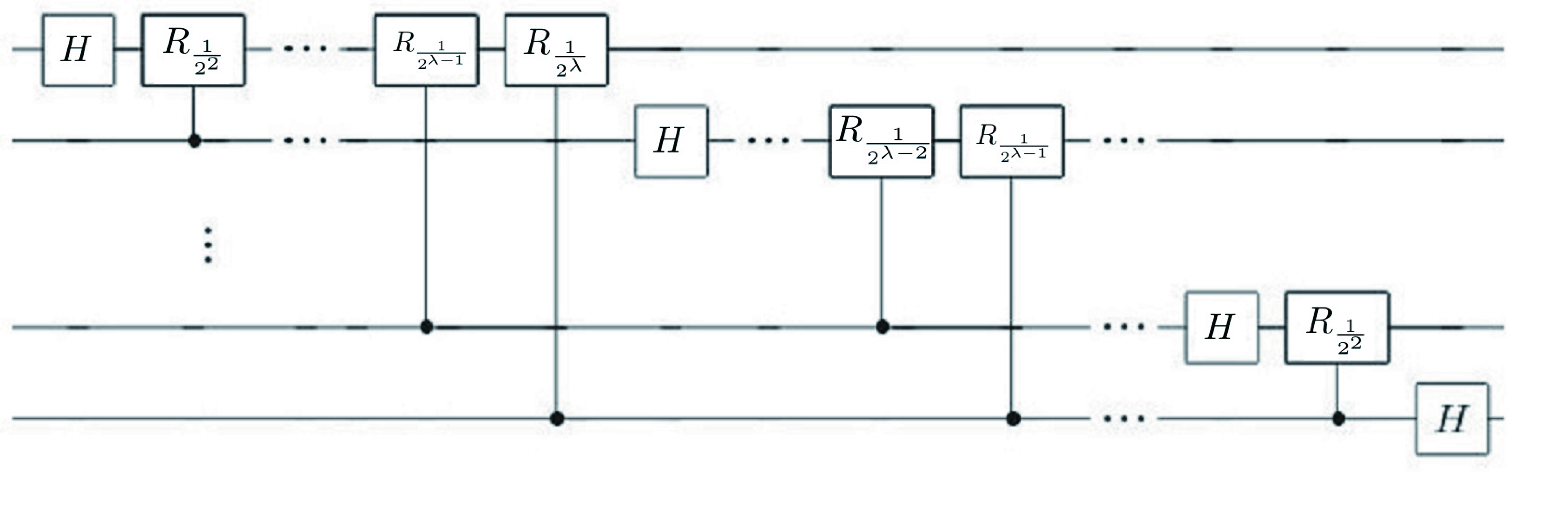}
\caption{Efficient quantum circuit for the quantum Fourier transform on $\lambda$ qubit system, where each line represents a 1-qubit, $H$ denotes the Hadamard gate, and the rotation $R_{\alpha}=\protect\begin{bmatrix}1 &  \protect\\  &  e^{2\mi\pi\alpha} \protect\end{bmatrix}$ for $\alpha\in[0,1)$.} \label{fig1}
\end{figure}

\begin{figure}[H]
\includegraphics[width=\textwidth]{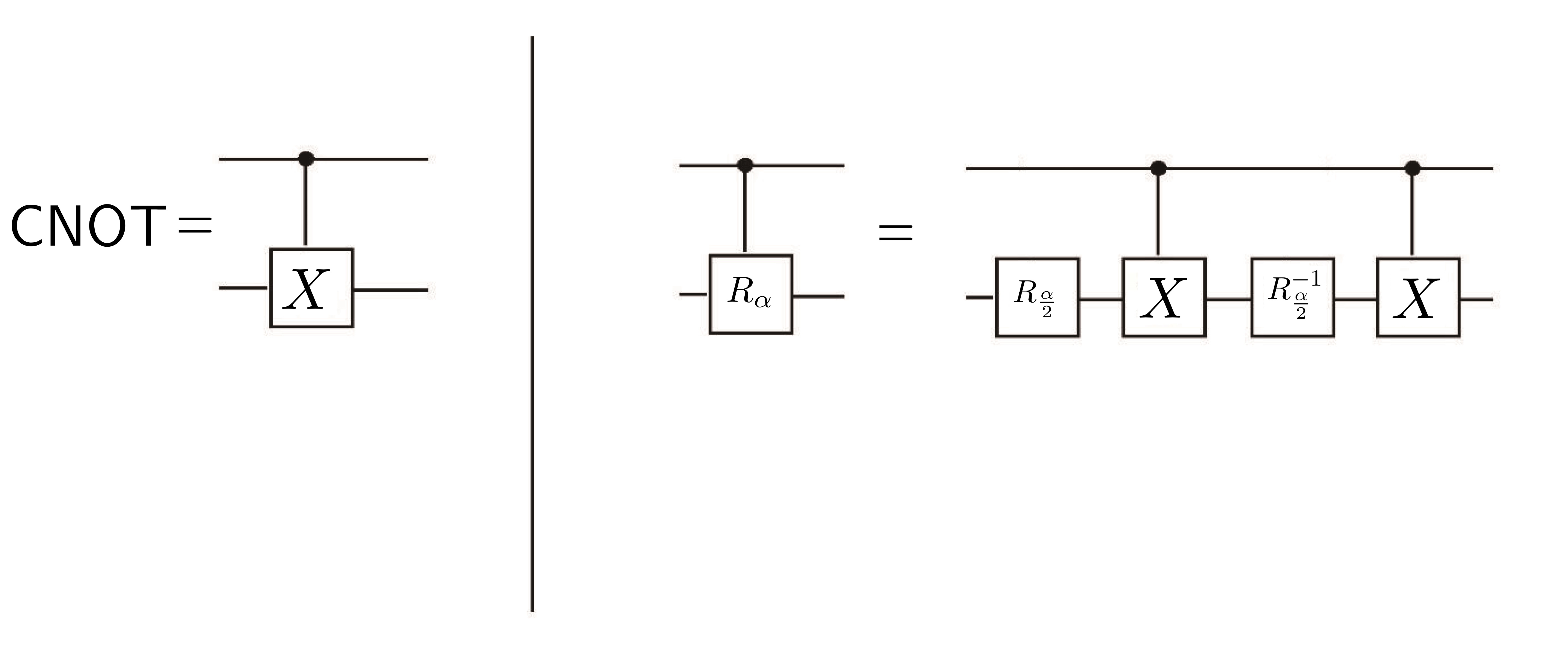}
\caption{(Left) the CNOT gate, where $X$ is the Pauli-$X$ matrix. (Right) The circuit of using 1-qubit gates and the CNOT gate to implement the conditional rotation $\bar{R}_{\alpha}$ for $\alpha\in[0,1)$.}\label{fig2}
\end{figure}

There are also some follows-up of QFHE in \cite{mahadev2018classical}. The rate-1 QFHE in \cite{chardouvelis2021rate} mainly focuses on reducing communication complexity. It seems to have no obvious advantage over the QFHE of \cite{mahadev2018classical} in the efficiency of evaluating quantum circuits. The QFHE scheme in \cite{brakerski2018quantum} achieves encrypted-CNOT by a different approach from \cite{mahadev2018classical}. Due to the lack of a basis for comparison, it is hard to make a fair comparison between our QFHE scheme and Brakerski's QFHE.

It is interesting to investigate the generalization of the encrypted-CNOT in \cite{brakerski2018quantum}, as what we have done for \cite{mahadev2018classical}. This will lead to a more efficient QFHE, because the encrypted-CNOT in \cite{brakerski2018quantum} is more efficient, due to the underlying classical FHE with just polynomial modulus.

\section{Acknowledgments}
The authors would like to thank anonymous referees for their insightful comments and helpful suggestions.

 \bibliographystyle{quantum}
 \bibliography{mybibliography}

\section{Appendix}
\subsection{A proof of Lemma \ref{61}.}\label{vcb1c}

\begin{proof}
 The main idea is to prove that when $\vec{\omega}$ is sampled from $\widetilde{D}_{\mathbb{Z}^{m+1}_q}$, $\rho_{0}(\vec{\omega})$ and $\rho_{1}(\vec{\omega})$ are with overwhelming probability so close to each other that the ratio $\frac{\rho_{0}(\vec{\omega} )}{\rho_{1}(\vec{\omega})}$ is $\lambda$-negligibly close to 1, so that the normalized form of $(\sqrt{\rho_{0}(\vec{\omega})}c_0|0\rangle+\sqrt{\rho_{1}(\vec{\omega})}c_1|1\rangle)$ is within $\lambda$-negligible trace distance to the state $c_0|0\rangle+c_1|1\rangle$. By (\ref{s1.1}), the truncated Gaussian distribution $D_{\mathbb{Z}^{m+1}_{q},\beta_{f}}$ has density function
\begin{equation}\label{nw4}
\rho_{0}(\vec{x})=\frac{e^{-\pi \frac{||\vec{x}||_2^2}{\beta_{f}^2} }}   { \displaystyle\sum_{\vec{x}\in \mathbb{Z}^{m+1}_{q},||\vec{x}||_{\infty} \leq \beta_{f}} e^{-\pi\frac{||\vec{x}||_2^2}{\beta_{f}^2}} }.
\end{equation}
For short, denote distribution $D_{\mathbb{Z}^{m+1}_q,\beta_{f}}$ by $D_0$, denote $\e'+D_{\mathbb{Z}^{m+1}_q,\beta_{f}}$ by $D_1$, and denote $\widetilde{D}_{\mathbb{Z}^{m+1}_q}$ by $\widetilde{D}$.
Let $\rho$ be the density of the distribution $\widetilde{D}$, then $\rho(\x)=p\rho_0(\x)+(1-p)\rho_1(\x),\ \forall \x\in\Z_{q}^{m+1}$.  Obviously, support sp$(D_0)=\{ \x | ||\x ||_{\infty}\leq \beta_{f}, \x\in \Z^{m+1}_{q} \}$, and $\frac{||\e'||_\infty}{\beta_{f}}\leq\frac{1}{(N+1)^{\eta}}= O(\frac{1}{\lambda^{ \Theta (\log \lambda)}})$ is $\lambda$-negligible.
Now, consider the set
$$S=\tn{sp}(D_0)\setminus\tn{sp}(D_1)+\tn{sp}(D_1)\setminus\tn{sp}(D_0).$$
If the vector $\vec{\omega}$ in (\ref{nw1}) is sampled from $S$, then $\bk{c'}$ is equal to a computation basis ($\bk{0}$ or $\bk{1}$) that could be very far away from $|c\rangle$. Fortunately, this happens with negligible probability, as proved below. To show for any $0\leq p\leq 1$,
\begin{align}\label{nn1}
\sum_{\vec{x}\in{S}} \rho (\vec{x}) =   p \sum_{\vec{x} \in \tn{sp}(D_0)\setminus\tn{sp}(D_1)} \rho_0 (\x) + (1-p)\sum_{\vec{x} \in \tn{sp}(D_1)\setminus\tn{sp}(D_0)} \rho_1 (\x) =\tn{negl}(\lambda)
\end{align}
we first prove $\sum\limits_{\vec{x} \in \tn{sp}(D_0)\setminus\tn{sp}(D_1)} \rho_0 (\x)=\tn{negl}(\lambda)$. Notice that the set $K=\{\x \big{|} ||\x||_{\infty}\leq \beta_f- ||\e'||_{\infty}, \x\in \Z^{m+1}_{q} \}\subseteq \tn{sp}(D_0)\cap \tn{sp}(D_1)$, and so
\begin{align}\label{nn2}
\sum_{\vec{x}\in \tn{sp}(D_0)\setminus\tn{sp}(D_1)}  \rho_0 (\vec{x})= \sum_{\vec{x}\in \tn{sp}(D_0)} \rho_0 (\vec{x})- \sum_{ \vec{x}\in \tn{sp}(D_0)\cap \tn{sp}(D_1) } \rho_0 (\vec{x})\leq 1- \sum_{x\in K} \rho_0 (\vec{x}).
\end{align}

By the shape of $\rho_0$, for any $\x\in K$, $\vec{y}\in \tn{sp}(D_0)\setminus K$, it holds that $\rho_{0}(\x)>\rho_{0}(\vec{y})$. Therefore,
\begin{align}\label{nn3}
 \sum_{x\in K} \rho_0 (\vec{x})> \frac{|K|}{|\tn{sp}(D_0)|} \left(\sum\limits_{x\in K} \rho_0 (\vec{x})+\sum\limits_{x\in \tn{sp}(D_0)\setminus K} \rho_0 (\vec{x}) \right)&=\left(\frac{\beta_f-||\e'||_{\infty}}{\beta_f}\right)^{m+1}\nonumber\\
 &\geq1-\frac{(m+1)||\e'||_{\infty}}{\beta_f},
\end{align}
where the last inequality is by $(1-a)^{b}\geq1-a b$ for any $b\geq 1$, $0\leq a\leq 1$. Combining (\ref{nn2}), (\ref{nn3}), $m$=poly($\lambda$) and $\frac{||\e'||_\infty}{\beta_{f}}= \tn{negl}(\lambda)$ gives
\begin{align}\label{yn1}
\sum\limits_{\vec{x}\in \tn{sp}(D_0)\setminus\tn{sp}(D_1)}  \rho_0 (\vec{x})\leq \frac{(m+1)||\e'||_{\infty}}{\beta_f} =\tn{negl}(\lambda).
\end{align}

Similarly,
\begin{align}\label{yn2}
\sum\limits_{\vec{x}\in \tn{sp}(D_1)\setminus\tn{sp}(D_0)}  \rho_1 (\vec{x})= \tn{negl}(\lambda).
\end{align}
and then (\ref{nn1}) follows. Now, define the ball $G:=\{ \vec{x}|\ ||\vec{x}||_2 \leq \beta_{f}\sqrt{m+1} \}\supseteq \tn{sp}(D_0)$. Now that $\rho(\vec{x})$ is supported on $\tn{sp}(D_0)\bigcup S \subseteq G\bigcup S$, by (\ref{nn1}),
$$1=\sum\limits_{\vec{x}\in G\bigcup S}\rho(\vec{x})= \sum\limits_{\vec{x}\in G\setminus S}\rho(\vec{x})+\textnormal{negl}(\lambda).$$
To complete the proof, it suffices to show that $\| |c'\rangle-|c\rangle \|_{H}= \textnormal{negl}(\lambda)$ for any $\vec{\omega}\in G\setminus S$. Indeed, for any $\vec{\omega}\in G\setminus S$,
%Now, we consider the $x$ sampling from the set $T:=Z^m_q \setminus S$, that is, $\rho_{0}(x) \neq0$ and $\rho_{1}(x) \neq 0$.
%We concern $x$ sampled from out of $S$.
 %\quad \lambda\mapsto\infty

\begin{equation}\label{nw6}
 \frac{\rho_{1}(\vec{\omega})}{\rho_{0}(\vec{\omega})} =\frac{e^{-\pi \frac{||\vec{\omega}+\e'||_2^2}{\beta_{f}^2}}}{e^{-\pi \frac{||\vec{\omega}||_2^2}{\beta_{f}^2} }}= e^{-\pi \frac{ 2\vec{\omega}\cdot\e' +||\e'||_2^2}{\beta_{f}^2} }.
\end{equation}
By $||\e'||_{2}\leq\beta_{init}(N+1)^{\eta_{c}}\sqrt{m+1}$, $||\vec{\omega}||_2\leq\beta_f\sqrt{m+1}$, and $\beta_f=\beta_{init}(N+1)^{\eta_{c}+\eta}$, it holds that
\begin{equation}\label{nw7}
\left| \frac{\vec{\omega}\cdot\e'}{\beta_{f}^2}\right|\leq \frac{ ||\vec{\omega}||_2||\e'||_2}{\beta_{f}^2} \leq \frac{m+1}{(N+1)^{\eta}}=\tn{negl}(\lambda), \qquad \frac{||\e'||_2^2}{\beta_{f}^2}\leq \frac{m+1}{(N+1)^{2\eta}} =\tn{negl}(\lambda).
\end{equation}
Combining (\ref{nw6}), (\ref{nw7}) and the Taylor expansion: $e^{x}=1+x+o(x)$ gives
\begin{equation}\label{7.1}
 \left| \frac{\rho_{1}(\vec{\omega})}{\rho_{0}(\vec{\omega})}-1 \right|=\left|e^{-\pi \frac{ 2\vec{\omega}\cdot\e' +||\e'||_2^2}{\beta_{f}^2} }-1\right|= \left| -\pi \frac{ 2\vec{\omega}\cdot\e' +||\e'||_2^2}{\beta_{f}^2} +o(-\pi \frac{ 2\vec{\omega}\cdot\e' +||\e'||_2^2}{\beta_{f}^2} ) \right|=\tn{negl}(\lambda).
\end{equation}

Similarly,
\begin{equation}\label{7.2}
\ \left| \frac{\rho_{0}(\vec{\omega})}{\rho_{1}(\vec{\omega})}-1 \right| = \textnormal{negl}(\lambda).
\end{equation}

Let $\bk{c}=c_0\bk{0}+c_1\bk{1}$ and $\bk{c'}=c'_0\bk{0}+c'_1\bk{1}$, and let $\Delta=\frac{\rho_1(\vec{\omega})}{\rho_0(\vec{\omega})}$. By (\ref{7.1}), $|\Delta-1|$=negl($\lambda$), and thus
\begin{equation}
|c'_0-c_0|=|\frac{c_0 \sqrt{\rho_{0}(\vec{\omega})}}{\sqrt{\rho_{0}(\vec{\omega})|c_0|^2+\rho_{1}(\vec{\omega})|c_1|^2}}-c_0|=|c_0||\frac{1}{\sqrt{|c_0|^2+\Delta |c_1|^2}}-\frac{1}{\sqrt{|c_0|^2+|c_1|^2}}|= \textnormal{negl}(\lambda).
\end{equation}
Similarly, by (\ref{7.2}), we have $|c'_1-c_1|=$ negl$(\lambda)$. Then, $\| |c\rangle-|c'\rangle \|_{H}= \sqrt{ \frac{1}{2} \sum_{j=0,1} |c_j-c'_j|^2}$ $=\tn{negl}(\lambda)$. By (\ref{42}), $\| |c\rangle- |c'\rangle \|_{tr}=$ negl$(\lambda)$. $\hfill\blacksquare$
\end{proof}

\subsection{Homomorphic approximate computation of Euler angles}\label{7.3}

Using the encrypted gate key Enc$(\vt)$ where $\vt=(t_1,t_2,t_3,t_4)\in\R^4$, one can directly compute the encrypted Euler angles Enc$(\alpha,\beta,\gamma)$ such that $U(\alpha,\beta,\gamma)\xlongequal{\tn{i.g.p.f}}U_{\vt}$, according to the relations (\ref{w2}). However, the this idea is simply too naive. There are two points ignored: 1. the gate key MHE.Enc($\vt'$) at hand is only an encrypted approximation to MHE.Enc($\vt$); 2. efficient implementation of homomorphic evaluation of (\ref{w2}). Below we make more careful consideration of the details.

\begin{itemize}
\item[]
\textbf{Open disc in the complex plane.} $D(z_0,r)=\{ z|\ |z-z_0| < r, z\in\C \}$.
\vspace{0.3cm}

\textbf{Modulo $[\ ]_r$.} For $\alpha\in\R$, $r\in \Z^{+}$, $[\alpha]_r$ is the real number in range [-$\frac{r}{2}$, $\frac{r}{2}$) such that $[\alpha]_r=\alpha \mod r$.
\end{itemize}

In \cite{cheon2019numerical}, Cheon et al. showed how to homomorphically evaluate square root and division efficiently\footnote{Lemma 1 and Lemma 2 in \cite{cheon2019numerical} show that for any $x\in(0,2)$ (or $x\in[0,1]$), $d$ iterations are sufficient for homomorphically computing $1/x$ (or $\sqrt{x}$) to precision $\frac{1}{x}(1-x)^{2^{d+1}}$ (or $\sqrt{x}(1-\frac{x}{4})^{2^{d+1}}$), which guarantees the exponential convergence rate in the number of iteration $d$. By using bootstrapping to refresh the noise in ciphertext after each iteration, the accumulated noise can be bounded, rather than increasing double exponentially in $d$.}. In the following, we investigate how to homomorphically evaluate the logarithmic function ln$(z)$ for $z\in\C$, which is defined to be ln$(z)=\tn{ln}|z|+\mi\tn{Arg}(z)$ where the principal value of the argument $-\pi<$ Arg$(z)\leq\pi$. For any $k\in\N$, one can approximate ln$(z)$ for $z=e^{\mi\theta}$ where $\theta\in[0,\frac{\pi}{2}]$ by a degree-$k$ polynomial to precision negl($k$):

\begin{lem}\label{c1}
For any $k\in\N$, there is a degree-$k$ complex polynomial $P(z)$ that can approximate function \tn{ln}$(z)$ to precision \tn{negl}$(k)$ for any $z\in D(e^{\mi\frac{\pi}{4}},0.9)\supseteq \{e^{\mi \theta}| \theta\in [0,\frac{\pi}{2}]\}$, and satisfies $|P(z_1)-P(z_2)|\leq 10|z_1-z_2|$ for any $z_1,\ z_2 \in D(e^{\mi\frac{\pi}{4}},0.9)$.
\end{lem}
\begin{proof}
By $|e^{\mi\frac{\pi}{4}}-1|\leq0.8$, it is easy to verify that $\{e^{\mi \theta}| \theta\in [0,\frac{\pi}{2}]\}\subseteq D(e^{\mi\frac{\pi}{4}},0.9)=\{ z||z-e^{\mi\frac{\pi}{4}}| < 0.9, z\in\C \}$. Note that $0\notin D(e^{\mi\frac{\pi}{4}},0.95)$. The function ln$(z)$ is a univalent analytic function on the disc $D(e^{\mi\frac{\pi}{4}},0.95)$. The Taylor expansion of ln($z$) for $z\in D(e^{\mi\frac{\pi}{4}},0.9)$ at point $a=e^{\mi\frac{\pi}{4}}$ with $k$ terms and remainder expression (cf. \cite{tca2020}) is:
\begin{align}
\tn{ln}(z)=\sum^{k}_{n=0}\frac{1}{n!} \tn{ln}^{(n)}(a)(z-a)^{n}+R_k(z)=\frac{\pi}{4}\mi+\sum^{k}_{n=1}&\frac{(-1)^{n-1}}{na^{n}}(z-a)^{n}+R_k(z),
\end{align}
where ln$^{(n)}(z)=(-1)^{n-1}(n-1)!z^{-n}$ is the $n$-th derivative of ln$(z)$, and the remainder is
\begin{align}\label{o2w}
R_k(z)=\frac{1}{2\mi\pi}\oint_{\partial D(a,0.95)}\frac{\tn{ln}(w)}{w-z}(\frac{z-a}{w-a})^{k}dw,
\end{align}
where $\partial D$ is the boundary of $D$. Let $M=\max\limits_{c\in \partial D(a,0.95)} \{|\tn{ln}(c)|\}$, then $M=\max\limits_{c\in \partial D(a,0.95)}$\\$\left|\tn{ ln}|c|+\mi\tn{Arg}(c) \right|$ $\leq \sqrt{1.95^2+\pi^2} < 4$. Let $y=\frac{z-a}{w-a}$, then for any $z\in D(a,0.9)$, any $w\in\partial D(a,0.95)$, $|y|\leq \frac{18}{19}$ and $|w-z|\geq0.05$ . Therefore, the remainder (\ref{o2w}) can be bounded by
\begin{align}
|R_k(z)|\leq\frac{1}{2\pi}(2\pi \times 0.95)\frac{M}{|w-z|}y^{k}\leq 76 (\frac{18}{19})^{k},
\end{align}
where $2\pi\times 0.95$ is the length of integral curve $\partial D(a,0.95)$.

As a result, ln$(z)$ can be approximated by $P(z)=\frac{\pi}{4}\mi+\sum^{k}_{n=1}\frac{(-1)^{n-1}}{n a^n}(z-a)^{n}$, to precision \textnormal{negl}($k$) for $z\in D(a ,0.9)$, where $a=e^{\mi\frac{\pi}{4}}$. For any $z_1$, $z_2 \in D(a,0.9)$,
\begin{align}\label{4.20}
&|P(z_1)-P(z_2)|\leq\sum^{k}_{n=1}|\frac{1}{na^{n}}|\ |(z_1-a)^{n}-(z_2-a)^{n}|  \nonumber\\
&= \sum^{k}_{n=1} \frac{1}{n}|z_1-z_2||(z_1-a)^{n-1}+(z_1-a)^{n-2}(z_2-a)+...+(z_2-a)^{n-1}|\nonumber\\
&< \sum^{k}_{n=1} 0.9^{n-1}|z_1-z_2|< 10 |z_1-z_2|.
\end{align} $\hfill\blacksquare$
\end{proof}

As a corollary of Lemma \ref{c1}, one can homomorphically evaluate the Arg($z$) for $z\in \C$ when $|z|-1$ is small.
\begin{cor}\label{c231}
Let $D_s=D(e^{\mi\pi(\frac{s}{2}-\frac{1}{4})},0.9)$ for $s=1,2,3,4$. Given two bitwise encrypted binary fractions MHE.Enc\\$(a, b)$, where $a,\ b\in[-1,1]$ such that $z=a+b\mi\in \displaystyle\cup^{4}_{s=1}D_s$, one can efficiently prepare a ciphertext MHE.Enc$(d)$, where $d\in[-1,1)$ such that $[d-\frac{\tn{Arg}(z)}{\pi}]_2=\tn{negl}(k)$, i.e., $d$ is $k$-negligibly close to $\frac{\tn{Arg}(z)}{\pi}$ in the ring $\R/2\Z$.
\end{cor}
\begin{proof}
The main idea is to use the sign bits of $a$, $b$ to determine the disc $D_s$ on which to perform the Taylor expansion. In the binary fraction representation, the sign bit of a positive number or zero is $0$, and the sign bit of a negative number is $1$. Let $\delta_{a}$ denote the sign bit of binary fraction $a$. It can be verified that if choosing $l_{a,b}=\delta_{a}-\delta_{b}-2\delta_{a}\delta_{b}+1$ mod 4 then $a+b \mi \in D_{l_{a,b}}$. By a proof similar to that of Lemma \ref{c1}, the following Taylor expansion satisfies $|P_{l_{a,b}}(z)-|z|-\mi\tn{Arg}(z)|=\tn{negl}(k)$ for $z\in D_{l_{a,b}}$:\footnote{Notice that if $a+b\mi=-1$, then $l_{a,b}=2$. This is consistent with that Arg($-1$)=$\pi$ can be approximated by the imaginary part of $P_{2}(-1)$.}
\begin{align}\label{cxz1}
P_{l_{a,b}}(z)=\mi \theta_{l_{a,b}} +\sum^{k}_{n=1}\frac{(-1)^{n-1}}{ne^{n\mi \theta_{l_{a,b}}}}(z-e^{\mi \theta_{l_{a,b}}})^{n},
\end{align}
where for $l_{a,b}=1,2,3,4$, $\theta_1=\frac{\pi}{4}$, $\theta_2=\frac{3\pi}{4}$, $\theta_3=\frac{-3\pi}{4}$, $\theta_4=\frac{-\pi}{4}$, respectively. For any $z=z_1+z_2\mi\in D_{l_{a,b}}$, the degree-$k$ real polynomial
\begin{align}
P'_{l_{a,b}} (z_1,z_2)=\frac{1}{\pi}Im(P_{l_{a,b}}(z_1+\mi z_2))
\end{align}
satisfies
\begin{align}\label{v12}
|P'_{l_{a,b}}(z_1,z_2)-\frac{1}{\pi}\tn{Arg}(z_1+z_2\mi)|=\tn{negl}(k).
\end{align}

Now, with MHE.Enc($a,b$) at hand, by homomorphic arithmetics on the encrypted sign bits of $a$, $b$, one first produces MHE.Enc($l_{a,b}$). Then by homomorphic evaluation of the following conditional function in $a$, $b$:
\begin{align}\label{cx5}
f(a, b)=\sum^{4}_{j=1}\Delta_{j,l_{a,b}}P'_j(a,b), \quad \quad \tn{where  if   } j=l_{a,b} \tn{ then } \Delta_{j,l_{a,b}}=1, \tn{otherwise   } \Delta_{j,l_{a,b} }=0,
\end{align}
one can obtain MHE.Enc($P'_{l_{a,b}}(a,b)$). Next, one can use MHE.Enc($P'_{l_{a,b}}(a,b)$) to prepare \\
MHE.Enc($[P'_{l_{a,b}}(a,b)]_2$). This can be easily done in bit-wise encryption scheme. Then the corollary holds by setting $d=[P'_{l_{a,b}}(a,b)]_2$, because $[d-\frac{\tn{Arg}(z)}{\pi}]_2=[[P'_{l_{a,b}}(a,b)]_2-\frac{\tn{Arg}(z)}{\pi}]_2=\tn{negl}(k)$. $\hfill\blacksquare$
\end{proof}

With the encryption of the approximate gate key at hand, one can homomorphically prepare the encryption of the approximate Euler angles as follows:

\begin{lem}\label{6zq}
Let $\vt=(t_1,...,t_4)\in\mS^3$. Given MHE.Enc$(\vec{t}')$, where $\vec{t}'\in\R^4$ such that $||\vt-\vt'||_{\infty}=\tn{negl}(k)$, one can efficiently prepare the encrypted approximate Euler angles MHE.Enc$(\alpha_0,\beta_0,\gamma_0)$, where $\alpha_0,\gamma_0\in[0,1), \beta_0\in[0,\frac{1}{2})$ such that $U(\alpha_0,\beta_0,\gamma_0)$ is, after ignoring a global factor, within negl($k$) L$^{\infty}$-distance to $U_{\vt}$ .
\end{lem}

%\hspace{-0.5cm}\emph{Proof $($sketch$)$.} Let $\alpha,\gamma\in[0,1), \beta\in[0,\frac{1}{2})$ be as defined in (\ref{w2}), such that $U(\alpha,\beta,\gamma)\xlongequal{\tn{i.g.p.f}}U_{\vt}$. For short, denote $\tilde{\vt}=(\tilde{t}_1,\tilde{t}_2)=(\sqrt{t^2_1+t^2_3},\sqrt{t^2_2+t^2_4})$, and denote $\tilde{\vt'}=(\tilde{t'}_1,\tilde{t'}_2)=(\sqrt{t'^2_1+t'^2_3}, \sqrt{t'^2_2+t'^2_4})$. Notice that $e^{\mi\pi\beta}=\tilde{t}_1+\mi\tilde{t}_2$, and $\pi\beta\in[0,\frac{\pi}{2})$. By Lemma \ref{c1}, on input MHE.Enc$(\tilde{\vt})$, one can produce the MHE encryption of a negl$(k)$-approximation of $\beta$. Now that with the approximation MHE.Enc$(\vt')$ at hand, where $||\vt'-\vt||_{\infty}=\tn{negl}(k)$, one first homomorphically computes the negl$(k)$-approximation of $\tilde{\vt'}$ by homomorphic square root \cite{cheon2019numerical}, followed by using Lemma \ref{c1} to homomorphically evaluate Arg to get a ciphertext MHE.Enc$(\beta_0)$. By the Lipschitz continuity shown in Lemma \ref{c1}, $\beta_0$ is a negl$(k)$-approximation of $\beta$.

%Similarly, by combining $(\ref{w2})$, Corollary \ref{c231}, homomorphic square root and inverse \cite{cheon2019numerical}, using MHE.Enc$(\vt')$ allows to homomorphically compute $\alpha_0$, $\gamma_0$, which is the approximation of $\alpha$, $\gamma$ respectively, such that $U(\alpha_0,\beta_0,\gamma_0)$ is, after ignoring a global factor, within negl($k$) L$^{\infty}$-distance from $U_{\vt}$. $\hfill\blacksquare$\\
%Add the above Proof (sketch), and move the following to Appendix?

\emph{(Sketch Proof).}
Let $\alpha,\gamma\in[0,1), \beta\in[0,\frac{1}{2})$ be as defined in (\ref{w2}), such that $U(\alpha,\beta,\gamma)$\\$\xlongequal{\tn{i.g.p.f}}U_{\vt}$. By (\ref{w2}), $\sqrt{t^2_1+t^2_3}+\mi \sqrt{t^2_2+t^2_4}=e^{\mi\pi\beta}\in D(e^{\mi\frac{\pi}{4}},0.8)$. By homomorphically computing $P(\sqrt{t^2_1+t^2_3}+\mi \sqrt{t^2_2+t^2_4})$ where $P$ is as in Lemma \ref{c1}, one can produce an encrypted negl$(k)$-approximation of $\beta$. When given MHE.Enc$(\vt')$ where $||\vt-\vt'||_{\infty}=\tn{negl}(k)$ so that $|\sqrt{t^2_1+t^2_3}-\sqrt{t'^2_1+t'^2_3}|=\tn{negl}(k)$, by the Lipschitz continuity of $P$, one can still get an encrypted negl$(k)$-approximation of $\beta$. Since the following Lipschitz continuity of $P_j$ in (\ref{cxz1}) holds for $j=1,2,3,4$:
\begin{align}
|P_{j}(z_1)-P_{j}(z_2)|\leq 10|z_1-z_2|, \quad \quad  \quad  \forall z_1, z_2 \in D(e^{\mi\theta_j},0.9),
\end{align}
now with MHE.Enc$(\vt')$, by Corollary \ref{c231}, one can homomorphically compute encrypted approximations to the Euler angles $\gamma$, $\alpha$ according to (\ref{w2}). The complete proof is as following:

\begin{proof}
Let $\alpha,\gamma\in[0,1), \beta\in[0,\frac{1}{2})$ be as defined in (\ref{w2}), such that $U(\alpha,\beta,\gamma)\xlongequal{\tn{i.g.p.f}}U_{\vt}$. Denote
\begin{align}
\widetilde{\vt}=(\widetilde{t}_1,\widetilde{t}_2)=(\sqrt{t^2_1+t^2_3},\sqrt{t^2_2+t^2_4}), \quad \widetilde{\vt'}=(\widetilde{t'_1},\widetilde{t'_2})=(\sqrt{t'^2_1+t'^2_3}, \sqrt{t'^2_2+t'^2_4}).
\end{align}
We begin with computing MHE.Enc$(\beta_0)$, the encryption of a negl$(k)$-approximation of $\beta$. Notice that $\pi\beta\in[0,\frac{\pi}{2})$ and $\widetilde{t}_1+\mi\widetilde{t}_2=e^{\mi\pi\beta}\in D(e^{\mi\frac{\pi}{4}},0.9)$. By Lemma \ref{c1}, there is a degree-$k$ complex polynomial such that $|P(\widetilde{t}_1+\mi\widetilde{t}_2)-\mi\pi\beta|=\tn{negl}(k)$ as follows:
\begin{align}\label{cxq}
P(z)=\frac{\pi}{4}\mi+\sum^{k}_{n=1}\frac{(-1)^{n-1}}{na^{n}}(z-a)^{n},  \qquad \tn{      where $a=e^{\mi\frac{\pi}{4}}$}, \quad z \in D(a,0.9).
\end{align}
The degree-$k$ real polynomial $P_I(t_1,t_2)=\frac{1}{\pi}Im\left( P(t_1+\mi t_2) \right)$ satisfies $|P_I(\widetilde{t}_1,\widetilde{t}_2)-\beta|=$negl($k$).

We prove that for any $\vh=(h_1,h_2)\in\R^2$ such that $||\vh-\widetilde{\vt}||_{\infty}$=negl($k$), it holds that $|P_{I}(\vh)-\beta|=$negl($k$). By setting $z_1=\widetilde{t}_1+\mi\widetilde{t}_2$, $z_2=h_1+\mi h_2$, we first show that $|P(z_1)-P(z_2)|=\tn{negl}(k)$. By $|z_1-a|=|e^{\mi\pi\beta}-a|\leq0.8$ and $||\vh-\widetilde{\vt}||_{\infty}$=negl($k$), there must exist $K_0\in \N$ such that for all $k\geq K_0$, $|z_2-a|\leq|z_2-z_1|+|z_1-a|=\tn{negl}(k)+0.8<0.9$. By Lemma \ref{c1}, for any $k>K_0$, $|P(z_1)-P(z_2)|=\tn{negl}(k)$. Then,
\begin{align}\label{we1}
|P_{I}(\vh)-\beta|&\leq |P_{I}(\vh)-P_{I}(\widetilde{\vt})|+|P_{I}(\widetilde{\vt})-\beta| \nonumber\\
&\leq \frac{1}{\pi}|P(h_1+\mi h_2)-P(\widetilde{t}_1+\mi\widetilde{t}_2)|+\tn{negl}(k)=\tn{negl}(k).
\end{align}

Now with MHE.Enc($\vec{t'}$), by homomorphic square root computation (cf. \cite{cheon2019numerical}), one can efficiently prepare an encrypted approximation of $\vec{\widetilde{t'}}$, denoted by MHE.Enc($\vec{\widetilde{T'}}$), such that $||\vec{\widetilde{T'}}-\vec{\widetilde{t'}}||_{\infty}=$negl($k$), and so
\begin{align}
||\vec{\widetilde{T'}}-\vec{\widetilde{t}}||_{\infty}\leq||\vec{\widetilde{T'}}-\vec{\widetilde{t'}}||_{\infty}+||\vec{\widetilde{t'}}-\vec{\widetilde{t}}||_{\infty}=\tn{negl}(k).
\end{align}
Setting $\vh=\vec{\widetilde{T'}}$, by (\ref{we1}), it holds that
$|P_{I}(\vec{\widetilde{T'}})-\beta|=\tn{negl}(k)$. Set $\beta_0=P_{I}(\vec{\widetilde{T'}})$. Then $|\beta_0-\beta|=\tn{negl}(k)$.

%\begin{align}\label{cx1}
%|\beta_0-\beta|\leq|P_{I}(\vec{\widetilde{T'}})-P_{I}(\widetilde{\vt})|+|P_{I}(\widetilde{\vt})-\beta|= \tn{negl}(k),
%\end{align}
%where last equation comes from (\ref{zwe2})-(\ref{we2}) with substituting $\widetilde{\vt'}=\vec{T}'$

Next, we compute MHE.Enc($\alpha_0,\gamma_0$). Assume initially that $t_1+t_3\mi$, $t_2+t_4\mi$ are not $k$-negligibly close to $0$, so $t^2_1+t^2_3\neq0,\ t^2_2+t^2_4\neq0,\ t'^2_1+t'^2_3\neq0$ and $t'^2_2+t'^2_4\neq0$. By (\ref{w2}),
\begin{align}\label{cx1}
e^{2\pi\mi\alpha}=\frac{-t_4+t_2\mi}{t_1+t_3\mi}\frac{\sqrt{t^2_1+t^2_3}}{\sqrt{t^2_2+t^2_4}}&=\hat{t}_1+\mi\hat{t}_2,\quad \tn{where} \nonumber\\
& \hat{t}_1=\frac{-t_1t_4+t_2t_3}{\sqrt{t^2_1+t^2_3}\sqrt{t^2_2+t^2_4}}, \quad \hat{t}_2=\frac{-t_1t_2+t_3t_4}{\sqrt{t^2_1+t^2_3}\sqrt{t^2_2+t^2_4}},
\end{align}
\begin{align}\label{cx2}
e^{2\pi\mi\gamma}=\frac{-t_4-t_2\mi}{t_1+t_3\mi}\frac{\sqrt{t^2_1+t^2_3}}{\sqrt{t^2_2+t^2_4}}&=\hat{t}_3+\mi\hat{t}_4, \quad \tn{where}\nonumber\\
 &  \hat{t}_3=\frac{t_1t_4+t_2t_3}{\sqrt{t^2_1+t^2_3}\sqrt{t^2_2+t^2_4}}, \quad \hat{t}_4=\frac{t_1t_2-t_3t_4}{\sqrt{t^2_1+t^2_3}\sqrt{t^2_2+t^2_4}}.
\end{align}

We prove that for any $\vg=(g_1,g_2)\in\R^2$ and $d_g\in[0,1)$ such that $g_1+g_2\mi=e^{-2\pi\mi d_g}$, using MHE.Enc$(\vg)$ allows to prepare MHE.Enc($d'_g$) with $[d'_g-d_g]_1=\tn{negl(k)}$. By Corollary \ref{c231}, on input MHE.Enc($\vg$), one can produce a ciphertext MHE.Enc($d$) where $d\in[-1,1)$, such that $[d-\frac{1}{\pi}\tn{Arg}(e^{2\pi\mi d_g})]_2=$ negl($k$), namely, $[d-2 d_g]_2=$negl($k$), and $[\frac{d}{2}-d_g]_1=$negl($k$). By homomorphic evaluation based on MHE.Enc($d$), one can continue to produce an encryption of a number $d'_g\in[0,1)$ such that $[d'_g]_1=\frac{d}{2}$. It holds that $[d'_g-d_g]_1=[\frac{d}{2}-d_g]_1=\tn{negl}(k)$.

Furthermore, we have the Lipschitz continuity of $P_j$ in (\ref{cxz1}) for $j=1,2,3,4$ as follows:
\begin{align}
|P_{j}(z_1)-P_{j}(z_2)|\leq 10|z_1-z_2|, \quad \quad  \quad  \forall z_1, z_2 \in D(e^{\mi\theta_{j}},0.9).
\end{align}
By combining the estimates in (\ref{we1}) and Corollary \ref{c231}, when given not $\tn{MHE.Enc}(\vg)$, but instead an encrypted approximate MHE.Enc($\vh$) where $||\vh-\vg||_{\infty}=\tn{negl}(k)$, one can produce a ciphertext MHE.Enc$(d''_g)$ such that $[d''_g-d_g]_1=\tn{negl}(k)$. Then by $e^{2\pi\mi d_g}=g_1+\mi g_2$,
\begin{align}\label{cq4}
|e^{2\pi\mi d''_g}-(g_1+\mi g_2)|=\tn{negl}(k).
\end{align}

%More specifically, the convergence rate of $\frac{1}{x}$ in $x\in[p(k),1]$, where $p(k)>\tn{negl}(k)$, is $(1-p(k))^{2^k}=(1-p(x))^{\frac{1}{p(x)}p(k)2^{k}}$

Below, we prove that with MHE.Enc$(\vt')$ at hand, one can prepare MHE.Enc($\alpha_0,\gamma_0$) such that $U(\alpha_0,\beta_0,\gamma_0)$ is, after ignoring a global factor, within negl($k$) L$^{\infty}$-distance to $U_{\vt}$. First, one can prepare MHE.Enc$(\hat{\vec{T'}})$, where $\hat{\vec{T'}}\in\R^4$ is an approximation to
\begin{align}\label{vcx}
\hat{\vt'}&=(\hat{t'_1},\hat{t'_2},\hat{t'_3},\hat{t'_4})\nonumber\\
&=(\frac{-t'_1t'_4+t'_2t'_3}{\sqrt{t'^2_1+t'^2_3}\sqrt{t'^2_2+t'^2_4}},\frac{-t'_1t'_2+t'_3t'_4}{\sqrt{t'^2_1+t'^2_3}\sqrt{t'^2_2+t'^2_4}},\frac{t'_1t'_4+t'_2t'_3}{\sqrt{t'^2_1+t'^2_3}\sqrt{t'^2_2+t'^2_4}},\frac{t'_1t'_2-t'_3t'_4}{\sqrt{t'^2_1+t'^2_3}\sqrt{t'^2_2+t'^2_4}}) \end{align}
such that $||\hat{\vt'}-\hat{\vec{T}'}||_{\infty}=$negl($k$).\footnote{When $ t'^2_1+t'^2_3$ is not negligibly small, the convergence rates of inverse and square root algorithms in \cite{cheon2019numerical} for homomorphically computing $\frac{1}{\sqrt{t'^2_1+t'^2_3}}$ are exponential. By Cauchy-Schwarz inequality, $||\hat{\vt'}||_{\infty}\leq 1$. Then, the exponential convergence rate is sufficient to guarantee that a ciphertext MHE.Enc$(\hat{\vec{T'}})$ with $||\hat{\vec{T}'}-\hat{\vt'}||_{\infty}=$negl($k$) can be prepared in time poly($k$) .} By the argument leading to (\ref{cq4}), if setting $\vg=(\hat{t}'_1,\hat{t}'_2)$ and the approximation $\vh=(\hat{T}'_1,\hat{T}'_2)$, or $\vg=(\hat{t}'_3,\hat{t}'_4)$ and $\vh=(\hat{T}'_3,\hat{T}'_4)$, then using MHE.Enc($\hat{\vec{T'}}$), one can prepare ciphertexts MHE.Enc($\alpha_0$), MHE.Enc($\gamma_0$) such that

\begin{align}\label{cx3}
|e^{2\pi\mi\alpha_{0}}-(\hat{t}'_1+\mi\hat{t}'_2)|=\tn{negl}(k), \qquad \qquad  |e^{2\pi\mi\gamma_{0}}-(\hat{t}'_3+\mi\hat{t}'_4)|=\tn{negl}(k).
\end{align}

From $|\beta_0-\beta|= \tn{negl}(k)$, one gets $|\sin \pi\beta_0-\sqrt{t^2_2+t^2_4}|$=negl($k$) and  $|\cos \pi \beta_0-\sqrt{t^2_1+t^2_3}|$\\=negl($k$), so
\begin{equation}\label{cx4}
|\sin  \pi\beta_0-\sqrt{t'^2_2+t'^2_4}=\tn{negl}(k), \quad  \quad |\cos  \pi\beta_0-\sqrt{t'^2_1+t'^2_3}|=\tn{negl}(k).
\end{equation}
By (\ref{cx1}), (\ref{cx2}) and (\ref{vcx}), we have
\begin{align}\label{cx7}
\hat{t}'_1+\mi\hat{t}'_2=\frac{-t'_4+t'_2\mi}{t'_1+t'_3\mi}\frac{\sqrt{t'^2_1+t'^2_3}}{\sqrt{t'^2_2+t'^2_4}}, \quad \quad \quad  \hat{t}'_3+\mi\hat{t}'_4=\frac{-t'_4-t'_2\mi}{t'_1+t'_3\mi}\frac{\sqrt{t'^2_1+t'^2_3}}{\sqrt{t'^2_2+t'^2_4}} .
\end{align}
By combining (\ref{cx3}) and (\ref{cx4}), (\ref{cx7}),
\begin{align}
&||U(\alpha_0,\beta_0,\gamma_0)-e^{\mi\delta}U_{\vt'}||_{\infty}\nonumber\\
= &\left\| \left[                %左括号
  \begin{array}{cc}   % 该矩阵一共3列，每一列都居中放置
\cos(\pi\beta_0)-\sqrt{t'^2_1+t'^2_3}, & -\sin(\pi\beta_0)e^{2\pi\mi \gamma_0}- (t'_4+t'_2 \mi \frac{\sqrt{t'^2_1+t'^2_3}}{t'_1+t'_3\mi}) \\
\sin(\pi\beta_0)e^{2\pi\mi\alpha_0}-(-t'_4+t'_2 \mi\frac{\sqrt{t'^2_1+t'^2_3}}{t'_1+t'_3\mi} ),  & \cos(\pi\beta_0)e^{2\pi\mi(\alpha_0+\gamma_0)}-(t'_1-t'_3\mi\frac{\sqrt{t'^2_1+t'^2_3}}{t'_1+t'_3\mi})
  \end{array}
\right]  \right\|_{\infty}\nonumber\\
=&\tn{negl}(k),
\end{align}
where $e^{\mi\delta}=\frac{\sqrt{t'^2_1+t'^2_3}}{t'_1+\mi t'_3}$. Then,
\begin{align}
||U(\alpha_0,\beta_0,\gamma_0)-e^{\mi\delta}U_{\vt}||_{\infty}=\tn{negl}(k).
\end{align}

In the above homomorphic calculations of $\alpha_0$, $\gamma_0$, we simply assume that $t^2_1+t^2_3$, $t^2_2+t^2_4$ are all not too small, i.e., $t^2_1+t^2_3\neq\tn{negl}(k),\ t^2_2+t^2_4\neq\tn{negl}(k)$. If $t^2_1+t^2_3=$negl($k$), then $t^2_2+t^2_4=1-$negl($k$), and by (\ref{qo1}),
\begin{align}
\frac{t_4+t_2\mi}{-t_4+t_2\mi}=e^{2\pi\mi(\gamma-\alpha+1/2)}.
\end{align}
Set $\alpha_0=0$ and $\widetilde{\gamma}_0=\frac{1}{2\pi}\tn{Arg}(\frac{t_4+t_2\mi}{-t_4+t_2\mi})+\frac{1}{2} \mod 1$. Since $|t_4+t_2\mi|=|-t_4+t_2\mi|=1-\tn{negl}(k)$ and $||\vt-\vt'||_{\infty}=\tn{negl}(k)$, with MHE.Enc($\vt'$) at hand, by homomorphic division (cf. \cite{cheon2019numerical}), one can efficiently prepare an encrypted negl($k$)-approximation to $\frac{t_4+t_2\mi}{-t_4+t_2\mi}$, and then homomorphically evaluate Arg to produce MHE.Enc($\gamma_0$) such that $|\gamma_0-\widetilde{\gamma}_0|=\tn{negl}(k)$.
By combining $e^{2\pi\mi\widetilde{\gamma}_0}=-\frac{t_4+t_2\mi}{-t_4+t_2\mi}$, $t^2_1+t^2_3=\tn{negl}(k)$ and $|\beta_0-\beta|=\tn{negl}(k)$,
\begin{align}
&||U(0,\beta_0,\widetilde{\gamma_0})-\frac{\sqrt{t^2_2+t^2_4}}{-t_4+t_2\mi}U_{\vt}||_{\infty}=\nonumber\\
&\left\| \left[                %左括号
  \begin{array}{cc}   % 该矩阵一共3列，每一列都居中放置
\tn{negl}(k), & -\sin(\pi\beta_0)e^{2\pi\mi \gamma_0}- \frac{t_4+t_2\mi}{-t_4+t_2\mi}\sqrt{t^2_2+t^2_4} \\
\sin(\pi\beta_0)-\sqrt{t^2_2+t^2_4},  & \tn{negl}(k)
  \end{array}
\right]\right\|_{\infty}
=\tn{negl}(k).
\end{align}
The obtained 3-tuple $(\alpha_0,\beta_0,\gamma_0)$ satisfies the requirement of the lemma:
\begin{align}
&||U(\alpha_0,\beta_0,\gamma_0)-\frac{\sqrt{t^2_2+t^2_4}}{-t_4+t_2\mi}U_{\vt}||_{\infty}=\nonumber\\
&||U(\alpha_0,\beta_0,\gamma_0)- U(0,\beta_0,\widetilde{\gamma_0})||_{\infty} + ||U(0,\beta_0,\widetilde{\gamma_0})-\frac{\sqrt{t^2_2+t^2_4}}{-t_4+t_2\mi}U_{\vt}||_{\infty}
=\tn{negl}(k).
\end{align}
The case of $t^2_2+t^2_4=$negl($k$) is similar. $\hfill\blacksquare$
\end{proof}

% the arithmetic of $P$ is in terms of homomorphic evaluation.

%By ratios, the maximal terms $k_0$ in $T(3/4)$ should satisfy
%\begin{align}
%\hspace{-1cm}\frac{(3/4)^2(2k_0+1)^2}{(2k_0-N+4)(2k_0-N+3)}\leq1\ \rm{and}\ \frac{(3/4)^2(2k_0-1)^2}{(2k_0-N+2)(2k_0-N+1)}\geq1.
%\end{align}
%It yields that $\lfloor k_0/2\rfloor\leq N \leq\lfloor k_0/2\rfloor+5$, and $2(N-5)\leq k_0\leq 2N+1$. Thus,
%\begin{align}
%T(3/4)\leq \frac{[ (2k_0-1)!! ]^2}{(2k_0-N+2)!}(3/4)^{2k_0-N+1}(2\times(2N+1) +\sum^{\infty}_{j=1}(\frac{6}{7})^j)\nonumber\\
%\end{align}

\end{document}